\documentclass[11pt,letterpaper,english]{article}
\usepackage[utf8]{inputenc}			
\usepackage{makeidx}
\usepackage[pdftex]{graphicx}
\usepackage{mathtools,amsfonts}
\usepackage{array}
\usepackage{multirow, hhline}
\usepackage{paralist}

\usepackage[usenames,dvipsnames,table]{xcolor}
\usepackage{wrapfig}
\usepackage[english]{babel}		
\usepackage[margin=1in]{geometry}	
\usepackage{lipsum}			
\usepackage{comment}
\usepackage{chngpage}			
\usepackage{amsfonts,amssymb,amsmath, amsthm}
\usepackage{amsopn}
\usepackage{color}

\usepackage{pgfplots}
\usepackage{pgfgantt}
\usepackage{tikz}
\usetikzlibrary{arrows}
\usetikzlibrary{calc}
\usepackage{algpseudocode}
\usepackage[Algorithm]{algorithm}
\usepackage{subcaption}
\usepackage{tabto}
\usepackage{csquotes} 
\usepackage{multirow}
\usepackage{amsthm}
\usepackage[colorlinks,urlcolor=black,citecolor=black,linkcolor=black,menucolor=black]{hyperref}

\newtheorem{theorem}{Theorem}
\newtheorem*{theorem*}{Theorem}

\newtheorem{remark}[theorem]{Remark}

\newtheorem{lemma}[theorem]{Lemma}
\newtheorem{problem}[theorem]{Problem}
\newtheorem*{problem*}{Problem} 
\newtheorem{fact}[theorem]{Fact}
\newtheorem*{fact*}{Fact}

\newtheorem{observation}[theorem]{Observation}
\newtheorem{claim}[theorem]{Claim}

\newcounter{example}[section]
\newenvironment{example}[1][]{\refstepcounter{example}\par\medskip
\noindent \textbf{Example~\theexample. #1} \rmfamily}{\medskip}

\newcommand{\abs}[1]{| #1 |}

\newcommand{\R}{\mathbb{R}}

\newcommand{\C}{\mathbf{C}}

\newcommand{\n}{\mathbf{n}}
\newcommand{\X}{\mathbf{X}}
\newcommand{\w}{\mathbf{w}}
\newcommand{\M}{\mathbf{M}}
\newcommand{\A}{\overline{a}}
\newcommand{\B}{\mathit{aux}}

\newcommand{\MPB}{{\mathit{MPB}}}
\date{}

\title{Dividing Bads is Harder than Dividing Goods:\\ On the Complexity of Fair and Efficient Division of Chores}
\author{Bhaskar Ray Chaudhury\thanks{MPI for Informatics, Saarland Informatics Campus, Graduate School of Computer Science, Saarbr\"ucken, Germany}\\ \texttt{\small braycha@mpi-inf.mpg.de} \and Jugal Garg\thanks{University of Illinois at Urbana-Champaign. Supported by NSF Grant CCF-1942321 (CAREER)}\\ \texttt{\small jugal@illinois.edu}  \and Peter  McGlaughlin \thanks{University of Illinois at Urbana-Champaign}\\ \texttt{\small mcglghl2@illinois.edu} \and Ruta Mehta\thanks{University of Illinois at Urbana-Champaign. Supported by NSF Grant CCF-1750436 (CAREER)}\\ \texttt{\small rutameht@illinois.edu}}

\begin{document}
\maketitle
\begin{abstract}
We study the chore division problem where a set of agents needs to divide a set of chores (bads) among themselves fairly and efficiently. We assume that agents have linear disutility (cost) functions. Like for the case of goods, \emph{competitive division} is known to be arguably the best mechanism for the bads as well. However, unlike goods, there are multiple competitive divisions with very different disutility value profiles in bads. Although all competitive divisions satisfy the standard notions of fairness and efficiency, some divisions are significantly fairer and efficient than the others. This raises two important natural questions: Does there exist a competitive division in which no agent is assigned a chore that she hugely dislikes? Are there simple sufficient conditions for the existence and polynomial-time algorithms assuming them? 

We investigate both these questions in this paper. We show that the first problem is strongly NP-hard. Further, we derive a simple sufficient condition for the existence, and we show that finding a competitive division is PPAD-hard assuming the condition. These results are in sharp contrast to the case of goods where both problems are strongly polynomial-time solvable. To the best of our knowledge, these are the first hardness results for the chore division problem, and, in fact, for any economic model under \emph{linear} preferences. 
\end{abstract}

\section{Introduction}\label{intro}
We consider the chore division problem where a set of $n$ agents needs to divide $m$ divisible chores (bads) among themselves in a \emph{fair} and \emph{efficient} manner. We assume that agents have linear disutility (cost) functions, i.e., each agent $i$ has a disutility $d(i,j)>0$ for a unit amount of chore $j$, and her disutility for an allocation in which $X_{ij}$ amount of chore $j$ is assigned to her is $\sum_j d(i,j)X_{ij}$. 

A division based on \emph{competitive equilibrium} is arguably the best mechanism for this problem due to its remarkable fairness and efficiency guarantees; we refer to the seminal papers of Bogomolnaia et al.~\cite{BogomolnaiaMSY17,BogomolnaiaMSY19} for an elaborate discussion. Competitive equilibrium is a central solution concept in economics, where prices and allocation are such that demand meets supply when each agent gets her most preferred assignment. Two most ideal economic models to study competitive division of chores are Fisher (fixed incomes) and Exchange. An exchange model is like a barter system, where each agent comes with an initial endowment of chores and exchanges them with others to minimize her disutility function. Fisher is a special case of exchange model where each agent has a fixed proportion of each chore. Competitive Equilibrium with Equal Incomes (CEEI)~\cite{Varian74} is a further special case of Fisher where each agent has the same endowment. 

The vast majority of work in fair division focuses on the case of \emph{disposable} goods that agents enjoy or at least can throw away at no cost. The equilibrium computation problem is well-understood in the case of goods. For the exchange model, even though an equilibrium may not always exist, there is a simple (polynomial-time) necessary and sufficient condition for the existence~\cite{Gale76}, convex programming formulations~\cite{Jain07,DevanurGV16}, and polynomial-time algorithms~\cite{Jain07,Ye07,DuanM15,DuanGM16,GargV19}. On the other hand, for the Fisher model, an equilibrium always exists, and there are further specialized convex programs~\cite{EisenbergG59,Shmyrev09,ColeDGJMVY17} and algorithms~\cite{DevanurPSV08,Orlin10,Vegh12,Vegh14} in addition to the ones for the exchange model. 

However, many situations contain undesirable bads (e.g., house chores and job shifts). Clearly, chores are \emph{nondisposable} and must be allocated. Bogomolnaia et al.~\cite{BogomolnaiaMSY17,BogomolnaiaMSY19} showed that the chore division problem is quite different than the case of goods; namely, there are multiple disconnected sets of equilibria with different disutility value profiles. They also conjectured that the computational problem is likely to be more difficult in this case; however, neither polynomial-time algorithms nor hardness results have been obtained so far. Polynomial-time algorithms are known only for the special case of constantly many agents (or chores)~\cite{BranzeiS19,GargM20}. Although these algorithms are obtained for the Fisher model, they easily extend to the exchange as well. 

Let us now consider an example where two agents need to divide two chores with unit supply each, and agents have the same endowment. Assume that the disutility values are: $d(1,1) = d(2,2) = 1; d(1,2)=d(2,1) = L$, where $L\gg 1$. There are three competitive divisions in this case: 

\begin{equation}\nonumber
X^1=\left[\begin{array}{ll} 1 & 0 \\ 0 & 1\end{array}\right], \ \ \ \ \ 
X^2=\left[\begin{array}{ll} 1 & \tfrac{L-1}{2L} \\ 0 & \tfrac{L+1}{2L}\end{array}\right], \ \ \text{ and } \ \ 
X^3=\left[\begin{array}{ll} \tfrac{L+1}{2L} & 0 \\ \tfrac{L-1}{2L} & 1\end{array}\right] \enspace . 
\end{equation}

The disutility value profiles at these divisions are $(1, 1), (\tfrac{L+1}{2}, \tfrac{L+1}{2L})$, and $(\tfrac{L+1}{2L}, \tfrac{L+1}{2})$ respectively.  Observe that as $L$ approaches infinity, one agent's disutility value approaches infinity in the last two divisions ($X^2$ and $X^3$). Although they all satisfy the standard fairness and efficiency properties, namely \emph{no envy}, \emph{proportionality}, \emph{Pareto efficiency}, and so on, one would still want to use the first division that is clearly significantly fairer and efficient among the three. Hence, a important natural question is: 
\medskip

\noindent{\bf Question 1.} Does there exist a competitive division where no agent $i$ is allocated a chore $j$ that she \emph{hugely dislikes}, i.e., for which $d(i,j)$ is greater than equal to a given threshold $\tau$? 
\medskip

In the above example, for any $\tau>1$ the answer is ``yes'', otherwise ``no''. This leads to another natural question: 
\medskip

\noindent{\bf Question 2.} Are there simple sufficient conditions for the existence? And, can we compute a competitive division in polynomial-time under them? 
\medskip

In this paper, we investigate both these questions for the exchange model. We show that answering the first question is strongly NP-hard, even in the CEEI setting. Further, we derive a simple sufficient condition for the existence, and for them, we show that finding a competitive division is PPAD-hard. These results are in sharp contrast to the case of goods where both problems are strongly polynomial-time solvable.\footnote{While dividing goods, the analogous question is to not allocate a good to any agent who {\em likes it very little}. It can be achieved by setting all the utility values less than or equal to the threshold $\tau$ to zero. We can check existence as well as compute equilibrium in strongly polynomial time~\cite{Gale76,DevanurGV16,GargV19}.} To the best of our knowledge, these are the first hardness results for the chore division problem, and, in fact, for any economic model under \emph{linear} preferences.

\subsection{Model and Notation}\label{prelim}
In this section, we introduce the most general problem of \emph{Chore Division} and then certain special variants of it. In an instance of chore division, we have a set $A =\langle a_1,a_2, \dots, a_n \rangle$ of agents and a set $B = \langle b_1,b_2,\dots, b_m \rangle$ of \emph{divisible} chores. Each agent $a_i$  has some initial endowment of chores. Let $w(a_i,b_j)$ denote $a_i$'s endowment of chore $b_j$. Intuitively, agent $a_i$ wants $w(a_i,b_j)$ units of chores to be done (by herself or other agents in the instance). Also, each agent $a_i$ has a linear disutility function that captures the amount of disutility she has for all possible chores that can be allocated to her: $\sum_{j \in [m]} d(a_i,b_j) \cdot X_{ij}$, where $d(a_i,b_j)$ denote the $a_i$'s disutility for one unit of chore $b_j$ and $X_{ij}$ denote the amount of chore $b_j$ that agent $a_i$ does. For convenience, we also refer to the disutility function $d(\cdot,\cdot)$ and the endowment function $w(\cdot,\cdot)$ as the disutility matrix and endowment matrix respectively.

Given price vector $p = \langle p(b_1), p(b_2), \dots , p(b_m) \rangle \in \mathbb{R}^m_{\geq 0}$, where $p(b_j)$ denotes the price of chore $b_j$ (the amount of money that any agent doing one unit of the chore will be paid), agent $a_i$ needs to earn ($\sum_{j \in [m]} w(a_i,b_j) \cdot p(b_j)$) in order to pay to get her own chores done. She can earn this amount by doing other chores, while minimizing her disutility -- this defines her {\em optimal bundle} (or optimal chore set). 

\begin{itemize}
\item {\em Optimal bundle.} For agent $a_i\in A$, $\langle X_{i1}, X_{i2}, \dots, X_{im} \rangle$ is her optimal bundle if it minimizes $\sum_{j \in [m]}d(a_i,b_j) \cdot X_{ij}$ subject to the constraint $\sum_{j \in [m]} X_{ij} \cdot p(b_j) \geq \sum_{j \in [m]} w(a_i,b_j) \cdot p(b_j)$.
\end{itemize}

It is easy to see that in her optimal bundle agent $a_i$ gets assigned only those chores that minimizes her disutility per dollar (buck) earned. Formally, 
\[
\forall j\in [m],\  \ X_{ij}>0 \ \  \Rightarrow \  \ \tfrac{d(a_i,b_j)}{p(b_j)} \leq \tfrac{d(a_i,b_{j'})}{p(b_{j'})}\ \  \forall j' \in [m].
\]

Price vector $p$ is said to be at {\em competitive equilibrium} (CE) if all chores are completely assigned when every agent gets her optimal bundle, {\em i.e.,} $\sum_{i\in[n]} X_{ij} =\sum_{i\in[n]} w(a_i, b_j),\ \forall j\in [m]$. Our focus in this paper is on {\em efficient} competitive equilibrium where no agent is assigned a chore for which her disutility is greater than or equal to a threshold $\tau$. Formally, this problem can be defined as follows: 

\begin{problem}\textbf{(Chore Division)}\label{CD}

	\noindent\textbf{Given:} A set $A = \langle a_1,a_2, \dots , a_n \rangle$ of agents, set $B = \langle b_1,b_2, \dots , b_m \rangle$ of  bads/chores, disutility function $d:A \times B \rightarrow (0,\infty)$,  an endowment  function $w:A \times B \rightarrow \mathbb{R}_{\geq 0}$, and a threshold $\tau$.\\
	\textbf{Find:} Price vector $p = \langle p(b_1),p(b_2), \dots ,p(b_m) \rangle \in \mathbb{R}^m_{\geq 0}$ and allocation $X \in \mathbb{R}^{n \times m}_{\geq 0}$ are such that,
	\begin{enumerate}
		\item $X_{ij} > 0$ only if $d(a_i,b_j) < \tau$ and $\tfrac{d(a_i,b_j)}{p(b_j)} \leq \tfrac{d(a_i,b_{j'})}{p(b_{j'})}$ for all $j' \in [m]$,   
		\item $\sum_{j \in [m]} X_{ij} \cdot p(b_j) = \sum_{j \in [m]}w(a_i,b_j) \cdot p(b_j)$ for all $i \in [n]$, and 
		\item $\sum_{i \in [n]} X_{ij} = \sum_{i \in [n]} w(a_i,b_j)$ for all $j \in [m]$. 
	\end{enumerate}
   A \emph{competitive equilibrium} is a tuple $\langle p, X \rangle $ satisfying the above constraints. 
\end{problem}

Observe that the equilibrium prices are scale invariant: if $p$ is an equilibrium price vector then so is $\alpha \cdot p$ for any positive scalar $\alpha$. Furthermore, at equilibrium $p(b_j)>0$ for each chore $j$, otherwise no agent would be willing to do it. 
\medskip

\noindent{\bf Approximate Competitive Equilibrium.} For an $\epsilon>0$, a $(1-\epsilon)$-competitive equilibrium is the tuple $\langle p, X \rangle$ that satisfy the first two conditions above, and the third condition approximately: 
\[(1 - \epsilon) \cdot \sum_{i \in [n]} w(a_i,b_j) \leq \sum_{i \in [n]} X_{ij} \leq \frac{1}{(1-\epsilon)} \sum_{i \in [n]} w(a_i,b_j) \mbox{ for all } j \in [m].\]

A competitive equilibrium $\langle p, X \rangle$ has many desirable properties like envy-freeness and Pareto optimality in the chore division with equal earnings~\cite{BogomolnaiaMSY17}. We briefly mention the similar properties that a competitive equilibrium satisfies in the context of the more general problem of chore division.
\begin{itemize}
	\item \emph{Weighted envy freeness:}  An allocation $X$ is said to be \emph{envy free} is for every pair of agents $i$ and $i'$, agent $i$'s disutility for her own bundle $\sum_{j \in [m]} d(a_i,b_j) \cdot X_{ij}$ is not higher than her disutility for $a_{i'}$'s bundle: $\sum_{j \in [m]} d(a_i,b_j) \cdot X_{i'j}$: 
	\begin{align*}
	\sum_{j \in [m]} d(a_i,b_j) \cdot X_{ij} &\leq  \sum_{j \in [m]} d(a_i,b_j) \cdot X_{i'j}\enspace .
	\end{align*}
	Weighted envy freeness is a generalized notion of envy freeness. 
	For any two agents $a_i$ and $a_{i'}$, we have that agent $a_i$'s disutility for its own bundle $\langle X_{i1}, \dots, X_{im} \rangle$, scaled down by the cost of the chores in her initial endowment is at most as high as her disutility for $a_{i'}$'s optimal bundle $\langle X_{i'1}, \dots , X_{i'm} \rangle$ scaled down by the cost of the chores in $a_{i'}$'s initial endowment:
	\begin{align*}
	  \frac{\sum_{j \in [m]}d(a_i,b_j) \cdot X_{ij}}{\sum_{j \in [m]} w(a_i,b_j) \cdot p(j)} \leq \frac{\sum_{j \in [m]}d(a_i,b_j) \cdot X_{i'j}}{\sum_{j \in [m]} w(a_{i'},b_j) \cdot p(j)} \enspace . 
	\end{align*}
	Intuitively, a ``fair share of disutility" of any agent is the total cost of chores she brings (her total initial endowment): higher the cost of the chores she brings, larger is her share of disutility. A competitive equilibrium ensures that every agent $a_i$ feels that she receives her fair share of disutility compared to all other agents. 
	\item \emph{Pareto optimality:} There is no other allocation $Y \in \mathbb{R}^{n \times m}_{\geq 0}$ such that disutility of all the agents is at most as high as in $X$ and at least one agent has strictly less disutility in $Y$, {\em i.e.,} $\sum_{j \in [m]}d(a_i,b_j) \cdot Y_{ij} \leq \sum_{j \in [m]}d(a_i,b_j) \cdot X_{ij}$ for all $i \in [n]$ with at least one strict inequality. 
\end{itemize} 

Analogous to the notion of Fisher model in the scenario where we divide only goods (see~\cite{brainard2005fisher}), we define the notion of \emph{Chore Division with Fixed Earnings}, where every agent has to earn a fixed amount of money by doing chores with minimum disutility to price ratio. While in the chore division (Definition \ref{CD}), the amount an agent needs to earn depends on the prices of chores she owns. It is without loss of generality to assume that one unit of each chore has to be done by scaling the disutility values appropriately. 

\begin{problem}\textbf{(Chore Division with Fixed Earnings)}
	\label{CDFI}

	\noindent\textbf{Given:} A set $A = \langle a_1,a_2, \dots , a_n \rangle$ of agents, set $B = \langle b_1,b_2, \dots , b_m \rangle$ of  bads/chores, disutility function $d:A \times B \rightarrow  (0,\infty)$, and an earning  function $e:A  \rightarrow \mathbb{R}_{\geq 0}$, and a threshold $\tau > 0$.\\
	\textbf{Find:} Price vector $p = \langle p(b_1),p(b_2), \dots ,p(b_m) \rangle \in \mathbb{R}^m_{\geq 0}$ and allocation $X \in \mathbb{R}^{n \times m}_{\geq 0}$ such that,
	\begin{enumerate}
		\item $X_{ij} > 0$ only if $d(a_i,b_j) < \tau$ and $\tfrac{d(a_i,b_j)}{p(b_j)} \leq \tfrac{d(a_i,b_{j'})}{p(b_{j'})}$ for all $j' \in [m]$,  
		\item $\sum_{j \in [m]} X_{ij} \cdot p(b_j) = e(a_i)$ for all $i \in [n]$, and 
		\item $\sum_{i \in [n]} X_{ij} = 1$ 
		for all $j \in [m]$.
	\end{enumerate}
	Tuple $\langle p, X\rangle$ satisfying the above is a competitive equilibrium.
\end{problem}

Analogously a $(1-\epsilon)$-CE satisfies (1), (2) as is, and (3) approximately.

Observe that Problem~\ref{CDFI} can be modeled as a special case of Problem~\ref{CD}. Given an instance $I = \langle A,B,d(\cdot, \cdot), e( \cdot ) \rangle$ of Problem~\ref{CDFI}, we can construct an instance $I' = \langle A,B,d(\cdot, \cdot), w( \cdot, \cdot ) \rangle$ of Problem~\ref{CD} such that $w(a_i,b_j) = e(a_i)$ for all $b_j \in B$. Since the equilibrium price vector $p = \langle p(b_1),p(b_2), \dots ,p(b_m) \rangle$ is scale-invariant, we can assume without loss of generality that $\sum_{j \in [m]} p(b_j)$ $ = 1$. One can verify now that at every competitive equilibrium of $I'$ we have $\sum_{j \in [m]} X_{ij} \cdot p(b_j) = \sum_{j \in [m]} w(a_i,b_j) \cdot p(b_j) = e(a_i) \cdot \sum_{j \in [m]}p(b_j) =  e(a_i)$ for all $i \in [n]$. Since all the other conditions that need to be satisfied by the equilibrium price vector at $I'$ and $I$ are identical, we have that every competitive equilibrium in $I'$ is also a competitive equilibrium in $I$. \emph{Therefore all the hardness results that apply for chore division with fixed earnings (Problem~\ref{CDFI}), will automatically apply to chore division (Problem~\ref{CD}) as well.} 
\medskip

\noindent{\bf Chore Division with Equal Earnings.} A special case of Problem~\ref{CDFI} (and consequently Problem~\ref{CD}) is the \emph{Chore Division with Equal Earnings} where we have $e(a_i) = 1$ for all $i \in [n]$. Analogous notions to chore division with equal earnings, in the context of goods is \emph{Competitive Equilibrium with Equal Incomes} (CEEI)~\cite{Varian74}. CEEI has been extensively studied in the scenario of dividing goods in the fair division community.

\subsection{Technical Overview}\label{technical-overview}
In this section, we sketch the main techniques and ideas used to prove our results in later sections. 
\subsubsection{NP-Hardness of Checking Existence}\label{mainres1}
We first show that there are very simple instances in chore division that does not admit a competitive equilibrium. 

\begin{example}\label{eg}
 Consider a scenario with two agents $a_1$ and $a_2$, and two chores $b_1$ and $b_2$. We have $w(a_i,b_j) = 1$ for all $i,j \in [2]$, and $d(a_1,b_1) = d(a_2,b_1) = 1$, $d(a_1,b_2) = \tau \gg 2$ and $d(a_2,b_2) = 2$. Let $p(b_1)$ and $p(b_2)$ be the prices of the chores at a competitive equilibrium. Observe that since $d(a_1,b_2) = \tau$,  $a_1$ earns her entire money from $b_1$.  Therefore, $2 \cdot p(b_1) \geq p(b_1) + p(b_2)$, implying that $p(b_1) \geq p(b_2)$. In that case, observe that the disutility to price ratio of $b_1$ is strictly less than that of $b_2$ for $a_2$: $\tfrac{d(a_2,b_1)}{p(b_1)}  = \tfrac{1}{p(b_1)} < \tfrac{2}{p(b_1)} \leq \tfrac{2}{p(b_2)} = \tfrac{d(a_2,b_2)}{p(b_2)}$. Thus, none of the agents is willing to do chore $b_2$, and therefore it remains unassigned, a contradiction. 
\end{example}

Observe that this example is an instance of chore division with equal earnings. It is a natural question to ask whether an instance admits a competitive equilibrium. In the analogous context of dividing goods (linear exchange model), there exists polynomial-time verifiable necessary and sufficient conditions~\cite{DevanurGV16} for an instance to admit a competitive equilibrium. However, this seems not the case for chore division and surprisingly this problem turns out to be intractable. To this end, we show the first result of our paper,

\begin{theorem}\label{mainthm1}
    Determining whether an instance of chore division with fixed earnings admits a competitive equilibrium is strongly NP-hard.
\end{theorem}   

In fact, we can also show the same hardness for finding an $(\tfrac{11}{12} + \delta)$-competitive equilibrium for any $\delta > 0$ for chore division with equal earnings too (all the details of the reduction are presented in Section~\ref{nphardness}).

\subsubsection{Instances that Admit a Competitive Equilibrium}\label{mainres2} 
\paragraph{Sufficiency Conditions.}The hard instances constructed to prove Theorem~\ref{mainthm1} have similar structure in disutility matrix as {Example \ref{eg}}. 
In particular, the non-existence is primarily attributed to the fact that there is a classification of the agents into two groups: \emph{rigid} and \emph{flexible}. The rigid agents have small disutility towards some  chores (let us call them \emph{easy} chores) and have disutility of $\tau$ for some chores (we call them \emph{hard} chores) while the \emph{flexible} agents have small disutility towards the easy chores and high (\textit{but still $<\tau$}) disutility towards the hard chores. Since the flexible agents need to do the hard chores entirely, it may not be possible to have a competitive equilibrium, as the prices of the hard chores need to be high (for them to have small disutility to price ratio) and sometimes it may be higher than the stipulated/fixed earning\footnote{the sum of prices of the chores owned by flexible agents in the case of Chore Division (Problem~\ref{CD}).}of the flexible agents. Therefore, we look into instances that do not have such a structure in the disutility matrix; in particular let us define the \emph{disutility graph} of an instance as the bipartite graph $D = (A \cup B, E_D)$ containing the agents $A$ as one vertex set and the chores $B$ as the other, and there exists an edge from agent $a \in A$ to chore $b \in B$ if and only if $d(a,b) <\tau$. Our first sufficiency condition is that,
\begin{align*}
\text{\textbf{Condition 1:} $D$ is a disjoint union of complete bipartite graphs $D_1,D_2,\dots ,D_d$, for some $d\ge 1$.}
\end{align*}
 
Observe that if an instance satisfies Condition 1, then there exists no global partitioning of the agents into groups of rigid and flexible agents: for every pair of agents, the set of chores towards which they have a disutility of less than $\tau$ is either identical or disjoint. We remark that this itself is a sufficient condition to guarantee the existence of competitive equilibrium in the context of chore division with fixed earnings (Problem~\ref{CDFI}): We break the instance into $[d]$ sub-instances formed by the agents $A_i$ and the chores $B_i$ for all $i \in [d]$. Note that in any competitive equilibria agents in $A_i$ will not earn any money from the chores in $B_{i'}$ for all $i \neq i'$. Therefore, it suffices to show that there exists competitive equilibria in each sub instance formed by $A_i$ and $B_i$ for all $i \in [d]$.  Bogomolnaia et al.~\cite{BogomolnaiaMSY17} show that a competitive equilibria always exists in chore division with fixed earnings if the disutility graph of the instance is a complete bipartite graph. Unfortunately, this condition alone is not sufficient to guarantee the existence of a competitive equilibrium in the context of general chore division (Problem \ref{CD}). Consider the following simple example: 

\begin{example}
	We have two agents $a_1$ and $a_2$ and two chores $b_1$ and $b_2$. Let us assume $w({a_1,b_1}) = 1$, $w({a_1,b_2}) = \tfrac{1}{2}$, $w({a_2,b_1}) = 0$, and $w({a_2,b_2}) = \tfrac{1}{2}$. Also, suppose $d(a_1,b_1) = d(a_2,b_2) = 1$ and $d(a_1,b_2) = d(a_2,b_1) = \tau$. Note that the disutility graph $D$ is the union of two disjoint complete bipartite graphs, namely $D_1 = (\left\{ a_1 \right\} \cup \left\{ b_1 \right\}, \left\{ (a_1,b_1) \right\})$ and $D_2 (\left\{ a_2 \right\} \cup \left\{ b_2 \right\}, \left\{ (a_2,b_2) \right\})$. However, we do not have any competitive equilibrium as the earning of $a_2$ will be half of the price of $b_2$ and thus $a_2$ can never do more than half of the chore $b_2$. Also, $a_1$ will not do any of chore $b_2$ as she has  disutility of $\tau$ for the same. Therefore, $b_2$ can not be fully assigned.
\end{example}

 Thus, we need more conditions to guarantee the existence of a competitive equilibrium. To this end, we define the \emph{exchange graph} of an instance as a  graph $W = ([d], E_W)$. We have $(i,j) \in E_W$ if and only if for every chore $b$ in the component $D_j$ of the disutility matrix, there is an agent $a \in D_i$ such that $w({a,b}) > 0$ \footnote{Intuitively, agents in the component $D_j$ are willing to do chores brought by the agents in $D_i$.}. We now propose our second sufficiency condition.
\begin{align*}
\text{\textbf{Condition 2:} $W$ is strongly connected.}
\end{align*}

With this, we state the second main result of our paper.

\begin{theorem}
	\label{mainthm2}
	All instances of chore division that satisfy Conditions 1 and 2, admit a competitive equilibrium.
\end{theorem}

 \paragraph{Existence under Sufficiency Conditions.} We now briefly sketch the ideas behind our proof of existence (Theorem~\ref{mainthm2}) and also highlight some crucial barriers that we overcome. Like most equilibrium existence results~\cite{Nash51, ArrowD54}, our proof uses a fixed-point theorem. Given a function or a correspondence $\phi$ from $D$ to $D$, a fixed-point theorem ensures existence of $x\in D$ such that $f(x)=x$ or $x \in f(x)$ respectively, if $\phi$ satisfies continuity-like properties, and $D$ is convex and compact~\cite{brouwer1911, kakutani1941}.

 The known existence proofs in the context of dividing goods, say $g_1, \dots, g_m$,  defines fixed-point formulations that operate on a simplex as the domain for price vectors~\cite{ArrowD54, Maxfield97}. That is $\Delta_n=\{p\ |\ p(g_j) \geq 0$ for all $j \in [m]$, and $\sum_{j \in [m]} p(g_j) = 1\}$. Given prices, it computes optimal bundles of agents and adjust prices based on excess demand, i.e., demand minus supply. 

This approach immediately fails for the chore division problem due to the issue of {\em infeasible optimal bundle} under this price domain. Here is why: Recall at prices $p$, an optimal bundle of agent $a_i$ minimizes her disutility $\sum_{j \in [m]} d(a_i,b_j) \cdot X_{ij}$ subject to 
\begin{itemize}
	\item[$(a)$] she earns at least the payment she needs to make for her endowment, i.e., $\sum_{j \in [m]}X_{ij} \cdot p(b_j) \geq \sum_{j \in [m]} w(a_i,b_j) \cdot p(b_j)$, and 
	\item[$(b)$] $X_{ij} > 0$ only if $d(a_i,b_j) < \tau$. 
\end{itemize}

Consider a price vector $p\in \Delta_n$ such that for a particular component $D_k = (A_k \cup B_k, E_k)$ in the disutility graph we have the price of every chore in $B_k$ to be zero. And, there is an agent $a_i \in A_k$ such that $\sum_{j \in [m]}w(a_i,b_j) \cdot p(b_j) > 0$. Note that the optimal bundle set for $a_i$ is empty here since under condition $(b)$ above, condition $(a)$ becomes infeasible. That is, all the chores where $a_i$ has a disutility of less than $\tau$ have zero prices and thus $a_i$ will never be able to earn enough money to match $\sum_{j \in [m]} w(a_i,b_j) \cdot p(b_j)$. This issue does not arise in the goods only case because there condition $(a)$ becomes ``agent spends only as much as she earns from her endowment'' which essentially puts an {\em upper bound} on her spending reversing the inequality, and then all $X_{ij}$'s being $0$ is always feasible.

\emph{In order to avoid the issue of ``zero prices", we add more constraints to domain of the price vectors}. The domain of our fixed-point formulation is defined by the price vectors $p = \langle p(b_1), \dots , p(b_m) \rangle$ (where each $p(b_j)$ is the price of the chore $b_j$)  and allocations $X = \langle X_{11},X_{12}, \dots, X_{nm} \rangle$ (where each $X_{ij} \geq 0$ indicates the consumption of chore $b_j$ by agent $a_i$) such that,
 
 \begin{itemize}
 	\item $p(b_j) \geq 0$ for all $j \in [m]$, and $\sum_{j \in [m]} p(b_j) = 1$, and
 	\item $\sum_{a_i \in A_k} \sum_{j \in [m]} w(a_i,b_j) \cdot p(b_j) = \sum_{b_j \in B_k} p(b_j)$ for each connected component $D_k = (A_k \cup B_k , E_k )$ of the disutility graph, and    
 	\item for all $i,j$, we have $0 \leq X_{ij} \leq M $ for some sufficiently large $M$.
 \end{itemize} 

Let $P$ denote the set of all prices and $\X$ denote the set of all allocations that satisfy the above conditions. First observe that for any $p \in P$ we cannot have all prices of a component $D_k$ of the disutility graph to be zero: Without loss of generality, let us assume that the sum of prices of the chores in the components $D_{k}$ for all $k \in [\ell]$  are zero and the sum of prices of the chores in all other components ($D_{k}$ for all $k \in [d] \setminus [\ell]$)  are non-zero. Note that since $\sum_{j \in [m]} p(b_j) =1$, there is at least one component where the sum of prices of the chores is non-zero and thus $[d] \setminus [\ell] \neq \emptyset$. Also, note that for all $k \in [\ell]$ we have $\sum_{a_i \in A_k} \sum_{j \in [m]} w(a_i,b_j) \cdot p(b_j) = \sum_{b_j \in B_k} p(b_j) = 0$ (by the third constraint in the definition of $P$). Since our instance satisfies Condition 2, we have that there exists an edge in the exchange graph $W=([d],E_W)$ from a $k \in [\ell]$ to a $k' \in [d] \setminus [\ell]$. This implies that for each chore in $D_{k'}$ (in particular for the chore with a non-zero price), there exists an agent in $ D_k$ with some positive endowment of that chore, further implying that $\sum_{a_i \in A_k} \sum_{j \in [m]} w(a_i,b_j) \cdot p(b_j) > 0$, which is a contradiction. Therefore, we cannot have all the prices of any component $D_k$ in the disutility graph to be zero. This resolves the issue of ``zero prices'' and thereby ensures feasibility of the optimal bundle sets (Lemma \ref{optimalbundlenotempty}). Now the crucial task is to define a ``continuous'' map that maintains the extra conditions on the price space. 

Let $S = \cup_{p \in P} \cup_{X \in \X} \langle p, X \rangle $. In order to capture competitive equilibria at fixed-points we define a \emph{correspondence} $\phi: S \rightarrow 2^{S}$ that satisfies the following properties: it maps any $\langle p,X  \rangle$ to all $ \langle q,X^p \rangle $ where,
\begin{itemize}
	\item the price vector $q$ is such that $q \in P$ and for all components $D_k = (A_k \cup B_k,E_k)$ of the disutility graph, we have
	 \begin{align}
	 \label{qdef}
	  \frac{q(b_j)}{q(b_{j'})} &= \frac{p(b_j) + \mathit{max}(\sum_{i \in [n]}w(a_i,b_j) - \sum_{i \in [n]} X_{ij},0)}{p(b_{j'}) + \mathit{max}(\sum_{i \in [n]}w(a_i,b_{j'}) - \sum_{i \in [n]} X_{ij'},0)}, 
	 \end{align}
	 for all chores $b_j$ and $b_{j'}$ in $B_k$.
	\item $X^p \in \X$ is an optimal bundle set at the prices $p$: For every agent $a_i$, we have $ \langle X^p_{i1}, \dots ,X^p_{im} \rangle$ to be the optimal demand bundle for $a_i$ at prices $p$, if and only if, $\langle X^p_{i1}, \dots ,X^p_{im} \rangle$ minimizes the total disutility $\sum_{j \in [m]} d(a_i,b_j) \cdot X^p_{i,j}$ of agent $a_i$ subject to 
	\begin{itemize}
		\item $X^p_{ij} > 0$ only if $d(a_i,b_j) < \tau$, and 
		\item the earning constraint $\sum_{j \in [m]}X^p_{ij} \cdot p(b_j) \geq \sum_{j \in [m]} w(a_i,b_j) \cdot p(b_j)$.
	\end{itemize}
\end{itemize} 
 
At the first glance, it may not even be clear that why such a correspondence can be defined: In particular, finding the price vector $q$ that satisfies the properties as mentioned above. However, we do manage to construct such a correspondence $\phi$ and show that it has at least one fixed-point in Section \ref{sufficiency}. Before we sketch the construction, let us first discuss how the above properties are important to show that the fixed-points of $\phi$ correspond to competitive equilibrium. 

Consider any fixed point of $\phi$, i.e., $\langle p, X \rangle$ such that $\langle p, X \rangle  \in \phi(\langle p, X \rangle)$. Let $r_j(X)$ denote the amount of chore $b_j$ left undone under $X$, i.e.,  $r_j(X) =\mathit{max}(\sum_{i \in [n]}w(a_i,b_j) - \sum_{i \in [n]} X_{ij},0)$. We have $q = p$ and $X^p = X$, where $q$ and $X^p$ are constructed from $p$ and $X$ as above. Thus, every agent earns her money from the chores in her optimal bundle (as $X = X^p$) and if $X^p$ is not a competitive equilibrium, then the only reason is because some chore $b_j$ is not completed, i.e.,  $r_j(X) > 0$. This also implies that there is some chore $b_{j'}$ that is overdone ($\sum_{i \in [n]}w(a_i,b_{j'}) - \sum_{i \in [n]} X_{ij'} < 0$, implying $r_{j'}(X)=0$), as the total earning of the agents is at least the total prices of the chores (by the definition of $\X^p$).  Therefore, we have $\tfrac{q(b_j)}{q(b_{j'})} = \tfrac{p(b_j) + r_{j}(X)}{p(b_{j'}) + r_{j'}(X)} > \tfrac{p(b_j)}{p(b_{j'}) + r_{j'}(X)} = \tfrac{p(b_j)}{p(b_{j'})}$, implying that $q \neq p$, which is a contradiction. 

Now, the main challenge (and also the main bulk of our effort) is to show that such a correspondence $\phi$ can be defined. That is, to get the price vector $q \in \mathbb{R}^m_{\ge 0}$ that satisfy the ratio property as outlined in~(\ref{qdef}), and all the constraints that define the price space $P$, especially the constraint that for every component $D_k$ in the disutility graph, the total price of chores owned by the agents in $A_k$ equals the total price of the chores in $B_k$. For this, we first compute $q'(b_j)$'s that satisfy the ratio property, but may not be in price space $P$. Then, for each component taking its total price as a variable, we define parameterized $q$ from $q'$ that preserves price ratios of $q'$. Now the goal is to find values of the variables once we replace them in $q$ we get $q\in P$. In other words, $q$ satisfy all the component-wise price constraints of domain $P$.

This condition essentially gives a linear system of equations. Let $M \in \mathbb{R}^{d\times d}$ denote the matrix of this linear system. Then, our goal becomes to find a non-negative vector $v \in \mathbb{R}^d$ in the null space of $M$.  It is not obvious why such a vector should exist. Our high-level approach to show the same is as follows: We can equivalently express the linear system of equations $M \cdot v =0$ as $M' \cdot v = v$, where $M' = M + I$, where $I$ is the identity matrix. We show that if we define a function $f:\mathbb{R}^d \rightarrow \mathbb{R}^d$ as $ f(v) = M' \cdot v$, then $f$ maps the $d$-dimensional simplex $\Delta_d$ to itself. Restricting $f$ to only the simplex, we get a continuous map $f:\Delta_d\rightarrow \Delta_d$ and therefore it has a fixed-point by the Brouwer's fixed-point theorem. At every fixed-point $v$ we have $M' \cdot v= v$ implying $M \cdot v=0$. Since $v \in \Delta_d$, we get the vector we needed. 

Finally, to apply the Kakutani's fixed point theorem, we need to show that our correspondence has a closed graph (analogous to continuity for a function) and $\phi(\langle p, X \rangle)$ is non-empty and convex. Since defining the correspondence involves multiple steps, this does not follow trivially. Once we prove this in Lemma~\ref{P4} in Section~\ref{sufficiency}, we get the existence of a fixed-point of $\phi$, and thereby the existence of a competitive equilibrium. We refer the reader to Section~\ref{sufficiency} for a detailed formal discussion.     

\subsubsection{Hardness of Determining Equilibrium under Sufficiency Conditions}\label{mainres3}
Recall from Section~\ref{mainres2} that the NP-hardness of Chore division crucially uses the structure of rigid and flexible agents. However, this structure in the disutility matrix is completely removed if we concentrate on instances satisfying Conditions 1 and 2, mentioned in Section~\ref{mainres2}. We show that even under these sufficiency conditions, chore division is still intractable. We now mention our third main technical result:

\begin{theorem}\label{mainthm3}
	 Chore division, restricted to all instances that satisfy Conditions 1 and 2, is PPAD-hard. 
\end{theorem}

We show the PPAD-hardness by reduction from normalized polymatrix game (defined below), which is known to be PPAD-hard~\cite{chen2017complexity}.
\begin{problem}\textbf{(Normalized Polymatrix Game)~\cite{chen2017complexity}}\\
\label{prob:ne}
	\textbf{Given}: A $2n \times 2n$ rational matrix $\M$ with every entry in $[0,1]$ and $\M_{i,2j-1} + \M_{i,2j} = 1$ for all $i \in [2n]$ and $j \in [n]$ .\\
	\textbf{Find}: Equilibrium strategy vector $x \in \mathbb{R}^{2n}_{\geq 0}$ such that $x_{2i-1} + x_{2i} = 1$ and 
	\begin{align*}
	& x^T \cdot \M_{*,{2i-1}} > x^T \cdot \M_{*,2i} + \tfrac{1}{n} \implies x_{2i} = 0.\\
	& x^T \cdot \M_{*,{2i}} > x^T \cdot \M_{*,2i-1} + \tfrac{1}{n} \implies x_{2i-1} = 0.
	\end{align*}
	where $M_{*,k}$ represents the $k^{\mathit{th}}$ column of the matrix $\M$.	
\end{problem}

Given an instance $I = \langle \M \rangle $ of the polymatrix game, we will now create an instance of chore division $E(I) = \langle A \cup  B, d( \cdot, \cdot) , w(\cdot , \cdot ) \rangle$ such that given any competitive equilibrium in $E(I)$, we can determine in polynomial time, an equilibrium strategy vector $x$ for $I$. We note that our reduction can be extended to work even in the restricted setting of $(1 - \tfrac{1}{\textup{poly}(n)})$-competitive equilibria. However, for simplicity we stick to the exact competitive equilibria in our reduction. Next, we give a brief sketch of the proof, while the formal proof with all the details is in Section \ref{ppadhardness}.

\paragraph{Sketch of the Reduction.} 
The key properties that our hard instance $E(I)$ exhibit are  \emph{pairwise equal endowments}, \emph{fixed earning}, \emph{price equality}, \emph{price regulation} and \emph{reverse ratio amplification} (we will give a precise definition of these properties shortly). These techniques (constructing hard instances exhibiting these properties) have been used earlier to prove PPAD-hardness for the exchange model with only goods when agents have constant elasticity of substitution (CES) utilities~\cite{chen2017complexity} and even for the Fisher model when agents have separable piecewise linear concave (SPLC) utilities~\cite{ChenT09}. \emph{However, in the context of chore division, we are able to obtain these gadgets (in particular the reverse ratio-amplification gadget) even when agents have linear disutility functions (dividing goods when agents have linear utility functions is tractable)}. 

We now give a brief overview of the reduction: Let us consider two sets of chores, namely $B_1 = \left\{ b^1_1,b^1_2, \dots ,b^1_{2n} \right\}$ and $B_2 = \left\{ b^2_1,b^2_2, \dots ,b^2_{2n} \right\}$. For all $i,j \in [2n]$ we define agents $a_{i,j}$ such that agent $a_{i,j}$ brings $\M_{i,j}$ units of chore $b^2_i$. The disutilities are as follows:
\begin{align*}
d(a_{i,2j-1},b^{1}_{2j-1}) &= (1- \alpha) &\text{and}&    & d(a_{i,2j-1},b^{1}_{2j}) &= (1+ \alpha)\\
d(a_{i,2j},b^{1}_{2j-1}) &= (1+ \alpha)   &\text{and}&    & d(a_{i,2j},b^{1}_{2j}) &= (1 - \alpha),
\end{align*}
for some infinitesimally small $\alpha > 0$. Agents $a_{i,2j-1}$ and $a_{i,2j}$ have disutility more than $\tau$ for all other chores. Our instance also has agents $a'_1,a'_2, \dots a'_{2n}$ such that for each  $i \in [2n]$, we have,
\begin{align*}
d(a'_{2i-1},b^{1}_{2i-1}) &= (1- \alpha) &\text{and}&    & d(a'_{2i-1},b^{1}_{2i}) &= (1+ \alpha)\\
d(a'_{2i},b^{1}_{2i-1}) &= (1+ \alpha)   &\text{and}&    & d(a'_{2i},b^{1}_{2i}) &= (1 - \alpha).
\end{align*}
Our instance $E(I)$ of chore division has additional agents and chores to ensure the following five crucial properties at equilibrium, however, none of the other agents in $E(I)$ have disutility less than $\tau$ towards the chores in $B_1$.
\begin{itemize}
	\item \emph{Pairwise equal endowments:} In $E(I)$, the total endowment of chore $b^{k}_{2i-1}$ equals the total endowment of chore $b^k_{2i}$ for all $k \in [2]$. Also the total endowments of each chore is $\mathcal{O}(n)$. 
	\item \emph{Fixed earning:} In any competitive equilibrium, for each $i \in [2n]$, we have the total earning of agent $a'_i$ is $(1- \beta) \cdot (2n - \sum_{j \in [2n]} \M_{j,i})$.  
	
	\item \emph{Price equality:} In any competitive equilibrium, the sum of prices of chores $b^k_{2i-1}$ and $b^k_{2i}$ is the same for all $i \in [n]$ and $k \in [2]$. Since the prices of chores at a competitive equilibrium is scale invariant, we can assume further, without loss of generality, that for all $i \in [n]$, we have $p(b^1_{2i-1}) + p(b^1_{2i}) = p(b^2_{2i-1}) + p(b^2_{2i}) = 2$.  
	\item \emph{Price regulation: } In any competitive equilibrium, for all $i \in [2n]$ we have 
	\begin{align*} 
	\frac{1 - \alpha}{1 + \alpha} \leq \frac{p(b^1_{2i-1})}{p(b^1_{2i})} \leq \frac{1 + \alpha}{1 - \alpha}& &\text{and}& &\frac{1 - \beta}{1 + \beta} \leq \frac{p(b^2_{2i-1})}{p(b^2_{2i})} \leq \frac{1 + \beta}{1 - \beta},
	\end{align*}
	for some infinitesimally small $\beta \gg {n^2 \cdot \alpha}$.
	\item \emph{Reverse ratio amplification:} In any competitive equilibrium, for all $i \in [2n]$, if $\tfrac{p(b^1_{2i-1})}{p(b^1_{2i})} = \tfrac{1 + \alpha}{1 - \alpha}$, then we have $\tfrac{p(b^2_{2i-1})}{p(b^2_{2i})} = \tfrac{1 - \beta}{1 + \beta}$ and similarly when  $\tfrac{p(b^1_{2i-1})}{p(b^1_{2i})} = \tfrac{1 - \alpha}{1 + \alpha}$, then we have $\tfrac{p(b^2_{2i-1})}{p(b^2_{2i})} = \tfrac{1 + \beta}{1 - \beta}$.									
\end{itemize}

Our instance $E(I)$ satisfies both Conditions 1 and 2 of Section~\ref{sufficiency}. As a result, it admits a competitive equilibrium. Given that a competitive equilibrium has to satisfy the above properties, we now show how we can use prices at the competitive equilibrium to obtain a Nash equilibrium strategy vector for the polymatrix game in polynomial time. Given a competitive equilibrium price vector $p$, construct vector $x$ as follows. 
\begin{align*}
\forall i\in [2n],\ \ \ x_i &= \frac{p(b^2_{i}) - (1-\beta)}{2 \cdot \beta}
\end{align*}

We will now show that $x$ is the desired Nash equilibrium strategy vector for instance $I$ of the polymatrix game, i.e., it satisfies conditions of Problem \ref{prob:ne}. To this end, first observe that since our instance satisfies price equality we have that $p(b^2_{2i-1}) + p(b^2_{2i}) = 2$ for all $i \in [n]$. Again, since our instance satisfies price regulation, we also have for all $i \in [n]$,  $\tfrac{1-\beta}{1+\beta} \leq \tfrac{p(b^2_{2i-1})}{p(b^2_{2i})} \leq \tfrac{1 + \beta}{1 - \beta}$, implying that $p(b^k_i) \geq 1- \beta$ for all $i \in [2n]$. Therefore, $x_i \geq 0$ for all $i \in [2n]$. 

Also note that we have $x_{2i-1} +x_{2i} =  \tfrac{p(b^2_{2i-1})  + p(b^2_{2i}) - 2(1-\beta)}{2 \cdot \beta} = \tfrac{2\beta}{2\beta} = 1$ (by the price equality property we have $p(b^2_{2i-1}) + p(b^2_{2i})=2$) for all $i \in [n]$. Now, we show that if $x^T \cdot \M_{*,{2i}} > x^T \cdot \M_{*,2i-1} + \tfrac{1}{n}$, then $ x_{2i-1} = 0$. The proof for the symmetric condition is similar. Let us assume that  $x^T \cdot \M_{*,{2i}} > x^T \cdot \M_{*,2i-1} + \tfrac{1}{n}$. Observe that the agents that have a disutility of $(1- \alpha)$ towards chore $b^1_{2i}$ are $\left\{ \cup_{j \in [2n]} a_{j,2i} \right\} \cup a'_{2i}$, and that their total earning is 
\begin{align*}
&=\sum_{j \in [2n]} \M_{j,2i} \cdot p(b^2_j) + (1- \beta) \cdot (2n - \sum_{j \in [2n]} \M_{j,2i}) &\text{(fixed earning)}\\
&=\sum_{j \in [2n]} \M_{j,2i} \cdot (2\beta \cdot x_{j} + (1-\beta)) +  (1- \beta) \cdot (2n - \sum_{j \in [2n]} \M_{j,2i}) &\text{(substituting $p(b^2_j)$)}\\ 
&=\sum_{j \in [2n]} 2\beta \cdot x_{j} \cdot \M_{j,2i} + (1-\beta) \cdot\sum_{j \in [2n]}  \M_{j,2i} +  (1- \beta) \cdot (2n - \sum_{j \in [2n]} \M_{j,2i})\\          
&= 2\beta x^T \cdot \M_{*,2i} +  2n \cdot (1-\beta).               
\end{align*}

Similarly, the total earning of the agents that have a disutility of $(1 - \alpha)$ towards $b^1_{2i-1}$ is  $2 \beta x^T \cdot \M_{*,2i-1} +  2n \cdot (1-\beta)$. Observe that the agents with disutility $(1-\alpha)$ towards $b^1_{2i}$ can earn all their money only from either $b^1_{2i}$ or $b^1_{2i-1}$ (as these are the only chores towards which they have finite disutility). Also note that the total endowment of both $b^1_{2i-1}$ and $b^1_{2i}$ is $cn$ for some $c \in \mathcal{O}(1)$ (as our instance also satisfies pairwise equal endowments). Now if, the agents $\left\{ \cup_{j \in [2n]} a_{j,2i} \right\} \cup a'_{2i}$ earn their money entirely from $b^1_{2i}$, then we will have $p(b^1_{2i}) \geq \tfrac{1}{cn} \cdot (2\beta x^T \cdot \M_{*,2i} +  2n \cdot (1-\beta))$ and $p(b^1_{2i-1}) \leq \tfrac{1}{cn} \cdot (2\beta x^T \cdot \M_{*,2i-1} +  2n \cdot (1-\beta) )$. Since,  $x^T \cdot \M_{*,{2i}} > x^T \cdot \M_{*,2i-1} + \tfrac{1}{n}$. we have $p(b^1_{2i}) > p(b^1_{2i-1}) + \tfrac{1}{n} \cdot \tfrac{2}{cn}\beta = \tfrac{2\beta}{cn^2}$. Again, since $\beta \gg n^2 \cdot \alpha $, we have that $\tfrac{p(b^1_{2i})}{p(b^1_{2i-1})} > \tfrac{1 + \alpha} {1 - \alpha}$, which is a contradiction, as our construction satisfies the price regulation property. Therefore, the agents that have a disutility of $(1-\alpha)$ towards $b^1_{2i}$ should also earn some positive units of  money from $b^1_{2i-1}$. But this is only possible if $\tfrac{p(b^1_{2i})}{p(b^1_{2i-1})} = \tfrac{1 - \alpha}{1 + \alpha}$. Since our instance also satisfies the reverse ratio amplification property, we have that $\tfrac{p(b^2_{2i})}{p(b^2_{2i-1})} = \tfrac{1 + \beta}{1 - \beta}$. Since $p(b^2_{2i}) + p(b^2_{2i-1}) = 2$ by price equality property, we have that $p(b^2_{2i-1}) = 1 - \beta$. Therefore we have,
\begin{align*}
x_{2i-1} &=\frac{(1- \beta) - (1-\beta)}{2 \cdot \beta}\\
&=0.
\end{align*}

A very similar argument shows that when $x^T \cdot \M_{*,{2i-1}} > x^T \cdot \M_{*,2i} + \tfrac{1}{n}$, then $ x_{2i} = 0$. A more elaborate discussion of the proof, where we give complete construction and prove that the instance exhibits all the properties, is in Section~\ref{ppadhardness}.  

\subsection{Further Related Work}
The fair division literature is too vast to survey here, so we refer to the excellent books~\cite{BramsT96,RobertsonW98,Moulin03} and a recent survey article~\cite{Moulin19}, and restrict attention to previous work that appears most relevant. 

Most of the work in fair division is focused on allocating goods with a few exceptions of chores~\cite{Su99,AzrieliS14, BramsT96, RobertsonW98}. Recent seminal papers of Bogomolnaia et al.~\cite{BogomolnaiaMSY17,BogomolnaiaMSY19} consider the case of mixed manna that contains both goods and bads. For the goods case, competitive equilibrium maximizes the Nash welfare, i.e., geometric mean of agents' utilities. In case of chores (or mixed manna),~\cite{BogomolnaiaMSY17} shows that critical points of the geometric mean of agents' disutilities on the (Pareto) efficiency frontier are the competitive equilibrium profiles. 

The fair allocation of \emph{indivisible} items is also an intensely studied problem for the case when all items are goods with a few recent exceptions~\cite{AzizRSW17,AzizCL19,HuangL19,AzizCL19a,AzizCIW19,SandomirskiyH19}. Since the standard notions of fairness such as envy-freeness are not applicable, alternate notions have been defined and studied for this case; see~\cite{LiptonMMS04,Budish11,CaragiannisKMPSW16, PlautR18, ChaudhuryGM20,GhodsiHSSY17,GargKK20} for a subset of notable work and references therein. The Nash welfare continues to serve as a major focal point in this case as well, for which approximation algorithms have been obtained under several classes of utility functions including linear~\cite{ColeG15,ColeDGJMVY17,AnariGSS17,AnariMGV18,BarmanKV18,GargHM18,ChaudhuryCGGHM18}.
\medskip

\paragraph{Organization of the Paper.} We present our three main results in the upcoming sections. Section~\ref{nphardness} contains the NP-hardness of chore division with fixed earnings (and equal earnings). Section~\ref{sufficiency} contains the sufficiency conditions under which a competitive equilibrium always exists and the proof of existence. Finally, Section~\ref{ppadhardness} contains the PPAD-hardness of finding a competitive allocation even under the sufficiency conditions. 

\section{Complexity of Determining Existence of Efficient Equilibrium}\label{nphardness}
In this section, we show that finding a $(\tfrac{11}{12} + \delta)$-competitive equilibrium for any $\delta>0$ in chore division with fixed earnings (Problem~\ref{CDFI}) and consequently chore division (Problem~\ref{CD}) is strongly NP-hard. Later, in the section we also show how to modify our instance to obtain NP-hardness for finding a $(\tfrac{11}{12} + \delta)$-competitive equilibrium even in chore division with equal earnings.  We show that any polynomial time algorithm that determines whether an instance admits a $(\tfrac{11}{12} + \delta)$-competitive equilibrium in chore division with fixed earnings  implies that 3-SAT is solvable in polynomial time. To this end, we recall the 3-SAT problem:

\begin{problem}\textbf{(3-SAT)}\\
	\textbf{Given}: A set of variables $X =  \left\{ x_1,x_2, \dots ,x_n \right\}$ and a set of clauses $\C = \left\{C_1,C_2,\dots, C_m\right\}$ where each clause is a disjunction of exactly three \emph{literals}\footnote{A literal is a variable or the negation of a variable}.\\
	\textbf{Find}: An assignment $A: X \rightarrow \left\{T,F\right\}$ such that all the clauses are \emph{satisfied}\footnote{A clause $C_{r} = \ell_1 \vee \ell_2 \vee \ell_3$ ($\ell_i$ is a literal) is satisfied if and only if $A(\ell_i) = T$ for at least one $i \in [3]$.} or output that no such assignment exists.
\end{problem} 

Given any instance $I = \langle X,\C \rangle$ of 3-SAT, we will create an instance $E(I)$ of chore division such that there exists an efficient competitive equilibrium in $E(I)$ if and only if there exists an assignment $A$ that satisfies all the clauses in $\C$ in $I$. We first briefly sketch the intuition, before we move to the construction of the gadgets required for our reduction.

\paragraph{Several Disconnected Equilibria.} We sketch a very simple scenario that could arise in chore division with fixed earnings: Consider an instance with two agents $a_1$ and $a_2$ with a fixed earning of one unit each. The disutility values are given below where $a_1$ has a disutility of $1$ for $b_1$ and $3$ for $b_2$, while $a_2$ has a disutility of $\tau$ for $b_1$ and $1$ for $b_2$. 
\begin{center}
	\begin{minipage}[b]{0.3\linewidth}
		\centering
		\begin{eqnarray*}
			\setlength{\arraycolsep}{0.5ex}\setlength{\extrarowheight}{0.25ex}
			\begin{array}{@{\hspace{1ex}}c@{\hspace{1ex}}||@{\hspace{1ex}}c@{\hspace{1ex}}|@{\hspace{1ex}}c@{\hspace{1ex}}|@{\hspace{1ex}}c@{\hspace{1ex}}|}
			     	\  & b_1 \ & b_2  \ \\[.5ex] \hline
				 a_1 \ & 1 \ & 3 \  \\[.5ex] \hline
				 a_2 \ & \tau \ & 1 \  \\[.5ex] 
			\end{array}
		\end{eqnarray*}
	\end{minipage}
\end{center}

Let $p = \langle p(b_1), p(b_2) \rangle$ be an equilibrium price vector. Also, throughout this section we use the notation $\MPB_a$  to denote the \emph{minimum pain per buck} bundle for agent $a$ at the prices $p$: a chore $b \in \MPB_a$ if and only if $\tfrac{d(a,b)}{p(b)} \leq \tfrac{d(a,b')}{p(b')}$ for all other chores $b'$ in the instance. Observe that this small instance exhibits exactly two competitive equilibria: 
\begin{itemize}
	\item The first competitive equilibrium is when both $p(b_1)$ and $p(b_2)$ are set to 1. Note that only $\MPB_{a_1} = \left\{b_1\right\}$ and $ \MPB_{a_2} = \left\{b_2\right\}$. Thus $a_1$ earns her entire one unit of money from $g_1$ and $a_2$ earns her entire one unit of money from $g_2$. 
	\item The second competitive equilibrium is when $a_1$ earns from both $b_1$ and $b_2$. For this we set $p(b_1)$ to  $1/2$ and $p(b_2)$ to $3/2$. Note that $MPB_{a_1} = \left\{b_1,b_2\right\}$ and $\MPB_{a_2} = \left\{b_2\right\}$. Under these prices, $a_2$ earns her entire money by doing $2/3$ of $b_2$, and  $a_1$ earns her money by doing all of $b_1$ and $1/3$ of $b_2$.  
\end{itemize}

Also, observe that there exists no competitive equilibrium at any other set of prices (in particular all price vectors that can be expressed as a convex combination of (1,1) and (1/2,3/2)). This is a striking difference to the scenario with only goods to divide, where all competitive equilibrium exists at a unique price vector. Now, let us introduce another agent $a_3$ and another chore $b_3$ in the instance. Let us say that $a_3$ has a fixed earning of one unit, and both agents $a_1$ and $a_2$ have a disutility of $\tau$ towards $b_3$. We now discuss two scenarios that may arise depending on $a_3$'s disutility towards the chores  
\begin{enumerate}
	\item $a_3$ has a disutility of $1$ towards $b_3$ and $b_2$ and $\tau$ towards $b_1$.
	\item $a_3$ has a disutility of $1$ towards $b_3$, $\tfrac{1}{2}$ towards $b_1$ and $\tau$ towards $b_2$.
\end{enumerate}

We will now show that, at a competitive equilibrium, in the first scenario $b_2 \notin \MPB_{a_1}$ and in the second scenario $b_2 \in \MPB_{a_1}$, suggesting that we need to carefully choose the \emph{local} equilibrium among the agents $a_1$, $a_2$ and chores $b_1$ and $b_2$. Let $p(b_1)$, $p(b_2)$ and $p(b_3)$ denote the prices of chores at an equilibrium. Note that since both $a_1$ and $a_2$ have a disutility of $\tau$ for $b_3$, they only earn money from $b_1$ and $b_2$. Thus $p(b_1) + p(b_2) \geq 2$. Note that in both scenarios $b_3$ should be in $\MPB_{a_3}$ as $a_3$ is the only agent with disutility less than $\tau$ towards it. Now,

\begin{itemize}
	\item  In the first scenario, we need to have $\tfrac{d(a_3,b_3)}{p(b_3)} \leq \tfrac{d(a_3,b_3)}{p(b_3)}$ or equivalently $\tfrac{1}{p(b_3)} \leq \tfrac{1}{p(b_2)}$, implying that $p(b_3) \geq p(b_2)$. This in turn implies that 
	    \begin{align*}
	     p(b_2) + 2   &\leq p(b_2) + (p(b_1) + p(b_2)) &\text{(as $p(b_1) + p(b_2)\geq 2$)}\\
	               &\leq p(b_1) + p(b_2) + p(b_3) &\text{(as $p(b_2) \leq p(b_3)$)}\\ 
	               &= 3.
	    \end{align*}
	     Thus we have $p(b_2) \leq 1$, implying that $p(b_1) \geq 1$. Therefore, we can conclude that $b_2 \notin \MPB_{a_1}$ as the disutility to price ratio of $b_1$ is strictly less than that of $b_2$ for agent $a_1$.   
	\item In the second scenario, we need to have $\tfrac{d(a_3,b_3)}{p(b_3)} \leq \tfrac{d(a_3,b_1)}{p(b_1)}$, we have  $\tfrac{1}{p(b_3)} \leq \tfrac{1}{2p(b_1)}$, implying that $p(b_3) \geq 2p(b_1)$. This in turn implies that 
	   \begin{align*}
	     2p(b_1) + 2   &\leq 2p(b_1) + (p(b_1) + p(b_2)) &\text{(as $p(b_1) + p(b_2) \geq 2$)}\\
	               &\leq p(b_1) + p(b_2) + p(b_3) &\text{(as $2p(b_1) \leq p(b_3)$)}\\ 
	                &= 3.
	   \end{align*} 
	    Thus we have $p(b_1) \leq \tfrac{1}{2}$, implying that $p(b_2) \geq \tfrac{3}{2}$. Therefore, we conclude that $b_2 \in \MPB_{a_1}$.
\end{itemize}

Thus, as mentioned earlier, depending on the valuations of the agents outside the local sub-instance, a careful selection of the local equilibrium (among the two disjoint local equilibria) among the agents $a_1$, $a_2$ and chores $b_1$ and $b_2$ is necessary. We will now show that when there are $n$ such local sub-instances (resulting in $2^n$ disjoint equilibria), choosing the correct combination of the local equilibria becomes intractable.    

\subsection{Variable Gadgets}
For each variable $x_i$, we introduce two agents $a_1^i$ and $a_2^i$ and two chores $b_1^i$ and $b_2^i$. We set
\begin{align*}
 d(a_1^i,b_1^i) &= 1,  &d(a_1^i,b_2^i) &= 3,\\
 d(a_2^i,b_1^i) &= \tau, &d(a_2^i,b_2^i) &= 1.
\end{align*} 
See Figure~\ref{nphardfig} for an illustration. We set the earnings of both $a_1^i$ and $a_2^i$ to be one, i.e., $e(a_1^i)= e(a_2^i)=1$. Also, for all $i \in [n]$ agents $a^i_1$ and $a^i_2$ have a disutility of $\tau$ for all other goods in the instance (that have been introduced and will be introduced by clause gadgets in the next section).

\subsection{Clause Gadgets}
For each clause $C_{r} = (\ell_i \vee \ell_j \vee \ell_k )$, where $\ell_i$ is either the variable $x_i$ or its negation $\neg x_i$, we introduce four agents $n_{i}^r$, $n_{j}^{r}$, $n_{k}^{r}$ and $\n^r$, and three chores $m_{i}^{r}$, $m_{j}^{r}$, and  $m_{k}^{r}$. We define the disutility of the agents as follows: For each literal $\ell_i$, if
\begin{itemize}
	\item $\ell_i = x_i$, then,
	             \begin{align*}
	                d(n_i^r,b_2^i) &=1 &\text{and}&  & d(n_i^r,m_i^r) &=\varepsilon\\
	                d(\n^r,b_2^i) &= 1 &\text{and}&  & d(\n^r,m_i^r) &= \varepsilon.      
	             \end{align*} 
	            for some $0 < \varepsilon \ll 1 $, but $\tfrac{1}{\varepsilon} \in \mathcal{O}(1)$.  
	 \item $\ell_i = \neg x_i$, then,
				 \begin{align*}
					d(n_i^r,b_1^i) &=\tfrac{2}{3} &\text{and}& & d(n_i^r,m_i^r) &=\tfrac{4\varepsilon}{3}\\
					d(\n^r,b_1^i) &= \tfrac{2}{3} &\text{and}& & d(\n^r,m_i^r) &= \tfrac{4\varepsilon}{3}.
				\end{align*}             
\end{itemize} 

For all other agents and chores pair, the disutility is $\tau$. See Figure~\ref{nphardfig} for an illustration. We set $e(n_i^r) = e(n_j^r) = e(n_k^r) = \varepsilon$ and $e(\n^r) = \#(C_r) \cdot (\tfrac{\varepsilon}{2}) + \overline{\#}(C_r) \cdot (\varepsilon) - \varepsilon'$, where $\#(C_r)$ is the number of literals in $C_r$ that are not negations of variables and $\overline{\#}(C_r)$ is the number of literals in $C_r$ that are negations of variables \footnote{This implies that $\#(C_r) + \overline{\#}(C_r) = 3$}, and $\varepsilon' < \tfrac{\varepsilon}{2}$ (the exact value of $\varepsilon'$ will depend on $\delta$ \footnote{Reminder to what $\delta$ is: recall that we are trying to show the hardness of determining whether an instance admits a $(\tfrac{11}{12}+ \delta)$-competitive equilibrium or not.} and will be made clear in the proof of Lemma~\ref{negative}). We make a small claim about the total earning requirements for the agents $n_i^r$, $n_j^r$, $n_k^r$ and $\n^r$.\

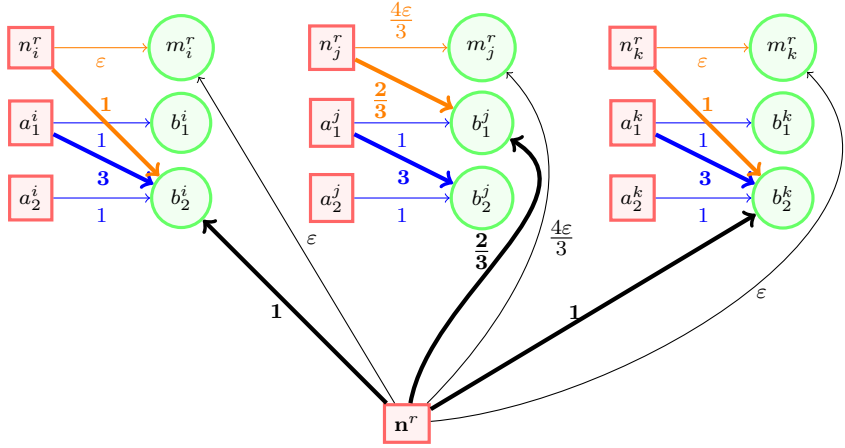
\begin{figure}
	\centering 
 \begin{tikzpicture}[
 roundnode/.style={circle, draw=green!60, fill=green!5, very thick, minimum size=7mm},
 squarednode/.style={rectangle, draw=red!60, fill=red!5, very thick, minimum size=5mm},
 ]
 \node[squarednode]      (ai1)      at (0,3)                       {$\scriptstyle{a^i_{1}}$};
 \node[squarednode]      (ai2)      at (0,2)                       {$\scriptstyle{a^i_{2}}$};
 \node[roundnode]      (bi1)      at (2,3)                       {$\scriptstyle{b^i_{1}}$};
 \node[roundnode]      (bi2)      at (2,2)                       {$\scriptstyle{b^i_{2}}$};		
 	
 \draw[blue,->] (ai1)--node[below]{$\scriptstyle{1}$}(bi1); 
 \draw[blue,->] (ai2)--node[below]{$\scriptstyle{1}$}(bi2);
 \draw[blue,ultra thick, ->] (ai1)--node[below]{$\mathbf{\scriptstyle{3}}$}(bi2);

 \node[squarednode]      (aj1)      at (0+4,3)                       {$\scriptstyle{a^j_{1}}$};
 	\node[squarednode]      (aj2)      at (0+4,2)                       {$\scriptstyle{a^j_{2}}$};
 		\node[roundnode]      (bj1)      at (2+4,3)                       {$\scriptstyle{b^j_{1}}$};
 			\node[roundnode]      (bj2)      at (2+4,2)                       {$\scriptstyle{b^j_{2}}$};		
 				
 				\draw[blue,->] (aj1)--node[below]{$\scriptstyle{1}$}(bj1); 
 				\draw[blue,->] (aj2)--node[below]{$\scriptstyle{1}$}(bj2);
 				\draw[blue,ultra thick, ->] (aj1)--node[below]{$\mathbf{\scriptstyle{3}}$}(bj2);

 \node[squarednode]      (ak1)      at (0+8,3)                       {$\scriptstyle{a^k_{1}}$};
 	\node[squarednode]      (ak2)      at (0+8,2)                       {$\scriptstyle{a^k_{2}}$};
 		\node[roundnode]      (bk1)      at (2+8,3)                       {$\scriptstyle{b^k_{1}}$};
 			\node[roundnode]      (bk2)      at (2+8,2)                       {$\scriptstyle{b^k_{2}}$};		
 				
 				\draw[blue,->] (ak1)--node[below]{$\scriptstyle{1}$}(bk1); 
 				\draw[blue,->] (ak2)--node[below]{$\scriptstyle{1}$}(bk2);
 				\draw[blue,ultra thick, ->] (ak1)--node[below]{$\mathbf{\scriptstyle{3}}$}(bk2);

\node[squarednode]      (ni)      at (0,4)                       {$\scriptstyle{n^r_{i}}$};
\node[squarednode]      (nj)      at (0+4,4)                       {$\scriptstyle{n^r_{j}}$};
\node[squarednode]      (nk)      at (0+8,4)                       {$\scriptstyle{n^r_{k}}$};

\node[roundnode]      (mi)      at (0+2,4)                       {$\scriptstyle{m^r_{i}}$};
\node[roundnode]      (mj)      at (0+4+2,4)                       {$\scriptstyle{m^r_{j}}$};
\node[roundnode]      (mk)      at (0+8+2,4)                       {$\scriptstyle{m^r_{k}}$};

\draw[orange,->] (ni)--node[below]{$\mathbf{\scriptstyle{\varepsilon}}$}(mi); 
\draw[orange,->] (nj)--node[above]{$\scriptstyle{\tfrac{4\varepsilon}{3}}$}(mj); 
\draw[orange,->] (nk)--node[below]{$\scriptstyle{\varepsilon}$}(mk);

\draw[orange,->, ultra thick] (ni)--node[above]{$\mathbf{\scriptstyle{1}}$}(bi2);  
\draw[orange,->, ultra thick] (nj)--(bj1);
\node  at (4.65,3.3) {\textcolor{orange}{$\mathbf{\scriptstyle{\tfrac{2}{3}}}$}};
\draw[orange,->, ultra thick] (nk)--node[above]{$\mathbf{\scriptstyle{1}}$}(bk2);

\node[squarednode]      (nr)      at (5,-1)                       {$\scriptstyle{\n^r}$};

\draw[black,->] (nr)--node[right]{${\scriptstyle{\varepsilon}}$} (mi);
\path[black,->, out=45,in=-45]    (nr) edge node[right]{${\scriptstyle{\tfrac{4\varepsilon}{3}}}$} (mj);
\path[black,->, out=5,in=-45]    (nr) edge node[right]{${\scriptstyle{\varepsilon}}$} (mk);

\draw[black,ultra thick,->] (nr)--node[left]{$\mathbf{\scriptstyle{1}}$}(bi2);
\path[black,ultra thick, ->, out=80,in=-30]    (nr) edge node[left]{$\mathbf{\scriptstyle{\tfrac{2}{3}}}$} (bj1);
\draw[black,ultra thick,->] (nr)--node[left]{$\mathbf{\scriptstyle{1}}$}(bk2);

 \end{tikzpicture}
 \caption{Illustration of the variable gadgets corresponding to $x_i$, $x_j$ and $x_k$, and the clause gadget $C_r = (x_i \vee \neg x_j \vee x_k)$. The red squared nodes represent the agents and the green circle nodes represent the chores. Only disutility values less than $\tau$ have been indicated. The disutility edges from agents in the variable gadgets are outlined by blue edges. The disutility edges from agents $n^r_{\ell}$ for $\ell \in \left\{i,j,k\right\}$ are outlined by orange edges and the disutility edges from agent $\n^r$  are outlined by black edges. Thicker disutility edges have a higher disutility than the thinner disutility edges of the same color.}
 \label{nphardfig}
\end{figure}

\begin{claim}
	\label{sumofbudgets}
	For each clause $C_r = \ell_i \vee \ell_j \vee \ell_k$ in $I$, we have $e(n_i^r) + e(n_j^r) + e(n_k^r) + e(\n^r) = \#(C_r) \cdot (\tfrac{3\varepsilon}{2}) + \overline{\#}(C_r) \cdot (2 \varepsilon) - \varepsilon'$ 
\end{claim}

\begin{proof}
	We have,
	\begin{align*}
	 e(n_i^r) + e(n_j^r) + e(n_k^r) + e(\n^r) &= 3 \varepsilon + \#(C_r) \cdot (\tfrac{\varepsilon}{2}) + \overline{\#}(C_r) \cdot (\varepsilon) - \varepsilon'\\ 
	                                          &= (\#(C_r) + \overline{\#}(C_r)) \varepsilon + \#(C_r) \cdot (\tfrac{\varepsilon}{2}) + \overline{\#}(C_r) \cdot (\varepsilon) - \varepsilon'\\
	                                          & = \#(C_r) \cdot (\tfrac{3\varepsilon}{2}) + \overline{\#}(C_r) \cdot (2\varepsilon) - \varepsilon' \qedhere
	\end{align*}
\end{proof}

We now show how to map any allocation in $E(I)$ to an assignment of variables in $I$. Consider any \emph{earning} $f$ under some prices $p$ in $E(I)$. If agent $i$ does $X_{ij}$ amount of chore $j$, then $f(i,j) = X_{ij}\cdot p(j)$. 
 \begin{center}
 	\emph{If agent $a_1^i$ does some of chore $b_2^i$, i.e., $f({a_1^i, b_2^i}) > 0$, then we set $x_i$ to $F$ and if $f({a_1^i, b_2^i}) = 0$, then we set $x_i$ to $T$.}
 \end{center}
 We now make some basic observations.

\begin{observation}\label{technical}
	Let $p$ be the prices of chores and $f$ the the money allocation corresponding to a competitive equilibrium in $E(I)$. Consider any clause $C_{r} = (\ell_i \vee \ell_j \vee \ell_k )$. Then, 
	\begin{enumerate}
		\item if $\ell_i = x_i$ and $f({a_1^i, b_2^i}) > 0$  then $p({m_i^r}) \geq \tfrac{3\varepsilon}{2}$, and
		\item if $\ell_i = \neg x_i$ and $f({a_1^i, b_2^i}) = 0$, then $p({m_i^r}) \geq 2\varepsilon$.
	\end{enumerate}
\end{observation}

\begin{proof}
	We first prove part 1. If $f({a_1^i, b_2^i}) > 0$, then $b_2^i \in MPB_{a_1^i}$, implying that $\tfrac{d(a_1^i,b_2^i)}{p(b_2^i)} \leq \tfrac{d(a_1^i,b_1^i)}{p(b_1^i)}$. Therefore, we have that $p(b_2^i) \geq \tfrac{d(a^i_1,b^i_2)}{d(a^i_1,b^i_1)} \cdot p(b^i_1) = 3p(b_1^i)$. Also, note that since agents $a_1^i$ and $a_2^i$ have disutility less than $\tau$ only for chores $b_1^i$ and $b_2^i$, they will only earn from chores $b_1^i$ and $b_2^i$. This implies that $p(b_1^i)+ p(b_2^i) \geq e(a_1^i) + e(a_2^i) = 2$. Also, since $p(b_2^i) \geq 3p(b_1^i)$ we have that $p(b_2^i) \geq \tfrac{3}{2}$. Now, observe that the only agents who have disutility less than $\tau$ towards $m_i^r$ are the agents $n_i^r$ and $\n^r$. Since $\ell_i = x_i$, both $n_i^r$ and $\n^r$ have a disutility of 1 towards $b_2^i$ and $\varepsilon$ towards $m_i^r$. Therefore, for $m_i^r$ to be in either $\MPB_{n_i^r}$ or $\MPB_{\n^r}$, we need $\tfrac{\varepsilon}{p(m_i^r)} \leq \tfrac{1}{p(b_2^i)} \leq \tfrac{2}{3}$. This implies that $p(m_i^r) \geq \tfrac{3\varepsilon}{2}$. 
	
	The proof of part 2 is very similar. Note that agent $a_1^i$ has disutility less than $\tau$ for only chores $b_1^i$ and $b_2^i$. If $f(a_1^i,b_2^i)=0$, then she only earns by doing chore $b_1^i$, implying that $p(b_1^i) \geq e(a_1^i) = 1$. Similar to the proof in part 1, observe that the only agents who have disutility less than $\tau$ towards $m_i^r$ are the agents $n_i^r$ and $\n^r$. Since $\ell_i = \neg x_i$, both $n_i^r$ and $\n^r$ have a disutility of $\tfrac{2}{3}$ towards $b_1^i$ and $\tfrac{4\varepsilon}{3}$ towards $m_i^r$. Therefore, for $m_i^r$ to be in either $\MPB_{n_i^r}$ or $\MPB_{\n^r}$, we need $\tfrac{4\varepsilon}{3p(m_i^r)} \leq \tfrac{2}{3p(b_1^i)} \leq \tfrac{2}{3}$ (as $p(b^1_i) \geq 1$). This implies that $p(m_i^r) \geq 2\varepsilon$.\qedhere  
\end{proof}

\begin{lemma}
	\label{negative}
	If there is no satisfying assignment to the instance $I = \langle X, \C \rangle $ of 3-SAT, then $E(I)$ does not admit any $(\tfrac{11}{12} + \delta)$-competitive equilibrium for any $\delta>0$.
\end{lemma}

\begin{proof}
 We prove by contradiction. Assume otherwise and let $p$ be the equilibrium prices of chores and $f$ be the corresponding money allocation. Recall the mapping from an equilibrium allocation to the assignment of variables: For each $i \in  [n]$ if $f(a_1^i,b_2^i)>0$, then we set $x_i$ to $F$ and if $f(a_1^i,b_2^i) = 0$, then we set $x_i$ to $T$. Since $I$ admits no satisfying assignment, there exists a clause $C_r = \ell_i \vee \ell_j \vee \ell_k$ which is unsatisfied. For every literal $\ell_i \in C_r$ such that $\ell_i = x_i$, note that $x_i$ is $F$. Therefore, we have that $f(a_1^i,b_2^i)>0$. This implies that $p(m_i^r) \geq \tfrac{3 \varepsilon}{2}$ (by Observation~\ref{technical}). Similarly for every literal $\ell_i$ in $C_r$ such that $\ell_i = \neg x_i$, note that $x_i$ is $T$. Therefore, we have that $f(a_1^i,b_2^i) = 0$, implying that $p(m_i^r) \geq 2\varepsilon$ (by Observation~\ref{technical}). We write the price of chore $m^r_i$, $p(m^r_i)$ as $\tfrac{3\varepsilon}{2} + \delta(m^r_i)$ if $\ell_i = x_i$ and $2\varepsilon + \delta(m^r_i)$ if $\ell_i = \neg x_i$, where $\delta(m^r_i)$ is the deviation of the price of $m^r_i$ from its lower bound. Therefore, we have $p(m_i^r) + p(m_j^r) + p(m_k^r) = \#(C_r) \cdot (\tfrac{3\varepsilon}{2}) + \overline{\#}(C_r) \cdot (2 \varepsilon) + \delta(m^r_i) + \delta(m^r_j) + \delta(m^r_k)$. Note that the only agents who have disutility less than $\tau$ for chores $m_i^r$, $m_j^r$ and $m_k^r$ are the agents $n_i^r$, $n_j^r$, $n_k^r$ and $\n^r$. However, by Claim~\ref{sumofbudgets}, we have that $e(n_i^r) + e(n_j^r) + e(n_k^r) + e(\n^r) = \#(C_r) \cdot (\tfrac{3\varepsilon}{2}) + \overline{\#}(C_r) \cdot (2 \varepsilon) - \varepsilon' $ which is strictly less that the sum of prices of chores $m_i^r$, $m_j^r$ and $m_k^r$. In particular we have, $\sum_{h \in \left\{i,j,k \right\}} p(m^r_h) - (\sum_{h \in \left\{i,j,k \right\}} e(n^r_h) + e(\n^r)) = \varepsilon' + \sum_{h \in \left\{i,j,k \right\}} \delta(m^r_h)$. Therefore, there exists at least one chore $m^r_{h'}$ such that the difference between the total price of the chore and the total money earned from the chore by the agents is $\tfrac{\varepsilon' + \sum_{h \in \left\{i,j,k \right\}} \delta(m^r_h)}{3} \geq \tfrac{\varepsilon' +  \delta(m^r_{h'})}{3}$. Thus, the portion of chore $m^r_{h'}$ left undone is at least,
 \begin{align*}
  &=\tfrac{\varepsilon' +  \delta(m^r_{h'})}{3 \cdot p(m^r_{h'})}\\
  &\geq \tfrac{\varepsilon' +  \delta(m^r_{h'})}{3 \cdot (2 \varepsilon + \delta(m^r_{h'}))} &\text{(as $p(m^r_{h'})$ is either $\tfrac{3 \varepsilon}{2} + \delta(m^r_{h'})$ or $2\varepsilon + \delta(m^r_{h'})$)}\\
  &\geq \tfrac{\varepsilon'}{3 \cdot (2 \varepsilon)} &\text{(as $\varepsilon' < \frac{\varepsilon}{2}$)}.
 \end{align*} 
  Since our reduction works for any choice of $\varepsilon' < \tfrac{\varepsilon}{2}$, we can choose an $\varepsilon'$ such that $\tfrac{\varepsilon'}{ (6 \varepsilon)} > \tfrac{1}{12} - \delta$, implying that we do not have a $(\tfrac{11}{12} + \delta)$-competitive equilibrium, which is a contradiction.
\end{proof}

\begin{lemma}
	\label{positive}
	If there exists a satisfying assignment to the instance $I = \langle X, \C \rangle $ of 3-SAT, then $E(I)$ admits a competitive equilibrium.
\end{lemma}

\begin{proof}
	Consider any satisfying assignment in $I$. We now show how to construct the prices $p$ and the money allocation $f$ corresponding to a competitive allocation. We will ensure that only the agents in the variable gadgets earn from the chores in the variable gadgets and  the agents in the clause gadgets earn only from the chores in the clause gadgets. 
	
	\paragraph{Prices and Allocation of Chores in Variable Gadgets.}  For each variable $x_i$, 
	\begin{itemize}
		\item If $x_i = T$, then we set  $p(b_1^i)=1$ and $p(b_2^i) = 1$.		              
		\item If $x_i = F$, then we set $p(b_1^i)=\tfrac{1}{2}$ and $p(b_2^i) = \tfrac{3}{2}$.
	\end{itemize}
	Since the agents in the variable gadgets have disutility less than $\tau$ only for some goods in the variable gadgets (and have disutility of $\tau$ for every good in the clause gadget) we can already define their optimal bundles ($\MPB$ bundles). If $x_i = T$, then observe that  $\MPB_{a_1^i} = \left\{b_1^i\right\}$ and $\MPB_{a_2^i} = \left\{b_2^i\right\}$. Thus agent $a_1^i$ earns $1$ unit of money from doing chore $b_1^i$ entirely and agent $a_2^i$ earns $1$ unit of money by doing chore $b_2^i$ entirely. When $x_i = F$, then observe that  $\MPB_{a_1^i} = \left\{b_1^i,b_2^i\right\}$ and $\MPB_{a_2^i} = \left\{b_2^i\right\}$. Thus agent $a_1^i$ earns $1$ unit of money from doing chore $b_1^i$ entirely and $b_2^i$ partly and agent $a_2^i$ earns $1$ unit of money by doing chore $b_2^i$ partly. Now we make an immediate, simple  observation:
	
	\begin{observation}
		When $x_i=T$, then $f(a_1^i,b_2^i) = 0$ and when $x_i = F$, we have $f(a_1^i,b_2^i) > 0$.
	\end{observation}  
	Observe that all the local sub-instances corresponding to the variable gadgets have cleared. It suffices to show that there exists a competitive equilibrium for local sub-instances corresponding to the clause gadgets. We now look into the agents and chores in the clause gadget.
	\paragraph{Prices and Allocation of Chores in Clause Gadgets.} Consider a clause $C_r = \ell_i \vee \ell_j \vee \ell_k$. Therefore, let $S_r \subseteq \left\{\ell_i,\ell_j,\ell_k \right\}$ be the literals that evaluate to $T$ \footnote{A literal $\ell_i = x_i$ evaluates to $T$ if $x_i$ is set to $T$ and the literal $\ell_i = \neg x_i$ evaluates to $T$ when $x_i$ is set to $F$.} and $U_r \subseteq \left\{\ell_i,\ell_j,\ell_k \right\}$ be the set of literals that evaluate to $F$ under the assignment $X$. Since $X$ is a satisfying assignment, at least one of the literals will evaluate to $T$ and thus $\lvert S_r \rvert  \geq 1$ and $\lvert U_r \rvert \leq 2$. Let $\#(S_r)$ and $\#(U_r)$ be the number of literals in $S_r$ and $U_r$ respectively that are not negations of variables and similarly let $\overline{\#}(S_r)$ and $\overline{\#}(U_r)$ be the number of literals that are negations of variables in $S_r$ and $U_r$ respectively. Let $\alpha_r$ be a scalar such that 
	\begin{align}
	\label{alphadef}
	 \alpha_r \cdot ( \#(U_r) \cdot \tfrac{3\varepsilon}{2} + \overline{\#}(U_r) \cdot (2\varepsilon) ) = \lvert U_r \rvert \cdot \varepsilon + e(\n^r) 
	\end{align} 
	We now set the prices of the chores in the clause gadgets. Consider any clause $C_r = \ell_i \vee \ell_j \vee \ell_k$ in $I$ (with $S_r$ and $U_r$ defined appropriately). For every literal $\ell_{\theta} \in S_r$, set,
	
	\[ p(m_{\theta}^r) =   \left\{
	\begin{array}{ll}
	\varepsilon & \text{ if } \ell_{\theta} = \neg x_{\theta},  \\
	\varepsilon & \text{ if } \ell_{\theta} = x_{\theta} \text{ and } U_r \neq \emptyset,  \\
	\varepsilon + \frac{e(n^r)}{\#(S_r)}  &\text{ if } \ell_{\theta} =  x_{\theta} \text{ and } U_r = \emptyset. \\
	\end{array} 
	\right. \]
	
	For every $\ell_{\theta} \in U_r$, set 
	
	\[ p(m_{\theta}^r) =   \left\{
	\begin{array}{ll}
	\alpha_r \cdot (\frac{3\varepsilon}{2}) & \text{ if } \ell_{\theta} = x_{\theta}  \\
	\alpha_r \cdot(2 \varepsilon)  &\text{ if } \ell_{\theta} = \neg x_{\theta}. \\
	\end{array} 
	\right. \]

	We will now show that under the above prices for the chores in the clause gadgets, we can determine a money flow where all the clause agents earn all of their money from their optimal bundles and all the clause chores will be completed. We distinguish two cases, depending on whether $U_r = \emptyset$ or not,
	
	\paragraph{Case $U_r \neq \emptyset$:} In this case, we first observe that $\alpha_r$ is strictly larger than $1$:
	\begin{observation}
		\label{alphaislarge}
		We have well defined scalar $\alpha_r > 1$.
	\end{observation}
	 \begin{proof}
	 	Since we are in the case where $U_r \neq 0$, we have $\#(U_r) \cdot \tfrac{3\varepsilon}{2} + \overline{\#}(U_r) \cdot (2\varepsilon) > 0 $, thus $\alpha_r$ is well defined. For the claim of the lemma, it suffices to show that $\lvert U_r \rvert \cdot \varepsilon + e(\n^r) > \#(U_r) \cdot \tfrac{3\varepsilon}{2} + \overline{\#}(U_r) \cdot (2\varepsilon)$. To this end,
	 	\begin{align}\label{alphalargeineq}  
	 	  \lvert U_r \rvert \cdot \varepsilon + e(\n^r) &= (\#(U_r) + \overline{\#}(U_r)) \cdot \varepsilon + e(\n^r) \nonumber \\
	 	                                                &= (\#(U_r) + \overline{\#}(U_r)) \cdot \varepsilon + \#(C_r) \cdot \tfrac{\varepsilon}{2} + \overline{\#}(C_r) \cdot (\varepsilon) - \varepsilon' \enspace . 
	 	 \end{align}                                               
	 	  Since the literals that are not negations of variables in $U_r$ are also not negations of variables in $C_r$ we have $\#(U_r) \leq \#(C_r)$. By a similar argument we also have $\overline{\#}(U_r) \leq \overline{\#}(C_r)$. Since $\lvert U_r \rvert \leq 2$ we also have $\#(U_r) + \overline{\#}(U_r) < \#(C_r) + \overline{\#}(C_r) $, implying that either $\#(U_r) < \#(C_r)$ or $\overline{\#}(U_r) < \overline{\#}(C_r)$. Therefore, we have that $\#(C_r) \cdot \tfrac{\varepsilon}{2} + \overline{\#}(C_r) \cdot (\varepsilon) \geq \#(U_r) \cdot \tfrac{\varepsilon}{2} + \overline{\#}(U_r) \cdot (\varepsilon) + \tfrac{\varepsilon}{2}$. Plugging this inequality in~(\ref{alphalargeineq}), we have                                               
	 	  \begin{align*}                                              
		 	   \lvert U_r \rvert \cdot \varepsilon + e(\n^r) &\geq (\#(U_r) + \overline{\#}(U_r)) \cdot \varepsilon + \#(U_r) \cdot \tfrac{\varepsilon}{2} + \overline{\#}(U_r) \cdot (\varepsilon) + \tfrac{\varepsilon}{2} - \varepsilon'\\
	 	                                                &> (\#(U_r) + \overline{\#}(U_r)) \cdot \varepsilon + \#(U_r) \cdot \tfrac{\varepsilon}{2} + \overline{\#}(U_r) \cdot (\varepsilon) &(\text{as  $\varepsilon' < \tfrac{\varepsilon}{2}$})\\
	 	                                                &= \#(U_r) \cdot \tfrac{3\varepsilon}{2} + \overline{\#}(U_r) \cdot (2\varepsilon) \enspace . \qedhere
	 	\end{align*}
	 \end{proof}

	We will now characterize the optimal bundles ($\MPB$ chores) for each agents under the set prices. 
		
	\begin{observation}
		\label{MPB_for_S_r}
		For each literal $\ell_{\theta} \in S_r$, we have $m_{\theta}^r \in \MPB_{n_{\theta}^r}$.
	\end{observation}
	\begin{proof}
		We consider the cases, whether the $\ell_{\theta} = x_{\theta}$ or $\ell_{\theta} = \neg x_{\theta}$.
		\begin{itemize}
			\item $\ell_{\theta} = x_{\theta}$: Note that the only other chore (other than $m_{\theta}^r$) for which agent $n_{\theta}^r$ has disutility less than $\tau$ is chore $b_2^{\theta}$. Since $\ell_{\theta} \in S_r$, this means that $x_{\theta} = T$ and therefore we have $p(b_2^{\theta}) = 1$. Now observe that,
			\begin{align*}
			  \frac{d(n_{\theta}^r,m_{\theta}^r)}{p(m_{\theta}^r)} &= \frac{\varepsilon}{\varepsilon}\\
			                                                        &= 1\\
			                                                        &=\frac{d(n_{\theta}^r,b_2^{\theta})}{p(b_2^{\theta})} \enspace .  
			\end{align*}
			Therefore $m_{\theta}^r \in \MPB_{n_{\theta}^r}$.
			\item $\ell_{\theta}  = \neg x_{\theta}$: Note that the only other chore (other than $m_{\theta}^r$) for which agent $n_{\theta}^r$ has disutility less than $\tau$ is chore $b_1^{\theta}$. Since $\ell_{\theta} \in S_r$, this means that $x_{\theta} = F$ and therefore we have $p(b_1^{\theta}) = \tfrac{1}{2}$. Now observe that,
			\begin{align*}
			\frac{d(n_{\theta}^r,m_{\theta}^r)}{p(m_{\theta}^r)} &= \frac{4\varepsilon}{3\varepsilon}\\
																 &= \frac{4}{3}\\
																 &= \frac{2}{3 \cdot \tfrac{1}{2}}\\
																 &=\frac{d(n_{\theta}^r,b_1^{\theta})}{p(b_1^{\theta})}  \enspace .
			\end{align*} 
			Therefore, $m_{\theta}^r \in \MPB_{n_{\theta}^r}$. \qedhere
		\end{itemize}
	\end{proof}
	This implies that for all literals $\ell_{\theta}$ in $S_r$, the agent $n_{\theta}^r$ will earn her entire money of $\varepsilon$ by doing the chore $\ell_{\theta}$ entirely. Therefore, now we only need to look at the agents $n_{\theta}^r$ and chores $m_{\theta}^r$ where $\ell_{\theta} \in U_r$. To this end we observe that,
	
	\begin{observation}
		\label{MPB_for_U_r}
		For each literal $\ell_{\theta} \in U_r$, we have $m_{\theta}^r \in \MPB_{n_{\theta}^r}$ and $m_{\theta}^r \in \MPB_{\n^r}$.
	\end{observation}

	\begin{proof}
     We first show that $m_{\theta}^r \in \MPB_{n_{\theta}^r}$. We make a distinction based on whether $\ell_{\theta} = x_{\theta}$ or $\ell_{\theta} = \neg x_{\theta}$. 
       
      \begin{itemize}
      	\item $\ell_{\theta} = x_{\theta}$: In this case we have $p(m_{\theta}^r) = \alpha_r \cdot (\tfrac{3\varepsilon}{2})$. Note that the only other chore (other than $m_{\theta}^r$) for which agent $n_{\theta}^r$ has disutility less than $\tau$ is chore $b_2^{\theta}$. Since $\ell_{\theta} \in U_r$, this means that $x_{\theta} = F$ and therefore we have $p(b_2^{\theta}) = \tfrac{3}{2}$. Now observe that,

      	 \begin{align}
      	  \frac{d(n_{\theta}^r,m_{\theta}^r)}{p(m_{\theta}^r)} &= \frac{1}{\alpha_r} \cdot \frac{\varepsilon}{\tfrac{3\varepsilon}{2}} \nonumber\\
														       &= \frac{1}{\alpha_r} \cdot \frac{2}{3} \label{hard1}\\
														       &= \frac{1}{\alpha_r} \cdot \frac{d(n_{\theta}^r,b_2^{\theta})}{p(b_2^{\theta})} \nonumber\\
														       &<  \frac{d(n_{\theta}^r,b_2^{\theta})}{p(b_2^{\theta})}\enspace . &(\text{as $\alpha_r >1$ by Observation~\ref{alphaislarge}}) \nonumber
      	 \end{align}

      	\item $\ell_{\theta} = \neg x_{\theta}$: In this case we have $p(m_{\theta}^r) = \alpha_r \cdot (2\varepsilon)$. Note that the only other chore (other than $m_{\theta}^r$) for which agent $n_{\theta}^r$ has disutility less than $\tau$ is chore $b_1^{\theta}$. Since $\ell_{\theta} \in U_r$, this means that $x_{\theta} = T$ and therefore we have $p(b_1^{\theta}) = 1$. Now observe that,
      	\begin{align}
      	\frac{d(n_{\theta}^r,m_{\theta}^r)}{p(m_{\theta}^r)} &= \frac{1}{\alpha_r} \cdot \frac{4\varepsilon}{3 \cdot 2\varepsilon} \nonumber\\
														     &= \frac{1}{\alpha_r} \cdot \frac{2}{3} \label{hard2} \\
														     &= \frac{1}{\alpha_r} \cdot \frac{d(n_{\theta}^r,b_1^{\theta})}{p(b_1^{\theta})} \nonumber\\
														     &<  \frac{d(n_{\theta}^r,b_1^{\theta})}{p(b_1^{\theta})}\enspace . &(\text{as $\alpha_r >1$ by Observation~\ref{alphaislarge}}) \nonumber
      	\end{align}
     
      \end{itemize}  
      Thus in both cases we  have $m_{\theta}^r \in \MPB_{n_{\theta}^r}$.
      
      We will now show that $m_{\theta}^r \in \MPB_{\n^r}$ as well. We do this by showing that the disutility to price ratio of the chores $m_{\theta}^r$, when $\ell_{\theta} \in U_r$, is minimum for the agent $\n^r$. \emph{To this end, first crucially observe that from~\eqref{hard1} and~\eqref{hard2}, irrespective of whether $\ell_{\theta} = x_{\theta}$ or $\ell_{\theta} = \neg x_{\theta}$, we have $\frac{d(n_{\theta}^r,m_{\theta}^r)}{p(m_{\theta}^r)} = \tfrac{1}{\alpha_r} \cdot \tfrac{2}{3}$}.  Also, note that the disutility profile agent $\n^r$ has for chore $m_{\theta}^r$ and the chores in the variable gadget of $x_{\theta}$ ($b^{\theta}_1$ and $b^{\theta}_2$) is identical to the disutility profile of agent $n_{\theta}^r$ for the same set of chores. Therefore, for all $\ell_{\theta} \in U_r$ we have $\frac{d(\n^r,m_{\theta}^r)}{p(m_{\theta}^r)} = \frac{1}{\alpha_r} \cdot \frac{2}{3}$ (irrespective of whether $\ell_{\theta} = x_{\theta}$ or $\ell_{\theta} = \neg x_{\theta}$) which is also strictly less than both $\frac{d(\n^r,b_2^{\theta})}{p(b_2^{\theta})}$ and $\frac{d(\n^r,b_1^{\theta})}{p(b_1^{\theta})}$. 
      We now look at disutility to price ratio that agent $\n^r$ has for chores in $S_r$. Observe that for all $\ell_{\beta} \in S_r$ we have $p(m_{\beta}^r) = \varepsilon$ and $d(\n^r,m_{\beta}^r) \geq \varepsilon$ (as the disutility is $\varepsilon$ if $\ell_{\beta} = x_{\beta}$ and is $\tfrac{4\varepsilon}{3}$ if $\ell_{\beta} = \neg x_{\beta}$ ). This implies that for all $\ell_{\beta} \in S_r$ we have $\frac{d(\n^r,m_{\beta}^r)}{p(m_{\beta}^r)} \geq 1 > \tfrac{2}{3} > \tfrac{1}{\alpha_r} \cdot \tfrac{2}{3}$ (as $\alpha_r >1$ by Observation~\ref{alphaislarge}). Therefore, the disutility to price ratio of the chores $m^r_{\theta}$, when $\ell_{\theta} \in U_r$,  for agent $\n^r$ is $\tfrac{1}{\alpha_r} \cdot \tfrac{2}{3}$ which is at most the disutility to price ratio of all the chores for which $\n^r$ has a disutility of less than $\tau$. Therefore, we have $ \bigcup_{\ell_{\theta} \in U_r} m_{\theta}^r \subseteq \MPB_{\n^r}$. 
	\end{proof}
	 Now that we have identified the $\MPB$ chores for all the agents in the clause gadgets, we are ready to show the money flow allocation. We set 
	 \begin{align*}
	   f(n_{\theta}^r,m_{\theta}^r) &= \varepsilon           &\text{ (for all $\ell_{\theta} \in S_r$) }\\
	   f(\n^r,m_{\theta}^r) &= p(m_{\theta}^r) - \varepsilon \enspace . &\text{ (for all $\ell_{\theta} \in U_r$) } 
	 \end{align*} 
    All agents spend on their corresponding $\MPB$ chores. Observe that for all $\ell_{\theta} \in S_r$, the agents $n_{\theta}^r$ earn their money of $\varepsilon$ by doing chore $m_{\theta}^r$ completely. Now, for all $\ell_{\theta} \in U_r$, the agents $n_{\theta}^r$ earn their money of $\varepsilon$ by doing chore $m_{\theta}^r$ partially. The agent $\n^r$ earns her entire money by completing whatever is left of the chores in $\bigcup_{\ell_{\theta} \in U_r} m_{\theta}^r$. It only suffices to show that agent $\n^r$ earns exactly $e(\n^r)$. To this end, we observe that the total money earned by $\n^r$ is 
    
    \begin{align*}
     \sum_{\ell_{\theta} \in U_r} f(\n^r,m_{\theta}^r) &= \sum_{\ell_{\theta} \in U_r}(p(m_{\theta}^r) -\varepsilon)\\      
                                                       &= \alpha_r \cdot ( \#(U_r) \cdot \tfrac{3\varepsilon}{2} + \overline{\#}(U_r) \cdot (2\varepsilon) ) - \lvert U_r \rvert \cdot \varepsilon \\
                                                       &= e(\n^r)  \enspace . &(\text{by~\eqref{alphadef})}
                                                    \end{align*}
       
   Therefore, we have an allocation where the agents in the corresponding variable gadgets earn their money by completing the chores in the variable gadgets and the agents in the clause gadget earn their entire money by completing the chores in the clause gadgets. This concludes the proof for the case $U_r \neq \emptyset$.
  
  \paragraph{Case $U_r = \emptyset$:}  In this case we have that all the literals in the clause $C_r$ belongs to the set $S_r$. Therefore, for all the literals $\ell_{\theta}$ occurring in $C_r$, we have,
  \[ p(m_{\theta}^r) =   \left\{
  \begin{array}{ll}
  \varepsilon & \text{ if } \ell_{\theta} = \neg x_{\theta},  \\
  \varepsilon + \frac{e(n^r)}{\#(S_r)}  &\text{ if } \ell_{\theta} =  x_{\theta} \\
  \end{array} 
  \right. \]
  Like earlier, we will identify the $\MPB$ chores for all the clause gadget agents and then will outline a money flow allocation where every agent earns all her money and all the chores are completed. We first look into the agents $n^r_{\theta}$. Very similar to Observation~\ref{MPB_for_S_r}, we can claim that $m^r_{\theta} \in \MPB_{n^r_{\theta}}$ with a very similar argument as the one used in the proof of Observation~\ref{MPB_for_S_r}: The agent $n^r_{\theta}$ has disutility less than $\tau$ only for chores $m^r_{\theta}$, $b^{\theta}_{2}$ if $\ell_{\theta} = x_{\theta}$, and only for chores $m^r_{\theta}$ and $b^{\theta}_{1}$ if $\ell_{\theta} = \neg x_{\theta}$,  and the price of the chore $p(m^r_{\theta})$ is at least $\varepsilon$ (it is more if $\ell_{\theta} = x_{\theta}$), while the prices of chores $b^1_{\theta}$ and $b^2_{\theta}$ are the same as in Observation~\ref{MPB_for_S_r}. 
  
  Now we look into the agent $\n^r$. Since the disutility profile of agent $\n^r$ is identical to that of $n^r_{\theta}$ when restricted to chores $b^{\theta}_1$, $b^{\theta}_2$ and $m^r_{\theta}$, we can conclude that the disutility to price ratio of $m^r_{\theta}$ for $\n^r$  is at most that of chores $b^{\theta}_1$ and $b^{\theta}_2$. Now observe that the disutility to price ratio of all chores $m^r_{\theta}$ for $\n^r$ where $\ell_{\theta} = x_{\theta}$ is $\tfrac{d(\n^r,m^r_{\theta})}{p(m^r_{\theta})} = \tfrac{\varepsilon}{p(m^r_{\theta})} \leq 1$ (as $p(m^r_{\theta}) = \varepsilon + \tfrac{e(\n^r)}{\#(S_r)}$), while the disutility to price ratio all chores $m^r_{\theta}$ for $\n^r$ where $\ell_{\theta} = \neg x_{\theta}$ is $\tfrac{d(\n^r,m^r_{\theta})}{p(m^r_{\theta})} = \tfrac{4\varepsilon}{3p(m^r_{\theta})} > 1$ (as $p(m^r_{\theta}) = \varepsilon $). Since $\n^r$ has a disutility of less than $\tau$ only for the chores in the clause gadget of $C_r$ and the chores in the corresponding variable gadgets, we can claim that $ \bigcup_{\left\{ \theta \mid \ell_{\theta} = x_{\theta} \right\}} m^r_{\theta} \subseteq \MPB_{\n^r}$. Now, that we have identified the $\MPB$ chores for the agents in the clause gadget, we outline a money flow,
  
   \begin{align*}
   f(n_{\theta}^r,m_{\theta}^r) &= \varepsilon           &\text{ (for all $\ell_{\theta}$) }\\
   f(\n^r,m_{\theta}^r) &= p(m_{\theta}^r) - \varepsilon \enspace . &\text{ (for all $\ell_{\theta} = x_{\theta}$) }
   \end{align*} 
   All the agents spend on their corresponding $\MPB$ chores. Observe that for all $\ell_{\theta} $, the agents $n_{\theta}^r$ earn their entire money of $\varepsilon$ by doing chore $m_{\theta}^r$ (partially if $\ell_{\theta} = x_{\theta}$ and completely when $\ell_{\theta} = \neg x_{\theta}$). The agent $\n^r$ earns her entire money by completing whatever is left of the chores in $\bigcup_{\left\{ \theta \mid \ell_{\theta} = x_{\theta} \right\}} m_{\theta}^r$. It only suffices to show that agent $\n^r$ earns exactly $e(\n^r)$. To this end, we observe that the total money earned by $\n^r$ is 
   \begin{align*}
   \sum_{\left\{ \theta \mid \ell_{\theta} = x_{\theta} \right\}} f(\n^r,m_{\theta}^r) &= \sum_{\left\{ \theta \mid \ell_{\theta} = x_{\theta} \right\}}(p(m_{\theta}^r) -\varepsilon)\\      
																					  &= \#(S_r) \cdot \bigg(\varepsilon + \frac{e(\n^r)}{\#(S_r)} - \varepsilon \bigg)\\
																					  &=  e(\n^r).		
   \end{align*}
   Therefore, we have an allocation where the agents in the variable gadgets earn their money by completing the chores in the variable gadgets and the agents in the clause gadgets earn their entire money by completing the chores in the clause gadgets. This concludes the proof for the case $U_r = \emptyset$. 
\end{proof}

This brings us to the main result of this section.

\begin{theorem}
	Determining an $(\tfrac{11}{12} + \delta)$-competitive equilibrium, for any $\delta>0$,  in chore division with fixed earnings is strongly NP-hard.
\end{theorem}
\begin{proof}
	Given any instance $I =\langle X,\C \rangle$ of 3-SAT, in polynomial time we can construct an instance $E(I)$ of chore division comprising of all variable gadgets and clause gadgets. Also, observe all the entries in the disutility matrix $d(\cdot,\cdot )$ and the money vector $e(\cdot)$ are constants (Thus all input parameters can be expressed with polynomial bit size in unary notation). Lemma~\ref{negative} implies that we have a $(\tfrac{11}{12} + \delta)$-competitive equilibrium only if $I$ is satisfiable and  Lemma~\ref{positive} implies that if $I$ is satisfiable, then  $E(I)$ admits a competitive equilibrium (and thus also a $(\tfrac{11}{12} + \delta)$-competitive equilibrium).    
\end{proof}

\begin{remark}
Note that every instance of chore division with fixed earnings $\langle A,B,d(\cdot,\cdot),e(\cdot) \rangle$, where $e(a)$ is an integer for all $a \in A$,  can be transformed into an instance $I' = \langle A',B,d'(\cdot,\cdot) \rangle$ of chore division with equal incomes (where $e(a) =1$ for all $a \in A'$) by creating $e(a)$ many identical copies (having the exact same disutility profile) of the agent $a \in A$ (the good set remains unchanged): Every $\alpha$-competitive equilibrium in $I'$ will also be an $\alpha$-competitive equilibrium in $I$. Observe that in our instance $E(I)$, we can scale the earning functions of all the agents by some large scalar $\gamma(\varepsilon,\varepsilon')$ to make the earnings of the agents integral. Again, since  $e(a) \in \mathcal{O}(1)$ and  $\tfrac{1}{\varepsilon}, \tfrac{1}{\varepsilon'} \in \mathcal{O}(1)$, we have $\lvert A' \rvert = \mathcal{O}(\lvert A \rvert)$ and  all the input parameters of $A'$ (all entries in the disutility matrix $d'(\cdot,\cdot)$) can be expressed with polynomial bit size in unary notation. \emph{Therefore, finding an $(\tfrac{11}{12} + \delta)$-competitive equilibrium, for any $\delta>0$,  in chore division with equal incomes is also strongly NP-hard.} 
\end{remark}

\section{Sufficiency Conditions for the Existence of Equilibrium}\label{sufficiency}
In this section, we formulate certain conditions that if any instance of chore division satisfies, will admit a  competitive equilibrium. The reader is encouraged to read Section~\ref{mainres2} to get an overall picture of the results, ideas and techniques used in this section. 

Recall that in all the instances that we construct to show non-existence of  competitive  equilibrium, we crucially use the following structure in the disutility matrix: there are agent sets $A_1$ and $A_2$ and sets of chores $B_1$ and $B_2$ such that agents in $A_1$ have disutility of less than $\tau$  for all chores in $B_1$, but disutility of more than $\tau$   for all chores in $B_2$; However, agents in $A_2$ have disutility of less than $\tau$  for some chores in $B_1$ and all chores in $B_2$. Therefore, we now look into instances which do not have such structure in the disutility matrix. To this end, we define the \emph{disutility graph} of an instance as the bipartite graph $D = (A \cup B, E_D)$ containing the agents $A$ as one vertex set and the chores $B$ as the other and there exists an edge from agent $a \in A$ to chore $b \in B$ if and only if $d(a,b) < \tau $. Our first sufficiency condition is that,
\begin{align*}
 \text{\textbf{Condition 1:} $D$ is a disjoint union of complete bipartite graphs $D_1,D_2,\dots ,D_d$ for some $d\ge 1$.}
\end{align*}

Observe that the above condition is violated by all the examples in Section~\ref{mainres1} and the instance $E(I)$ in Section~\ref{nphardness}.  Recall that we showed in Section~\ref{mainres2} that Condition 1 alone is not sufficient for chore division (it is sufficient for chore division with fixed earnings (Problem~\ref{CDFI})). To this end, we define the \emph{exchange graph} of an instance as a  graph $W = ([d], E_W)$. We have $(i,j) \in E_W$ if and only if for every chore $b$ in the component $D_j$ of the disutility matrix, there is an agent $a \in D_i$ such that $w_{a,b} > 0$. We now propose our second sufficiency condition.
\begin{align*}
\text{\textbf{Condition 2:} $W$ is strongly connected.}
\end{align*}

Let $\mathcal{I}$ denote all the instances of chore division that satisfy Condition 1 and Condition 2. We now show that all instances in $\mathcal{I}$ admit a competitive  equilibrium. Consider any instance $I = \langle D,W \rangle \in \mathcal{I}$ such that $D = \cup_{i \in [d]} D_i$ where each $D_i = (A_i \cup B_i, E_{D_i})$ is a complete bipartite graph, disjoint from $D_{i'}$ ($i' \neq i$). For ease of notation:
\begin{itemize}
	\item We represent our set $A$ of $n$ agents as $[n]$ and the set $B$ of $m$ chores as $[m]$.
	\item We also write $p_j$ to denote the price of chore $j$ and $w_{i,j}$ to represent the agent $i$'s initial endowment of chore $j$.
	\item Lastly, we also assume without loss of generality that the total endowment of each chore is one: $\sum_{i \in [n]}w_{i,j} =1$.
\end{itemize}
    Now, we briefly introduce some basic definitions and concepts required to prove the existence of   competitive equilibrium.

\paragraph{Normalized Prices and Bounded Allocations.} A price vector $p = \langle p_1,p_2, \dots ,p_m \rangle$ is called a \emph{normalized price vector} if
 \begin{itemize}
 	\item $p_j \geq 0$ for all $j \in [m]$,
 	\item $\sum_{j \in [m]} p_j =1$, and 
 	\item $\sum_{i \in A_k} \sum_{j \in [m]} w_{i,j} \cdot p_j = \sum_{j \in B_k}p_j$ for each component $D_k$ in the disutility graph (sum of prices of chores in  $D_k$ equals the sum of total money of the agents in $D_k$). 
 \end{itemize}
   Let $P$ be the set of all normalized price vectors. We first show that the set $P$ is non-empty.
   
   \begin{observation}
   	\label{Pisnonempty}
   	   We have $P \neq \emptyset$.
   \end{observation}
   
   \begin{proof}
   	Here we will make use of a general fact that will be useful for a proof later as well.
   	
   	\begin{fact}
   		\label{stochasticfixedpoint}
   		Let $Z \in \mathbb{R}^{n \times n}$ be a square matrix such that $Z_{ij} \geq 0$ for all $j \neq i$ (all the non-diagonal entries of $Z$ are non-negative) and $\sum_{i \in [n]} Z_{ij} = 0$ for all $j \in [n]$ (column sums are zero), then there exists a vector $t \in \mathbb{R}^n_{\geq 0}$ such that $\sum_{i \in [n]} t_i =1$  and $Z \cdot t = 0$.
   	\end{fact}
    The proof of this fact can be found at the end of this section. Using this fact, we will outline a proof that $P$ is non-empty. For each component $D_k$ of the disutility matrix we pick a chore $b_k \in B_k$ and we set $p_j = 0$ for all $j \in B_k \setminus \left\{b_k \right\}$. Note that to show that $P$ is non-empty, it suffices to show that there exists a vector $p' = \langle p'_1,p'_2, \dots ,p'_d \rangle$ (intuitively each $p'_k$ corresponds to the price of chore $b_k \in B_k$, i.e., $p_{b_k}$) such that $p'_{k} \geq 0$ for all $k \in [d]$, $\sum_{k \in [d]} p'_{k} =1$ and we have,
    
    \begin{align}
     \label{Wmatrixeqns}
       \sum_{i \in A_{k}} \sum_{{k'} \in [d]} w_{i,b_{k'}} \cdot p'_{k'} - p'_{k}&=0  &\text{for all $k \in [d]$}
    \end{align}
    Let $W$ be the coefficient matrix of the system of equations in~\eqref{Wmatrixeqns}, i.e., $W \cdot p' = 0$ represents the system of equations in~\eqref{Wmatrixeqns}. Observe that $W_{kk'} = \sum_{i \in A_{k}}w_{i,b_{k'}}$ if $k \neq k'$ and $W_{kk} = \sum_{i \in A_{k}}w_{i,b_k} - 1$. Therefore the non-diagonal entries of $W$ are non-negative and also note that the column sum is zero:
    \begin{align*}
     \sum_{k \in [d]} M_{kk'} &= \sum_{k \in [d]} \sum_{i \in A_{k}} w_{i,b_{k'}} - 1\\
                               &=\sum_{i \in [n]} w_{i,b_{k'}} -1\\
                               &=1-1 &\text{(total endowment is one)}\\
                               &=0.
    \end{align*}
    Therefore $W$ satisfies all the conditions in Fact~\ref{stochasticfixedpoint}. Therefore, by Fact~\ref{stochasticfixedpoint} there exists a $p' \in \mathbb{R}^d_{\geq 0}$, such that $\sum_{k \in [d]} p'_d =1$ and $W \cdot p' = 0$. Therefore, $P$ is non-empty.    	
   \end{proof}
   
   Since $P$ is defined by a set of linear equalities and inequalities, $P$ is closed and convex too. Additionally, since $p \in \mathbb{R}^m_{\geq 0}$ and  $\sum_{j \in [m]} p_j =1$ for all $p \in P$, $P$ is compact. We now make a small observation about any $p \in P$ that will be useful later.
   
   \begin{observation}
   	\label{nozeroprice}
   	Consider any $p \in P$. For all $k \in [d]$, we have that the total price of chores in a component is strictly larger than zero: $\sum_{j \in B_k} p_j > 0$. 
   \end{observation}
   \begin{proof}
   	We prove this by contradiction. Assume otherwise and let $\ell \in [d]$ be such that $\sum_{j \in B_{\ell}} p_j = 0$. Since $p \in P$, we have $\sum_{i \in A_{\ell}} \sum_{j \in [m]} w_{i,j} \cdot p_j = 0$ as well. Consider any $\ell' \in [d]$ such that $(\ell,\ell') \in E_W$ (recall that $E_W$ is the edge set of the exchange graph $W$). By Condition 2, we have that for each chore $b \in B_{\ell'}$, there is an agent $a \in A_{\ell}$ such that $w_{a,b} > 0$. Therefore, the only way $\sum_{i \in A_{\ell}} \sum_{ j \in [m]} w_{i,j} \cdot p_j =0$ is if all the prices of the chores in $B_{\ell'}$ is $0$. This implies that if the sum of prices of the chores in any component $D_{\ell}$ is zero, then the sum of prices of the chores in all the components $D_{\ell'}$ such that $(\ell,\ell') \in E_W$ is also zero. Since the exchange graph $W$ is strongly connected, a repeated application of this argument will imply that sum of prices of all chores is zero for every component, which is a contradiction as we have $\sum_{j \in [m]} p_j =1$.  
   \end{proof}

  An allocation $X \in \mathbb{R}_{\geq 0} ^{n \times m}$, is called a \emph{bounded allocation} if each $X_{ij}$ (quantifies the amount of chore $j$ allocated to agent $i$) is non-negative and is at most $m \cdot \tfrac{d_{\mathit{max}}}{d_{\mathit{min}}}$, where $d_{\mathit{max}}$ and $d_{\mathit{min}}$ refer to the largest and smallest entry less than $\tau$ in the disutility matrix. Let $\X$ be the set of all bounded allocations. Observe that the set $\X$ is non-empty, convex and compact. Also, we have that $P$ is non-empty, convex and compact.  We define a compact, convex and non-empty subset of $\mathbb{R}^{(m + nm)}$,  $S = \bigcup_{p \in P} \bigcup_{X \in \X} \langle p,X \rangle$ \footnote{We abuse notation slightly here: $\langle p, X \rangle$ refers to the $(m+nm)$-dimensional vector $\langle p_1,p_2, \dots, p_m, X_{11}, X_{12}, \dots ,X_{nm} \rangle $.}.

\paragraph{Correspondence $\phi$.} Our goal is to define a \emph{correspondence} or equivalently a \emph{set valued function} $\phi: S \rightarrow 2^S$, such that $\phi$ has at least one fixed point and any fixed point of $\phi$ will correspond to  competitive  equilibrium. We will first show some properties that if satisfied by $\phi$, then $\phi$ will have at least one fixed point and any fixed point of $\phi$ will correspond to a competitive  equilibrium. Then, we will define a $\phi$ that satisfies these properties.

\paragraph{Properties.} We first make some basic definitions that will help us to state the properties. We call a bounded allocation $Y \in \X$ an \emph{optimal allocation at the price vector  $p$} if and only if, 
\begin{itemize}
	\item $Y_{ij} > 0$ only if $d(i,j) < \tau$, $\tfrac{d(i,j)}{p_j} \leq \tfrac{d(i, \ell)}{p_{\ell}}$ for all $\ell \in [m]$, and 
	\item $\sum_{j \in [m]} Y_{ij} \cdot p_j = \sum_{j \in [m]} w_{i,j} \cdot p_j$ for all $i \in [n]$.
\end{itemize}
Let $\X^p\subseteq \X$ denote the set of all optimal allocations at the price vector $p$. Right now, it may not be immediately clear that $\X^p$ is non-empty. However, we show that this is indeed the case, as agents are allowed to consume goods to a significant extent ($Y_{ij}$ is allowed to be as large as $m \cdot \tfrac{d_{\mathit{max}}}{d_{\mathit{min}}}$).

\begin{lemma}
	\label{optimalbundlenotempty}
	For all $p \in P$, we have  $\X^p \subseteq \X$ and $\X^p \neq \emptyset$.
\end{lemma}

\begin{proof}
	By definition $\X^p \subseteq \X$. Therefore, it suffices to show that it is non-empty. Consider any $p \in P$. Consider any agent $a \in A_k$ (recall that $A_k$ is the set of agents that belong to the component $D_k$ of the disutility graph). Let $\w(a) = \sum_{j \in [m]} w_{a,j} \cdot p_j$. If $\w(a) = 0$, then  we set $Y_{aj} = 0$ for all $j \in [m]$ and we trivially have $\sum_{j \in [m]} Y_{aj} \cdot p_j = \sum_{j \in [m]} w_{a,j} \cdot p_j = 0$ and $\langle Y_{a1}, \dots, Y_{am} \rangle$ is an optimal bundle for agent $a$ at $p$. So assume that $\w(a) > 0$. Since $p \in P$, we have that the sum of prices of the chores in $D_k$, $\sum_{j \in B_k}p_j = \sum_{i \in A_k} \sum_{j \in [m]} w_{ i,j} \cdot p_j \geq \sum_{j \in [m]} w_{ aj} \cdot p_j = \w(a)$. This implies that there is at least one chore $b \in B_k$ such that $p_b \geq \tfrac{\w(a)}{m}$. Let $b'$ be a chore such that $d(a,b') < \tau $,  $\tfrac{d(a,b')}{p_{b'}} \leq \tfrac{d(a, \ell)}{p_{\ell}}$ for all $\ell \in [m]$. This implies that $\tfrac{d(a,b')}{p_{b'}} \leq \tfrac{d(a,b)}{p_{b}}$. Therefore, we have that 
	\begin{align*}
	 p_{b'} &\geq \frac{d(a,b')}{d(a,b)} \cdot p_b \\
	        &\geq \frac{d_{\mathit{min}}}{d_{\mathit{max}}} \cdot p_b\\
	        &\geq  \frac{d_{\mathit{min}}}{md_{\mathit{max}}} \cdot \w(a).
	\end{align*}
	We set $Y_{ab'} = \tfrac{\w(a)}{p_{b'}}$. Observe that $Y_{ab'} \leq m \cdot \tfrac{d_{\mathit{max}}}{d_{\mathit{min}}}$. Therefore, $Y$ is a bounded allocation. Also, note that each agent $a$ earns her entire money of $\w(a)$ by doing  $\tfrac{\w(a)}{p_{b'}}$ amount of chore $b'$ such that $d(a,b') < \tau $, $\tfrac{d(a,b')}{p_{b'}} \leq \tfrac{d(a, \ell)}{p_{\ell}}$ for all $\ell \in [m]$. Thus, $Y$ is an optimal bundle also. Therefore, $\X^p \neq \emptyset$. 
\end{proof}

 We are now ready to define the properties of $\phi$. For any point $ \langle p,X \rangle \in S$, consider any point $\langle p', X' \rangle \in \phi(\langle p, X \rangle)$. Then,    

\begin{itemize}
	\item Property $\mathbf{P}_1$: $X' \in \X^p$ and $p' \in P$.
	\item Property $\mathbf{P}_2$: For any two agents $i$ and $j$ that belong to the same component of the disutility graph $D$ (say $i$, $j \in A_k$), and $p_j \neq 0$,  we have 
	           \begin{align*}\frac{p'_i}{p'_j} &= \frac{p_i + \mathit{max}(1 - \sum_{\ell \in [n]}X_{\ell i},0)}{p_j + \mathit{max}(1 - \sum_{\ell \in [n]}X_{\ell j},0)}. \end{align*}. 
	\item Property $\mathbf{P}_3$: $\phi(\langle p,X \rangle )$ is non-empty and convex.
	\item Property $\mathbf{P}_4$: $\phi$ has a \emph{closed graph}\footnote{A correspondence $\phi: X \rightarrow 2^Y$ has a \emph{closed graph} if for all sequences $\left\{x_n \right\}_{n \in \mathbb{N}}$ and $\left\{y_n\right\}_{n \in \mathbb{N}}$, with $x_n \in X$ and $y_n \in \phi(x_n)$ for all $n$, such that $x_n \rightarrow x$ and $y_n \rightarrow y$, we have $y \in \phi(x)$.}. 
\end{itemize}

We will now show that any correspondence $\phi$ that satisfies $\mathbf{P}_1$, $\mathbf{P}_2$, $\mathbf{P}_3$ and $\mathbf{P}_4$ will have at least one fixed point and any fixed point will correspond to   competitive equilibrium. We first show that $\phi$ has a fixed point.

\begin{lemma}
	\label{existsfixedpoint}
	Consider any correspondence $\phi$ that satisfies properties $\mathbf{P}_1$, $\mathbf{P}_2$, $\mathbf{P}_3$ and $\mathbf{P}_4$. $\phi$ has a fixed point.
\end{lemma}
\begin{proof}
	By property $\mathbf{P}_1$ we have that if $\langle p',X' \rangle \in \phi(\langle p, X \rangle )$, then $\langle p', X' \rangle \in S$ (as $p' \in P$ and $X' \in \X^p \subseteq \X$). Therefore, $\phi:S \rightarrow 2^S$. The set $S$ is non-empty, compact and convex. Furthermore, by properties $\mathbf{P}_3$ and $\mathbf{P}_4$, we have that $\phi(\langle p,X \rangle )$ is non-empty and convex, and $\phi$ has a {closed graph}. Therefore, by {Kakutani's fixed point theorem}, $\phi$ has a fixed point. 
\end{proof}

Now we show that any fixed point of a correspondence $\phi$ that satisfies properties $\mathbf{P}_1$, $\mathbf{P}_2$, $\mathbf{P}_3$ and $\mathbf{P}_4$ gives a competitive equilibrium.

\begin{lemma}
	\label{fixedpointisequilibrium}
	 Consider any correspondence $\phi$ that satisfies properties $\mathbf{P}_1$, $\mathbf{P}_2$, $\mathbf{P}_3$ and $\mathbf{P}_4$. Consider any fixed point $\langle p,X \rangle$ of $\phi$. Then $\langle p, X \rangle$ is a competitive equilibrium.
\end{lemma}
\begin{proof}
	Consider any fixed point $\langle p,X \rangle  \in \phi(\langle p, X \rangle)$. By property $\mathbf{P}_1$ it follows that $X \in \X^p$.  Now, it suffices to show that we have $p_j >0$ for all $j \in [m]$ and  $\sum_{i \in [n]} X_{ij} =1$ for all chores $j \in [m]$. We first show that $p_j > 0$ for all $j \in [m]$. We prove this by contradiction. Let us assume that there is some chore $b$ in some component $D_{\ell}$ of the disutility graph such that $p_b = 0$. Note that by Observation~\ref{nozeroprice}, there is at least one chore $b' \in B_{\ell}$ such that $p_{b'} > 0$. This implies that for all agents $i \in A_{\ell}$,  we have $\tfrac{d(i,b')}{p_{b'}} < \tfrac{d(i,b)}{p_b}$. Therefore, we have that $X_{ib} = 0$ for all $i \in A_{\ell}$ and also for all $i \in [n]$ (as $X_{i\ell} > 0$ only if $d(i,\ell) < \tau$ and for all agents in $[n] \setminus A_k$ we have $d(i,b) \geq \tau$), implying $\sum_{\ell \in [n]} X_{ \ell b} = 0$. Now observe that, 
	\begin{align*}
	 \frac{p'_b}{p'_{b'}} &=\frac{p_b + \mathit{max}(1 - \sum_{\ell \in [n]}X_{\ell b},0)}{p_{b'} + \mathit{max}(1 - \sum_{\ell \in [n]}X_{\ell b'},0)} &\text{(by property $\mathbf{{P}_2}$)}\\
	                      &= \frac{0 + 1}{p_{b'} + \mathit{max}(1 - \sum_{\ell \in [n]}X_{\ell b'},0)}\\
	                      & \neq 0\\
	                      &= \frac{p_b}{p_{b'}}.
	\end{align*}	
	This implies that $p' \neq p$, which is a contradiction to $\langle p, X \rangle$ being a fixed point.
	
	Therefore, we have $p_j > 0$ for all $j \in [m]$. We now show that $\sum_{i \in [n]} X_{ij} =1$ for all $j \in [m]$. We prove this also by contradiction. So assume otherwise and for some chore $b \in B_k$ we have $\sum_{i \in [n]} X_{ib} > 1$ (or $\sum_{i \in [n]} X_{ib} < 1$). Note that, since $p \in P$,  for the component $D_k$ of the disutility graph, we have,
	\begin{align}
	  \label{eq1}
	  \sum_{j \in B_k}p_j &= \sum_{i \in A_k} \sum_{j \in [m]} w_{i,j} \cdot p_j.  
	\end{align}
	Also, since $X \in \X^p$,  for every agent $i \in A_k$, we have $\sum_{j \in [m]} w_{i,j} \cdot p_j = \sum_{j \in [m]} X_{ij} \cdot p_j = \sum_{j \in B_k} X_{ij} \cdot p_j$. Substituting $\sum_{j \in [m]} w_{i,j} \cdot p_j$ as $\sum_{j \in B_k} X_{ij} \cdot p_j$ in~\eqref{eq1} we have,
	\begin{align*}
	 \sum_{j \in B_k}p_j &= \sum_{i \in A_k} \sum_{j \in B_k} X_{ij} \cdot p_j\\       
	                     &= \sum_{i \in [n]} \sum_{j \in B_k} X_{ij} \cdot p_j\\
	                     &= \sum_{j \in B_k} p_j \cdot ( \sum_{i \in [n]} X_{ij}) \enspace .
	\end{align*}    
	Therefore, if  $\sum_{i \in [n]} X_{ib} >1$ (or $\sum_{i \in [n]} X_{ib} <1$) for some $b \in B_k$, then there exists a $b' \in B_k$ such that $\sum_{i \in [n]} X_{ib'} <1$ (or $\sum_{i \in [n]} X_{ib'} >1$). This would imply that $\frac{p_b + \mathit{max}(1 - \sum_{\ell \in [n]}X_{\ell b},0)}{p_{b'} + \mathit{max}(1 - \sum_{\ell \in [n]}X_{\ell b'},0)} < \tfrac{p_b}{p_{b'}} $ when $\sum_{i \in [n]} X_{ib} >1$  and $\frac{p_b + \mathit{max}(1 - \sum_{\ell \in [n]}X_{\ell b},0)}{p_{b'} + \mathit{max}(1 - \sum_{\ell \in [n]}X_{\ell b'},0)} > \tfrac{p_b}{p_{b'}} $ when $\sum_{i \in [n]} X_{ib} <1$, which is a contradiction (as $\frac{p_b + \mathit{max}(1 - \sum_{\ell \in [n]}X_{\ell b},0)}{p_{b'} + \mathit{max}(1 - \sum_{\ell \in [n]}X_{\ell b'},0)} = \tfrac{p_b}{p_{b'}} $ if $\langle p, X \rangle$ is a fixed point by property $\mathbf{P}_2$).
\end{proof}

Now, it suffices to show that there exists a correspondence $\phi$ that satisfies all the four properties to show the existence of   competitive equilibrium for every instance $I \in \mathcal{I}$. To this end, we first define a correspondence $\phi$ and show that it satisfies all the four properties.

\paragraph{Finding a Correspondence $\phi$ that Satisfies all the Properties.} Given a $p \in P$ and $X \in \X$, we define the vector $q = \langle q_1,q_2, \dots, q_m \rangle$ such that 
\begin{align}
\label{qdefinition}
 q_j = p_j + \mathit{max}(1 - \sum_{i \in [n]}X_{ij}, 0)\enspace .
\end{align}

We will now outline a system of linear equations that needs to be satisfied by a vector $\tilde{p} = \langle \tilde{p}_1, \tilde{p}_2, \dots, \tilde{p}_d \rangle$ (recall that $d$ is the number of components in the disutility graph $D$). So think of each $\tilde{p}_i$ as a variable now. As of now, let us think of each $\tilde{p}_k$ as the sum of prices of the chores in the component $D_k$ and $\tfrac{q_j}{Q_k} \cdot \tilde{p}_k$ as the price of each chore  $j \in B_k$, where $Q_k = \sum_{j \in B_k} q_j$. We make a small observation about $Q_k$, which will be useful later.
\begin{claim}
	\label{Q_kpositive}
	For all $k \in [d]$, we have $Q_k > 0$. 
\end{claim}
\begin{proof}
  Each $q_j$ is at least as large as $p_j$, implying $Q_k = \sum_{j \in B_k} q_j \geq \sum_{j \in B_k} p_j$ and by Observation~\ref{nozeroprice} we have $\sum_{j \in B_k} p_j > 0$.	
\end{proof}
 With these prices ($\tilde{p}_1, \tilde{p}_2, \dots , \tilde{p}_d$) in mind, for each component $D_k$ of $D$, we write the  equation (variables being $\bigcup_{k \in [d]} \tilde{p}_k$) that represents the total money the agents earn in the component equals the total prices of the chores in the same component.
\begin{align}
\label{eq2}
 \sum_{i \in A_k} \sum_{k' \in [d]} \sum_{j \in B_{k'}} w_{i,j} \cdot \frac{q_j}{Q_{k'}} \cdot \tilde{p}_{k'} -  \tilde{p}_{k} &=0  \enspace .
\end{align} 

We represent the system of equations in~\eqref{eq2} as 
\begin{align}
\label{Mdefinition}
 M \cdot \tilde{p} = \mathbf{0} \enspace . 
\end{align}

First observe that every entry of the matrix $M$ is bounded: This is primarily due to the fact that $Q_k > 0$ for all $k \in [d]$. 

\begin{observation}
	\label{Misspecial}
	  We have $M_{kk'} \geq 0$ as long as $k \neq k'$ (every non-diagonal entry of $M$ is non-zero) and $\sum_{k \in [d]} M_{kk'} = 0$ for all $k' \in [d]$ (column sums are zero).
\end{observation} 
\begin{proof}
	We first carefully look at any column $M_{*k'}$ of $M$.  Note that for all $k \neq k'$, we have, $M_{kk'} = \sum_{i \in A_{k}} \sum_{j \in B_{k'}} w_{i,j} \cdot \tfrac{q_j}{Q_{k'}}$. We have $ M_{kk} = \sum_{i \in A_{k}} \sum_{j \in B_{k}} w_{i,j} \cdot \tfrac{q_j}{Q_k}- 1$. Therefore, every non-diagonal entry in $M$ is non-negative. Now we just need to show that $\mathbf{1}^T \cdot M_{*k'} = 0$.   Observe,
	\begin{align*}
	 \mathbf{1}^T \cdot M_{*k'} &= \sum_{k \in [d]} \sum_{i \in A_{k}} \sum_{j \in B_{k'}} w_{i,j} \cdot \tfrac{q_j}{Q_{k'}} - 1\\
	                           &= \tfrac{1}{Q_{k'}} \cdot \sum_{j \in B_{k'}} q_j \cdot \sum_{k \in [d]} \sum_{i \in A_{k}} w_{i,j} - 1\\
	                           &= \tfrac{1}{Q_{k'}} \cdot \sum_{j \in B_{k'}} q_j \cdot \sum_{i \in [n]} w_{i,j} - 1\\
	                           &=\tfrac{1}{Q_{k'}} \cdot \sum_{j \in B_{k'}} q_j -1\\
	                           &=0.
	\end{align*} 
	This shows that $\mathbf{1}^T \cdot M = \mathbf{0}^T$. 
\end{proof}

We first make some observations about the solution to the system of equations in~\eqref{Mdefinition} (and consequently~\eqref{eq2}). Observe that $M$ satisfies all the conditions in Fact~\ref{stochasticfixedpoint}. Therefore, we have 

\begin{observation}
	\label{nonnegativep}
	There exists a solution to the system of equations in~\eqref{Mdefinition} such that $\tilde{p}_j \geq 0$ for all $j \in [d]$ and $\sum_{j \in [d]} \tilde{p}_j =1$. 
\end{observation}

We are now ready to define the correspondence. Given any $\langle p,X \rangle  \in S$, we determine the vector $q$ as in~\eqref{qdefinition}. Observe that $q$ is a function of $p$ and $X$. Hence, from now on, for clarity we would write the vector $q$, the scalars $Q_k$ for all $k \in [d]$  as $q(p,X)$ and  $Q_k(p,X)$ respectively. Let $\tilde{P}(p,X) \subseteq 2^{R^{d}_{\geq 0}}$ be the set of all vectors that satisfy the conditions in Observation~\ref{nonnegativep}. For each $\tilde{p} \in \tilde{P}(p,X)$, we define the vector $\overline{p}(\tilde{p}) = \langle \overline{p}_1 (\tilde{p}), \dots, \overline{p}_m (\tilde{p})\rangle $ as follows: For every chore $j$ that lies in the component $D_k$ of the disutility graph $D$, we define
\begin{align*}
 \overline{p}_j(\tilde{p}) = \frac{q_j(p,X)}{Q_k(p,X)} \cdot \tilde{p}_k \enspace .
\end{align*}

Let $\overline{P}(p,X) = \left\{\overline{p}(\tilde{p}) \mid \tilde{p} \in \tilde{P}(p,X)  \right\}$. Given any $\langle p,X \rangle \in S$, we define,
\begin{align*}
 \phi(\langle p,X \rangle ) = \left\{ \langle \overline{p}, X' \rangle  \mid \overline{p} \in \overline{P}(p,X) \text{ and } X' \in \X^p \right\} \enspace .
\end{align*}

For the rest of this section, we will now show that $\phi$ satisfies properties $\mathbf{P}_1$, $\mathbf{P}_2$, $\mathbf{P}_3$ and $\mathbf{P}_4$.

\paragraph{$\phi$ satisfies properties $\mathbf{P}_1$, $\mathbf{P}_2$, $\mathbf{P}_3$ and $\mathbf{P}_4$.}
\begin{lemma}
	\label{P1}
	$\phi$ satisfies property $\mathbf{P}_1$.
 \end{lemma}
\begin{proof}
	We need to show that $p' \in P$ and $X' \in \X^p \subseteq \X$. Note that by the definition of $\phi$ we have $X' \in \X^p \subseteq \X$. Therefore, we only need to show that $p' \in P$. Given $p$ and $X$ let $q(p,X)$ be the vector obtained as in~\eqref{qdefinition} and let $Q_k(p,X) = \sum_{j \in B_k} q_j(p,X)$ be the scalars for all $k \in [d]$. Then $p' = \overline{p}(\tilde{p})$ for some $\tilde{p} \in \tilde{P}(p,X)$. Now we make three claims which show that $p' \in P$. 
	\begin{claim}
		\label{tc1}
		We have $p'_j \geq 0$ for all $j \in [m]$.
	\end{claim}
	\begin{proof}
		Let us consider any chore $j$ that belongs to the component $D_k$ in the disutility matrix.
		\begin{align*}
		    p'_j &= \overline{p}_j(\tilde{p})\\
		         &= \frac{q_j(p,X)}{Q_k(p,X)} \cdot \tilde{p}_k\enspace .	 	
		\end{align*}
		Observe from the definition of $q(p,X)$ in~\eqref{qdefinition} that $q_{\ell} \geq 0$ for all $\ell \in [m]$. Also, from the definition of the vector $\tilde{p}$, we have $\tilde{p}_{\ell} \geq 0$ for all $\ell \in [d]$. This implies that $p'_j \geq 0$.
	\end{proof}
	
	\begin{claim}
		\label{tc2}
		We have $\sum_{j \in [m]} p'_j =1$.
	\end{claim}
	\begin{proof}
		We have $\sum_{j \in [m]} p'_j = \sum_{j \in [m]} \overline{p}_j(\tilde{p}) = \sum_{k \in [d]}\sum_{j \in B_k} \tfrac{q_j(p,X)}{Q_k(p,X)} \cdot \tilde{p}_k= \sum_{k \in [d]} \tilde{p}_k=1$.
	\end{proof}
	
	\begin{claim}\label{tc3}
		For each component $D_k$ of the disutility graph, we have $\sum_{i \in A_k} \sum_{j \in [m]} w_{i,j} \cdot p'_j = \sum_{j \in B_k} p'_j$.
	\end{claim}
	\begin{proof}
		We have,
		\begin{align*}
		 \sum_{i \in A_k} \sum_{j \in [m]} w_{i,j} \cdot p'_j &= \sum_{i \in A_k} \sum_{j \in [m]} w_{i,j} \cdot \overline{p}_j(\tilde{p})\\           
		                                                     &= \sum_{i \in A_k} \sum_{k' \in [d]} \sum_{j \in B_{k'}} w_{i,j} \cdot \frac{q_j(p,X)}{Q_{k'}(p,X)} \cdot \tilde{p}_{k'}\enspace .
		\end{align*}
		Since $\tilde{p}$ satisfies~\eqref{Mdefinition} and therefore also~\eqref{eq2}, we have $\sum_{i \in A_k} \sum_{k' \in [d]} \sum_{j \in B_{k'}} w_{i,j} \cdot \frac{q_j(p,X)}{Q_{k'}(p,X)} \cdot \tilde{p}_{k'} = \tilde{p}_{k}$. Therefore, we have
		\begin{align*}
		 \sum_{i \in A_k} \sum_{j \in [m]} w_{i,j} \cdot p'_j &=\tilde{p}_k\\
		                                                      &=\sum_{j \in B_k} \frac{q_j(p,X)}{Q_k(p,X)} \cdot \tilde{p}_{k}\\
		                                                      &=\sum_{j \in B_k} \overline{p}_j(\tilde{p})\\
		                                                      &=\sum_{j \in B_k} p'_j. \qedhere
		\end{align*}
	\end{proof}
	This shows that $p' \in P$ and completes the proof. 
\end{proof}

\begin{lemma}
	\label{P2}
	$\phi$ satisfies property $\mathbf{P}_2$.
\end{lemma}

\begin{proof}
	Consider any $\langle p',X' \rangle \in \phi(\langle p,X \rangle)$. Let $i,j$ be two chores in the component $D_k$ of the disutility graph such that $p'_j \neq 0$. Since $p'_j = \frac{q_j(p,X)}{Q_k(p,X)} \cdot \tilde{p}_{k} \neq 0$, we have $\tilde{p}_k \neq 0$. Now observe that,
	\begin{align*}
	 \frac{p'_i}{p'_j} &=\frac{\overline{p}_i(\tilde{p})}{\overline{p}_j(\tilde{p})}\\
	                   &=\frac{\frac{q_i(p,X)}{Q_k(p,X)} \cdot \tilde{p}_{k}}  {\frac{q_j(p,X)}{Q_k(p,X)} \cdot \tilde{p}_{k}}\\
	                   &=\frac{q_i(p,X)}{q_j(p,X)} &(\text{as $Q_k(p,X) \neq 0$ and $\tilde{p}_k \neq 0$})\\
	                   &=\frac{p_i + \mathit{max}(1 - \sum_{\ell \in [n]}X_{\ell i},0)}{p_j + \mathit{max}(1 - \sum_{\ell \in [n]}X_{\ell j},0)}\enspace . &(\text{by definition of $q(p,X)$ in~\eqref{qdefinition}})
	\end{align*} 
\end{proof}

\begin{lemma}
	\label{P3}
	$\phi$ satisfies property $\mathbf{P}_3$.
\end{lemma}
	
\begin{proof}
	We first show that $\X^p$ is convex. Consider $Y \in \X^p$ and $Y' \in \X^p$. Let $Y'' = \lambda \cdot Y + (1- \lambda) \cdot Y'$ for some $\lambda \in [0,1]$. First observe that $0 \leq \mathit{min}(Y_{ij}, Y'_{ij}) \leq Y''_{ij} \leq \mathit{max}(Y_{ij},Y'_{ij}) \leq m \cdot \tfrac{d_{\mathit{max}}}{d_{\mathit{min}}}$. Therefore, $Y'' \in \X$. Also, notice that we have $Y_{ij} > 0$ only if $d(i,j) < \tau $ and $\tfrac{d(i,j)}{p_j} \leq \tfrac{d(i,{\ell})}{p_{\ell}}$ for all $\ell \in [m]$. Therefore, we also have $Y'_{ij}>0$ only if $d(i,j) < \tau $ and $\tfrac{d(i,j)}{p_j} \leq \tfrac{d(i,{\ell})}{p_{\ell}}$ for all $\ell \in [m]$. Lastly we have,
	\begin{align*}
	\sum_{j \in [m]} Y''_{ij} \cdot p_j &=\sum_{j \in [m]} (\lambda \cdot Y_{ij} + (1- \lambda) \cdot Y'_{ij}) \cdot p_j\\ 
	&= \lambda \cdot (\sum_{j \in [m]} Y_{ij} \cdot p_j) + (1 - \lambda) \cdot (\sum_{j \in [m]} Y'_{ij} \cdot p_j)\\ 
	&= \lambda \cdot \sum_{j \in [m]}w_{i,j} \cdot p_j + (1-\lambda) \cdot  \sum_{j \in [m]}w_{i,j} \cdot p_j\\ 
	&= \sum_{j \in [m]}w_{i,j} \cdot p_j \enspace .
	\end{align*}
	Thus, $Y'' \in \X^p$. Therefore, $\X^p$ is convex and non-empty (by Lemma~\ref{optimalbundlenotempty}).

	We now show that $\overline{P}(p,X)$ is convex and non-empty. By Observation~\ref{nonnegativep}, we have that $\tilde{P}(p,X)$ is non-empty and therefore $\overline{P}$ is also non-empty. To show that it is convex, consider two price vectors $t$ and $t'$ in $\tilde{P}(p,X)$. To show convexity, we need to show that $\lambda \cdot \overline{p}(t) + (1- \lambda) \cdot \overline{p}(t')$ is in $\overline{P}$ for all $\lambda \in [0,1]$. Note that  for all $j \in [m]$, that belongs to the component $D_k$ in the disutility graph, we have
	\begin{align*}
	\lambda \cdot \overline{p}_j(t) + (1-\lambda) \cdot \overline{p}_j(t') = \tfrac{q_j(p,X)}{Q_k(p,X)} \cdot (\lambda \cdot t_j + (1-\lambda) \cdot t'_j) \enspace .
	\end{align*}
	Thus it suffices to show that $\lambda \cdot t + (1-\lambda) \cdot t' \in \tilde{P}(p,X)$. This is true as $\tilde{P}(p,X)$ is defined by a set of linear equalities and inequalities, and therefore is convex.
	
	Therefore, we have that both the sets $\overline{P}(p,X)$ and $\X^p$ are non-empty and convex. thus, $\phi(\langle p, X \rangle)$ is also non-empty and convex as it is the Cartesian product of $\overline{P}(p,X)$ and $\X^p$.
\end{proof}

\begin{lemma}
	\label{P4}
	$\phi$ satisfies property $\mathbf{P}_4$.
\end{lemma}
	
\begin{proof}
	Consider a sequence $\langle p^n, X^n \rangle$ that converges to $\langle p^*, X^* \rangle$ and $\langle p^n, X^n \rangle \in S$ for all $n$. Similarly, consider the sequence $\langle r^n, Y^n \rangle$ that converges to $\langle r^*, Y^* \rangle$ such that $\langle r^n, Y^n \rangle \in \phi( \langle p^n, X^n \rangle )$ for all $n$. To show that $\phi$ has a closed graph, we need to show that $\langle r^*, Y^* \rangle \in \phi( \langle p^*, X^* \rangle )$. To show that, $\langle r^*, Y^* \rangle \in \phi( \langle p^*, X^* \rangle )$, we need to show,
	\begin{enumerate}
		\item $r^* \in \overline{P}(p^*,X^*)$, and 
		\item $Y^* \in \X^{p^*}$.
	\end{enumerate}

	\paragraph{Proving $r^* \in \overline{P}(p^*,X^*)$:} We first outline the necessary and sufficient condition for any vector $p'$ to be in $\overline{P}(p,X)$, as this will be useful for our proof.
	
	\begin{observation}
		\label{necessarysufficient}
		$p' \in \overline{P}(p,X)$ if and only if 
		\begin{enumerate}
			\item $p' \in P$, and 
			\item for each chore $j$ in component $D_k$, we have $p'_j = \tfrac{q_j(p,X)}{Q_k(p,X)} \cdot \sum_{j \in B_k} p'_j$.
		\end{enumerate}
	\end{observation}
	
	\begin{proof}
		We first show the ``if" direction. For each component $D_k$ of the disutility graph, let $\tilde{p}_k = \sum_{j \in B_k} p'_j$. Then for each chore $j \in B_k$, we have $p'_j = \tfrac{q_j(p,X)}{Q_k(p,X)} \cdot \tilde{p}_k$. Let $\tilde{p} = \langle \tilde{p}_1, \dots ,\tilde{p}_d \rangle$. To show that $p' \in \overline{P}(p,X)$, it suffices to show that $\tilde{p} \in \tilde{P}(p,X)$. To this end, observe that $\tilde{p}_k = \sum_{j \in B_k} p'_j \geq 0$ as $p'_j \geq 0$ for all $j \in [m]$ (as $p' \in P$). Furthermore, $\sum_{k \in [d]} \tilde{p}_k = \sum_{j \in [m]} p'_j = 1$ (as $p' \in P$). Now, to show $\tilde{p} \in \tilde{P}(p,X)$ it suffices to show that $\tilde{p}$ satisfies the system of equations in~\eqref{Mdefinition} or equivalently those in~\eqref{eq2}. To this end, since $p' \in P$, for each component $D_k$ we have,
		\begin{align*}
		\sum_{i \in A_k} \sum_{j \in [m]} w_{i,j} \cdot p'_j = \sum_{j \in B_k} p'_j \enspace .
		\end{align*} 
		Or equivalently,
		\begin{align*}
		\sum_{i \in A_k} \sum_{k' \in [d]} \sum_{j \in B_{k'}} w_{i,j} \cdot p'_j = \sum_{j \in B_k} p'_j \enspace .
		\end{align*} 
		Substituting every $p'_j$ as $\tfrac{q_j(p,X)}{Q_k(p,X)} \cdot \tilde{p}_k$ where chore $j$ is in the component $D_k$ we have,
		\begin{align*}
		\sum_{i \in A_k} \sum_{k' \in [d]} \sum_{j \in B_{k'}} w_{i,j} \cdot \frac{q_j(p,X)}{Q_k(p,X)} \tilde{p}_{k'} = \tilde{p}_k \enspace .
		\end{align*} 	
		Therefore, $\tilde{p}_k$ satisfies~\eqref{eq2}. Thus $\tilde{p} \in \tilde{P}(p,X)$.
		
		Now we show the ``only if" direction. So assume $p' \in \overline{P}(p,X)$. Then, by Claims~\ref{tc1},~\ref{tc2} and ~\ref{tc3} we have that $p' \in P$. Also  by the definition of $\overline{P}(p,X)$, we also have that there exists a vector $\tilde{p} = \langle \tilde{p}_1, \dots , \tilde{p}_d \rangle$ such that for all $j \in [m]$ we have ${p}'_j = \tfrac{q_j(p,X)}{Q_k(p,X)} \cdot \tilde{p}_k$ where $D_k$ is the component in the disutility graph containing chore $j$. So it just suffices to show that $\tilde{p}_k = \sum_{j \in B_k} p'_j$ for all $k \in [d]$. Now observe that, 
		\begin{align*}
		\sum_{j \in B_k} p'_j &= \sum_{j \in B_k} \frac{q_j(p,X)}{Q_k(p,X)} \cdot \tilde{p}_k\\
		&= \frac{\sum_{j \in B_k} q_j(p,X)}{Q_k(p,X)} \cdot \tilde{p}_k\\		
		&= \tilde{p}_k.  \qedhere
		\end{align*}
	\end{proof}
	
	Now we make an observation about the continuity of the functions $q(p,X)$ and $Q_k(p,X) = \sum_{j \in B_k} q_j(p,X)$ for all $k \in [d]$. Given $p$ and $X$, let $q(p,X)$ be the vector obtained as in~\eqref{qdefinition}. Observe that $q(p,X)$ is a continuous function in $p$ and $X$ for all $p \in P$ and $X \in \X$. As a result, even $Q_k(p,X)$  is a continuous function in $p$ and $X$ for all $k \in [d]$. We are now ready to show that $r^* \in \overline{P}(p^*,X^*)$.  
	
	The $r^*$ is the limit of the sequence $r^n$. Since for all $n$, $\langle r^n ,Y^n \rangle \in \phi( \langle p^n, X^n \rangle)$, we can conclude that each $r^n \in P$. Since the set $P$ is compact (and therefore closed), we have that $r^* \in P$ as well. Now, by Observation~\ref{necessarysufficient} it suffices to show that for each chore $j$ in component $D_k$, we have $r^*_j = \tfrac{q_j(p^*,X^*)}{Q_k(p^*,X^*)} \cdot \sum_{j' \in B_k} r^*_{j'}$. Again, since every $r^n \in \overline{P}(p^n,X^n)$, we have for each chore $j$ in component $D_k$,  $r^n_j = \tfrac{q_j(p^n,X^n)}{Q_k(p^n,X^n)} \cdot \sum_{j' \in B_k} r^n_{j'}$. Let $f_j(r,p,X) = r_j- \tfrac{q_j(p,X)}{Q_k(p,X)} \cdot \sum_{j' \in B_k} r_{j'}$. Observe that for all $j \in [m]$ the limit of the sequence $h^n_j = f_j(r^n,p^n,X^n)$, $h^*_j$ is zero (as $h^n_j$ is zero for all $n$). Again, since $q(p,X)$ and $Q_k(p,X)$ (for all $k \in [d]$) are continuous functions in $p \in P$ and $X \in \X$, and  $Q_k(p,X) > 0$  (for all $k \in [d]$) by Claim~\ref{Q_kpositive} for all $p \in P$ and $X \in \X$, we can conclude that $f_j(r,p,X)$ is well defined and also a continuous function in $p$, $X$ and $r$ for all $j \in [m]$. Thus the limit of the sequence $h^n_j = f_j(r^n,p^n,X^n)$ is $h^*_j = f_j(r^*,p^*,X^*) = r^*_j- \tfrac{q_j(p^*,X^*)}{Q_k(p^*,X^*)} \cdot \sum_{j' \in B_k} r^*_{j'} $. Since we know that $h^*_j =0$ for all $j \in [m]$, we can conclude that  for all $j \in [m]$, we have $r^*_j - \tfrac{q_j(p^*,X^*)}{Q_k(p^*,X^*)} \cdot \sum_{j \in B_k} r^*_j = 0$, implying $r^*_j = \tfrac{q_j(p^*,X^*)}{Q_k(p^*,X^*)} \cdot \sum_{j \in B_k} r^*_j$.

	\paragraph{Proving $Y^* \in \X^{p^*}$:} To show $Y^* \in X^{p^*}$, we need to show, for every agent $i$,
	\begin{enumerate}
		\item $Y^{*}_{ij} > 0$ only if $\tfrac{d(i,j')}{p^*_j} \leq \tfrac{d(i,j)}{p^*_{j'}}$ for all $j' \in [m]$, and  
		\item $\sum_{j \in [m]}w_{i,j} \cdot p^*_j = \sum_{j \in [m]} Y^*_{ij} \cdot p^*_j$.
	\end{enumerate}
	
	We first show part 1 by contradiction, Assume that there exists $i$, $j$ and $j'$ such that $Y^*_{ij} > \beta > 0$ and $\tfrac{d(i,j)}{p^*_j} > \tfrac{d(i,j')}{p^*_{j'}}(1 + \delta)$ for some $\delta > 0$. Since $p^* \in P$,  by Observation~\ref{nozeroprice} we have that there is at least one chore in each component with a non-zero price, and therefore, we can assume without loss of generality that $p^*_{j'} > 0$. Let $\varepsilon > 0$ be such that $\varepsilon \ll \mathit{min}(\beta, \tfrac{d(i,j) \cdot p^*_{j'} - (1 + \delta) \cdot d(i,j') \cdot p^*_j}{(1+ \delta) \cdot d(i,j') + d(i,j)})$. Such an $\varepsilon$ exists as $\tfrac{d(i,j) \cdot p^*_{j'} - (1+ \delta) \cdot d(i,j') \cdot p^*_j}{(1+ \delta) \cdot d(i,j') + d(i,j)} > 0$ for all values of $p^*_j$: when $p^*_j = 0$, then $\tfrac{d(i,j) \cdot p^*_{j'}}{(1+ \delta) \cdot d(i,j') + d(i,j)} > 0$ as $p^*_{j'} > 0$ and when $p^*_j > 0$, then $d(i,j) \cdot p^*_{j'} - (1+ \delta) \cdot d(i,j') \cdot p^*_j > 0$ as  $\tfrac{d(i,j)}{p^*_j} > (1+ \delta) \cdot \tfrac{d(i,j)}{p^*_{j'}}$. Since $Y^n$ converges to $Y^*$ and $p^n$ converges to $p^*$, we know that there exists an $n_0 \in \mathbb{N}$ be such that for $n > {n}_0$ we have $\abs{Y^*_{ij} - Y_{ij}} < \varepsilon$  and $\abs{p^*_j - p^n_j} < \varepsilon$ and $\abs{p^*_{j'} - p^n_{j'}} < \varepsilon$. It can be verified easily that for our choice of $\varepsilon$ we have $Y^n_{ij} > 0$, and $\tfrac{d(i,j)}{p^n_j} > \tfrac{d(i,j)}{p^n_{j'}}$, for all $n > n_0$,  which implies that $Y^n \notin \X^{p^n}$, which is a contradiction.
	
	We also prove part 2 by contradiction. Assume that $\sum_{j \in [m]}w_{i,j} \cdot p^*_j - \sum_{j \in [m]} Y^*_{ij} \cdot p^*_j = \delta$ for some non-zero $\delta$.  Since $Y^n$ converges to $Y^*$ and $p^n$ converges to $p^*$, we know that for every $\varepsilon > 0$, there exists an $n_0 \in \mathbb{N}$ be such that for $n > {n}_0$ we have $\abs{Y^*_{ij} - Y_{ij}} < \varepsilon$  and $\abs{p^*_j - p^n_j} < \varepsilon$ for all $j \in [m]$. Therefore, by choosing a sufficiently small $\varepsilon$ we can ensure that $\sum_{j \in [m]}w_{i,j} \cdot p^n_j - \sum_{j \in [m]} Y^n_{ij} \cdot p^n_j \neq 0$, for all $n > n_0$,  which would imply that $Y^n \notin \X^{p^n}$, which is a contradiction.
\end{proof}
		
We are now ready to state the main result of this section 
\begin{theorem}
	\label{existencethm}
	Every instance $I \in \mathcal{I}$ admits a competitive efficient allocation.
\end{theorem}

\begin{proof}
	We defined a correspondence $\phi$ that satisfies properties $\mathbf{P}_1$, $\mathbf{P}_2$, $\mathbf{P}_3$ and $\mathbf{P}_4$ by Lemmas~\ref{P1},~\ref{P2},~\ref{P3} and~\ref{P4}. By Lemma~\ref{existsfixedpoint} we have that any correspondence that satisfies the properties $\mathbf{P}_1$, $\mathbf{P}_2$, $\mathbf{P}_3$ and $\mathbf{P}_4$ has a fixed point.  Finally, by Lemma~\ref{fixedpointisequilibrium}, any fixed point of this correspondence will correspond to   competitive equilibrium in $I$. Therefore, our correspondence $\phi$ has at least one fixed point and this fixed point corresponds to  a competitive equilibrium.
\end{proof}

\paragraph{Proof of Fact~\ref{stochasticfixedpoint}:} Recall Fact~\ref{stochasticfixedpoint}.

\begin{fact*}
	Let $Z \in \mathbb{R}^{n \times n}$ be a square matrix such that $Z_{ij} \geq 0$ for all $j \neq i$ (all the non-diagonal entries of $Z$ are non-negative) and $\sum_{i \in [n]} Z_{ij} = 0$ (column sums are zero), then there exists a vector $t \in \mathbb{R}^n_{\geq 0}$ such that $\sum_{i \in [n]} t_i =1$  and $Z \cdot t = 0$.
\end{fact*}

\begin{proof}
Let $\lambda \gg \mathit{max}_{i,j \in [n]} (Z_{ij})$. Let $Z' = \tfrac{1}{\lambda} Z$. Observe that every $t$ that satisfies $Z' \cdot t = 0$, also satisfies $Z \cdot t = 0$ and vice versa. Also, each entry in the matrix  $Z'$ has absolute value is less than one. Let $Z'' = (Z' + I)$ where $I$ is the identity matrix. Note that every entry in the matrix $Z''$ is non-negative. Also every $t$ that satisfies $Z'' \cdot t = t$, also satisfies $Z' \cdot t = 0 $ and therefore also satisfies $Z \cdot t = 0$ and vice versa. From here on, we will be dealing with the following system of equations
\begin{align}
Z'' \cdot t = t\enspace .
\end{align}

We first observe that the matrix $Z''$ is column stochastic: For all $j \in [n]$, we have 
\begin{align*}
\sum_{i \in [n]}Z''_{ij} &= \sum_{i \in [n]}(\frac{1}{\lambda} \cdot Z_{ij} + I_{ij})\\ 
                       &=\sum_{i \in [n]}\frac{1}{\lambda} \cdot Z_{ij} + 1\\
                       &=0 + 1 &(\text{Column sums are zero in $Z$})\\
                       &=1\enspace .
\end{align*}
Now, let $R = \left\{ r \in \mathbb{R}^n_{\geq 0} \mid  \sum_{j \in [n]}r_j=1 \right\}$. Observe that the set $R$ is non-empty, convex and compact. We first make a small claim.
\begin{claim}
	Let $r' = Z'' \cdot r$. If $r \in R$ then $r' \in R$.
\end{claim}
\begin{proof}
	Since every entry in the matrix $Z''$ and every component of the vector $r$ is non-negative, we also have that every component of $r'$ is also non-negative: $r'_j \geq 0$ for all $j \in [d]$. Now observe that 
	\begin{align*}
	\sum_{j \in [n]} r'_j &= \mathbf{1}^T \cdot r'\\
	&= \mathbf{1}^T \cdot Z'' \cdot r\\
	&= \mathbf{1}^T \cdot r &(\text{as $Z''$ is column stochastic})\\
	&= \sum_{j \in [m]} r_j\\
	&=1. 
	\end{align*}
	Thus, $r' \in R$.
\end{proof}
We define $f:\R \rightarrow \R$ such that $f(r) = Z'' \cdot r$. Observe that $f$ is also continuous. Thus, by Brouwer's fixed point theorem there is a $t \in R$, such that $f(t) = t$ or equivalently $Z'' \cdot t = t$. 	
\end{proof}

\section{Hardness of Finding Equilibrium under Sufficiency Conditions}\label{ppadhardness}
In this section, we show that chore division may still be intractable even for the instances that satisfy Conditions 1 and 2 mentioned in Section~\ref{sufficiency}. We show that it is PPAD-hard to find a competitive equilibrium on instances that satisfy Conditions 1 and 2 in Section~\ref{sufficiency}. We will show that any polynomial time algorithm that determines a competitive equilibrium on instances that satisfy Conditions 1 and 2, will find an equilibrium in a \emph{normalized polymatrix game}. The normalized polymatrix game is known to be PPAD-hard~\cite{chen2017complexity}. Recall the normalized polymatrix game:

\begin{problem*}\textbf{(Normalized Polymatrix Game)~\cite{chen2017complexity}}\\
	\textbf{Given}: A $2n \times 2n$ rational matrix $\M$ with every entry in $[0,1]$ and $\M_{i,2j-1} + \M_{i,2j} = 1$ for all $i \in [2n]$ and $j \in [n]$ .\\
	\textbf{Find}: Equilibrium strategy vector $x \in \mathbb{R}^{2n}_{\geq 0}$ such that $x_{2i-1} + x_{2i} = 1$ and 
	\begin{align*}
	& x^T \cdot \M_{*,{2i-1}} > x^T \cdot \M_{*,2i} + \tfrac{1}{n} \implies x_{2i} = 0.\\
	& x^T \cdot \M_{*,{2i}} > x^T \cdot \M_{*,2i-1} + \tfrac{1}{n} \implies x_{2i-1} = 0.
	\end{align*}
	where $M_{*,k}$ represents the $k^{\mathit{th}}$ column of the matrix $\M$.	
\end{problem*}

From the next section onward, we elaborate our proof: We first introduce all agents and chores. Thereafter, we define the disutility matrix and endowment matrix and show that our instance satisfies the sufficiency conditions of Section~\ref{sufficiency}, and therefore admits a competitive equilibrium. Then, we show that our instance exhibits the four properties of \emph{pairwise equal endowments}, \emph{fixed earning}, \emph{price equality}, \emph{price regulation} and \emph{reverse ratio amplification} (as discussed in Section~\ref{mainres3}), and thus in polynomial-time we can construct the equilibrium strategy vector $x$ for $I$ from any competitive equilibrium in $E(I)$. The reader is highly encouraged to read Section~\ref{mainres3} before reading the elaborate version of the proof to get the idea of the overall proof sketch. 

\subsection{Agent and Chore Sets}
We define the set of $K = 2c \cdot \lceil \log(n) \rceil$ ($c =4$) many sets of chores (observe crucially that $K$ is even),
\begin{align*}
B_k &= \left\{\cup_{i \in [2n]} b^k_i \right\} &\text{for all $k \in [K]$}, 
\end{align*} 
and $K$ many sets of agents 
\begin{align*}
A_k &=\begin{cases} 
                \left\{\cup_{i \in [2n]} a^1_i\right\} \cup \left\{\cup_{i \in [2n]} a'_i \right\}  &\text{when $k=1$}, \\
				\left\{\cup_{i \in [2n]} a^k_i \right\}  \cup \left\{\cup_{i \in [n]} \A^k_i\right\} &\text{when $ 2 \leq k \leq K-1$}, \\
				\left\{\cup_{i,j \in [2n]}a^K_{i,j}\right\} \cup \left\{\cup_{i \in [n]} \A^K_i\right\}   &\text{when $k = K$}.
	\end{cases}
\end{align*} 

We remark that the sets $A_1$, $A_K$ of agents and sets $B_1$, $B_K$ of chores are to enforce the fixed earning, price equality and price regulation properties as mentioned in sketch of the reduction in Section~\ref{mainres3}, while the sets $A_k$ of agents and $B_k$ of chores for all $2 \leq k \leq K-1$ are to enforce reverse ratio amplification property as mentioned in Section~\ref{mainres3}.  We now define the disutility matrix and the endowment matrix of the instance.

\paragraph{Disutility Matrix and the Disutility Graph.} The disutility graph for our instance will be a disjoint union of  complete bipartite graphs and the entries in our disutility matrix will be to enforce price-regulation and reverse ratio-amplification properties. We now describe the disutility matrix: We define only the disutility values that are less than $\tau$ in the matrix (the disutility of all agent-chore pair not mentioned should be assumed to be at least $\tau$).  For all $k \in [K]$, for each pair of chores $b^k_{2i-1}$ and $b^k_{2i}$, there are a set of agents that have disutility less than $\tau$ towards them and have disutility larger or equal to $\tau$ towards all other chores; Additionally these agents also happen to be either in $A_k$ or $A_{k-1}$ (indices are modulo $K$). We now outline these agents and their disutilities for every $k \in [K]$. To define the disutility less than $\tau$, we introduce the scalars $\tfrac{1}{n^{3c}} = \alpha_1, \alpha_2, \dots , \alpha_K$ such that each $\alpha_{i+1} = \tfrac{3}{2} \cdot \alpha_i$ for all $i \in [K-1]$. Before we define the disutility matrix, we make an obvious claim about the scalars $\alpha_i$ for all $i \in [K]$, which will be useful later,
\begin{claim}
	\label{alpha-technical}
	We have $ n^c \cdot \alpha_1 < \alpha_K \leq \tfrac{1}{n^c}$.
\end{claim} 
\begin{proof}
	We first show the lower bound. We have $\alpha_K = (\tfrac{3}{2})^K \cdot \alpha_1 = (\tfrac{3}{2})^{2c \lceil \log (n) \rceil} \cdot \alpha_1 > 2^{c \log(n)} \cdot \alpha_1 = n^c \cdot \alpha_1$. Similarly, for the upper bound, we have, $\alpha_K = (\tfrac{3}{2})^K \cdot \alpha_1 = (\tfrac{3}{2})^{2c \lceil \log (n) \rceil } \cdot \alpha_1 < 2^{2c \log(n)} \cdot \alpha_1 = n^{2c} \cdot \alpha_1 \leq \tfrac{1}{n^{c}}$ (as $\alpha_1 = \tfrac{1}{n^{3c}}$). 
\end{proof}
We now define the disutility matrix:
\begin{itemize}
	\item $k =1$: For each $i \in [n]$, we first define the disutilities of the agents that have disutility less than $\tau$ for chores $b^1_{2i-1}$ and $b^1_{2i}$. For each $i \in [n]$ we have,
							\begin{align*}
								d(a^K_{i',2i-1},b^{1}_{2i-1}) &= (1- \alpha_1) &\text{and}&    & d(a^K_{i',2i-1},b^{1}_{2i}) &= (1+ \alpha_1) &\text{for all $i' \in [2n]$}\\
							    d(a^K_{i',2i},b^{1}_{2i-1}) &= (1+ \alpha_1)   &\text{and}&    & d(a^K_{i',2i},b^{1}_{2i}) &= (1 - \alpha_1) &\text{for all $i' \in [2n]$}\\
							    d(a'_{2i-1},b^{1}_{2i-1}) &= (1- \alpha_1) &\text{and}&    & d(a'_{2i-1},b^{1}_{2i}) &= (1 +  \alpha_1)\\
							    d(a'_{2i},b^{1}_{2i-1}) &= (1+ \alpha_1)   &\text{and}&    & d(a'_{2i},b^{1}_{2i}) &= (1 - \alpha_1).
							\end{align*}   
				   		Therefore, for each $i \in [2n]$,  we have a component $D_i^1$ in the disutility graph which is a complete bipartite graph comprising of agents $\left\{ \cup_{i' \in [2n]} a^K_{i',2i-1} \right\} \bigcup \left\{ \cup_{i' \in [2n]} a^K_{i',2i} \right\} \bigcup \left\{a'_{2i-1},a'_{2i} \right\}$ and chores $\left\{ b_{2i-1}^1, b_{2i}^1 \right\}$ (see Figure~\ref{disutility-graph} (left subfigure) for an illustration).
				   		
	\item $ 2 \leq k \leq K$: For each $i \in [n]$ we have,
				   \begin{align*}
					  d(a^{k-1}_{2i-1},b^{k}_{2i-1}) &= (1- \alpha_k) &\text{and}&    & d(a^{k-1}_{2i-1},b^{k}_{2i}) &= (1 +  \alpha_k)\\
					  d(a^{k-1}_{2i},b^{k}_{2i-1}) &= (1+ \alpha_k)   &\text{and}&    & d(a^{k-1}_{2i},b^{k}_{2i}) &= (1 - \alpha_k)\\
					  d(\A^k_{i},b^k_{2i-1}) &= (1- \alpha_k)         &\text{and}&    & d(\A^k_{i},b^k_{2i}) &= (1- \alpha_k) \enspace .\\
				  \end{align*}
				 Therefore, for every $k$ such that $2 \leq k \leq K$, for each $i \in [2n]$,  we have a connected component $D_i^k$ in the disutility graph which is a complete bipartite graph comprising of agents $\left\{ a^{k-1}_{2i-1}, a^{k-1}_{2i}, \A^k_i \right\}$ and chores $\left\{ b_{2i-1}^k , b_{2i}^k \right\}$ (see Figure~\ref{disutility-graph} (right subfigure) for an illustration).   
\end{itemize}
It is clear that the disutility graph is a disjoint union of complete bipartite graphs, namely, the union of $D^k_i$ for all $i \in [n]$ and $k \in [K]$. Therefore,
\begin{center}
  $E(I)$ satisfies Condition 1 in Section~\ref{sufficiency}.	
\end{center}	

\begin{figure}[htbp]
	\centering
	\begin{subfigure}[b]{0.4 \textwidth}
\begin{tikzpicture}[
roundnode/.style={circle, draw=green!60, fill=green!5, very thick, minimum size=7mm},
squarednode/.style={rectangle, draw=red!60, fill=red!5, very thick, minimum size=5mm},
]
\node[squarednode]      (a11)      at (0,8+2)                       {$\scriptstyle{a^K_{1(2i-1)}}$};
\node[squarednode]      (an1)      at (0,5+2)                        {$\scriptstyle{a^K_{n(2i-1)}}$};
\node[squarednode]      (a12)      at (0,4-2)                       {$\scriptstyle{a^K_{1(2i)}}$};
\node[squarednode]      (an2)      at (0,1-2)                        {$\scriptstyle{a^K_{n(2i)}}$};
\node[squarednode]      (a'1)      at (0,5.5)                       {$\scriptstyle{a'_{2i-1}}$};
\node[squarednode]      (a'2)      at (0,3.5)                        {$\scriptstyle{a'_{2i}}$};


\node[roundnode]      (b1)      at (4,5.5)                       {$\scriptstyle{b_{2i-1}^1}$};
\node[roundnode]      (b2)      at (4,3.5)                        {$\scriptstyle{b_{2i}^1}$};

\draw[blue,->] (a11)--(b1);
\draw[blue,->] (an1)--(b1);
\draw[blue,->] (a12)--(b2);
\draw[blue,->] (an2)--(b2);

\draw[blue,->,ultra thick]  (a11)--(b2);
\draw[blue,->, ultra thick] (an1)--(b2);
\draw[blue,->, ultra thick] (a12)--(b1);
\draw[blue,->, ultra thick] (an2)--(b1);

\draw[blue,->] (a'1)--(b1);
\draw[blue,->] (a'2)--(b2);
\draw[blue,->,ultra thick]  (a'1)--(b2);
\draw[blue,->, ultra thick] (a'2)--(b1);

\draw[black, very thick] (-0.75,4.55+2) rectangle (0.75,8.5+2);
\draw[black, very thick] (-0.75,0.5-2) rectangle (0.75,4.45-2);

\filldraw[color=black!60, fill=black!5, very thick](0,7+2) circle (0.02);
\filldraw[color=black!60, fill=black!5, very thick](0,6.5+2) circle (0.02);
\filldraw[color=black!60, fill=black!5, very thick](0,6+2) circle (0.02);

\filldraw[color=black!60, fill=black!5, very thick](0,3-2) circle (0.02);
\filldraw[color=black!60, fill=black!5, very thick](0,2.5-2) circle (0.02);
\filldraw[color=black!60, fill=black!5, very thick](0,2-2) circle (0.02);

\node at (2,-2.25) {$D^1_i$};
\end{tikzpicture}
	\end{subfigure}		
	\begin{subfigure}[b]{0.4 \textwidth}
\begin{tikzpicture}
\node[rectangle, draw=red!60, fill=red!5, very thick, minimum size=5mm]      (a11)      at (0,5)                       {$\scriptstyle{a^{k-1}_{2i-1}}$};
\node[rectangle, draw=red!60, fill=red!5, very thick, minimum size=5mm]      (a22)      at (0,3)                        {$\scriptstyle{a^{k-1}_{2i}}$};
\node[rectangle, draw=red!60, fill=red!5, very thick, minimum size=5mm]      (a3)      at (0,1)                       {$\scriptstyle{\A^k_{i}}$};
\node[circle, draw=green!60, fill=green!5, very thick, minimum size=7mm]      (b1)      at (4,5)                       {$\scriptstyle{b_{2i-1}^k}$};
\node[circle, draw=green!60, fill=green!5, very thick, minimum size=7mm]      (b2)      at (4,3)                        {$\scriptstyle{b_{2i}^k}$};
\draw[blue,->] (a11)--(b1);
\draw[blue,->] (a22)--(b2);
\draw[blue,->] (a3)--(b1);
\draw[blue,->] (a3)--(b2);
\draw[blue,->,ultra thick]  (a11)--(b2);
\draw[blue,->, ultra thick] (a22)--(b1);
\node at (2,0) {$D^k_i$};
\end{tikzpicture}
	\end{subfigure}	
	\caption{Illustration of the disutility graph corresponding to the disutility matrix: On the left we have the component $D_i^1$, and on the right we have $D_i^k$ when $2 \leq k \leq K$. The edges are colored in order to also encode the disutility matrix. The thin blue edges from agents to chores depict a disutility of $1 - \alpha_1$ for $D^1_i$ (left), and $1 - \alpha_k$ for $D^k_i$ when $2 \leq k \leq K$ (right). Similarly, the thick blue edges from agents to chores depict a disutility of $1 + \alpha_1$ for $D^1_i$ (left) and $1 + \alpha_k$ for $D^k_i$ (right).} 
	\label{disutility-graph}
\end{figure}
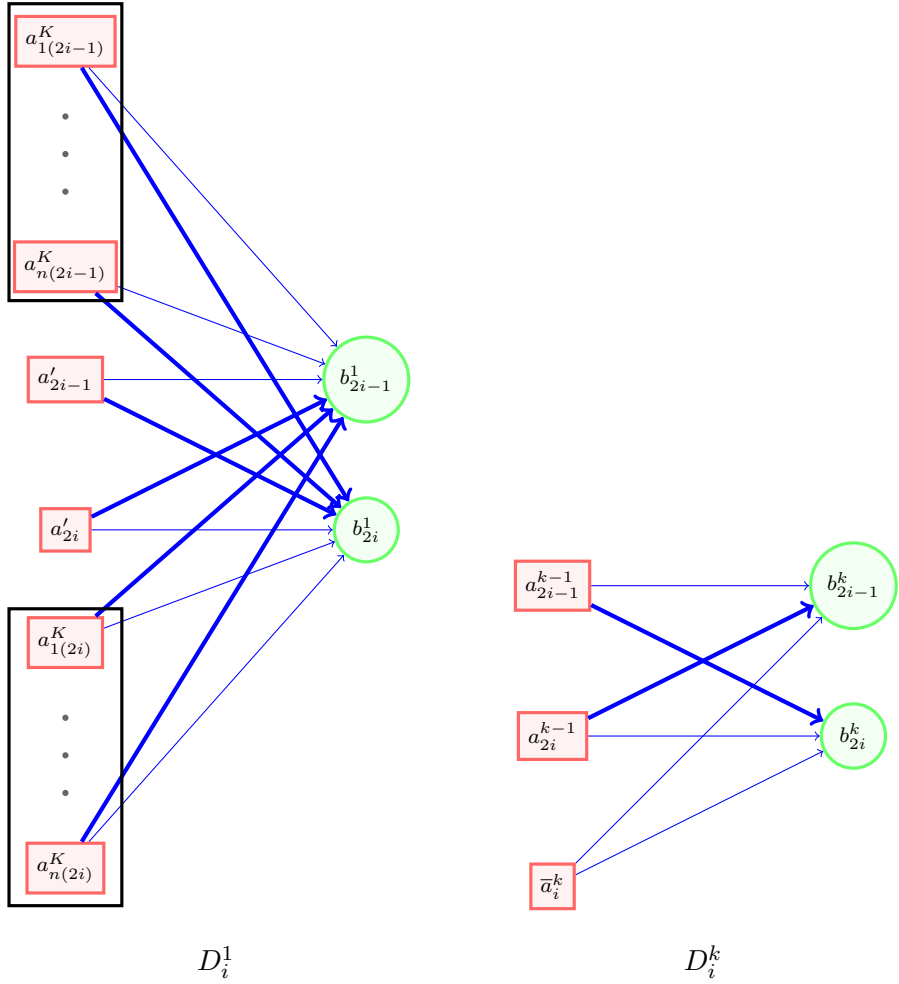

\paragraph{Endowment Matrix.} All agents in $A_k$ have endowments of chores only in $B_k$ for all $k \in [K]$. We only mention the non-zero agent-chore endowments (all agent-chore endowments, if not mentioned, are zero). 
\begin{itemize}
	\item $k=1$: For each $i \in [2n]$ we have,
	\begin{align*}
	w(a^1_i,b^1_i) &=n.
	\end{align*}
	Also, for each $i \in [n]$ we have 
	\begin{align*}
	w(a'_{2i-1},b^1_{2i-1})=w(a'_{2i-1},b^1_{2i})  &= \frac{1}{2} \cdot (1 - \alpha_K) \cdot (2n - \sum_{j \in [2n]} \M_{j,2i-1})\\
	w(a'_{2i},b^1_{2i-1})  = w(a'_{2i},b^1_{2i}) &= \frac{1}{2} \cdot  (1 - \alpha_K) \cdot (2n - \sum_{j \in [2n]} \M_{j,2i}).
	\end{align*}
	\item $2 \leq k \leq K-1$: For each $i \in [n]$, we have,
	\begin{align*}
	w(a^k_{2i-1},b^k_{2i-1}) &=n      &\text{and}&     & w(a^k_{2i},b^k_{2i})&=n\\
	w(\A^k_i, b^k_{2i-1}) &= \delta_k &\text{and}& & w(\A^k_i, b^k_{2i}) &= \delta_k\enspace ,
	\end{align*}
	where $\delta_k  = \tfrac{n \cdot \alpha_k}{2}$. The reason behind the exact choice of the value of $\delta_k$ will become explicit when we show that our instance satisfies the reverse ratio amplification property in Section~\ref{property-satisfication}. As of now, the reader is encouraged to think of it just as a small scalar.
	\item $k =K$: For each $i \in [n]$ we have,
	\begin{align*}
	w(a^K_{2i-1,j},b^K_{2i-1}) &=\M_{2i-1,j}  &\text{and}&   & w(a^K_{2i,j},b^K_{2i}) &=\M_{2i,j}& &\text{ for all $j \in [2n]$}\\
	w(\A^K_i, b^K_{2i-1}) &= \delta_K        &\text{and}&    & w(\A^K_i, b^K_{2i}) &= \delta_K\enspace ,
	\end{align*}
	where $\delta_K = \tfrac{n \cdot \alpha_K}{2}$ (the reason behind the choice of value will become explicit in Section~\ref{property-satisfication}).   	     
\end{itemize}
 
\paragraph{Exchange Graph.} We now construct the exchange graph of our instance and show that it is strongly connected. Observe that the disutility graph consists of connected components $D^k_i$ for  $k \in [K]$ and $i \in [n]$. Also observe that every component $D^k_i$ in the disutility graph comprises of exactly two chores $b^k_{2i-1}$ and $b^k_{2i}$. Therefore, to show that there exists an edge from component $D^{k'}_{i'}$ to $D^k_i$ in the exchange graph, it suffices to show that $D^{k'}_{i'}$ contains agents that own parts of chores $b^k_{2i-1}$ and $b^k_{2i}$. We now outline the edges in our exchange graph:
\begin{itemize}
	\item For all $i \in [n]$, and $2 \leq k \leq K$ there is an edge in the exchange graph from $D^k_i$ to $D^{k-1}_i$: $D^k_i$ contains the agents $a^{k-1}_{2i-1}$ and $a^{k-1}_{2i}$ that own parts of chores $b^{k-1}_{2i-1}$ and $b^{k-1}_{2i}$ respectively (see Figure~\ref{exchange-graph}).
	\item For all $i \in [n]$, there is an edge in the exchange graph from $D^1_i$ to $D^K_j$ for all $j \in [n]$: Consider any $j \in [n]$. Observe that the component $D^1_i$ contains the agents $a^K_{2j-1,2i}$ and $a^K_{2j,2i}$ and the agents $a^K_{2j-1,2i}$ and $a^K_{2j,2i}$ own parts of chores $b^K_{2j-1}$ and $b^K_{2j}$ respectively (see Figure~\ref{exchange-graph}).
\end{itemize}   

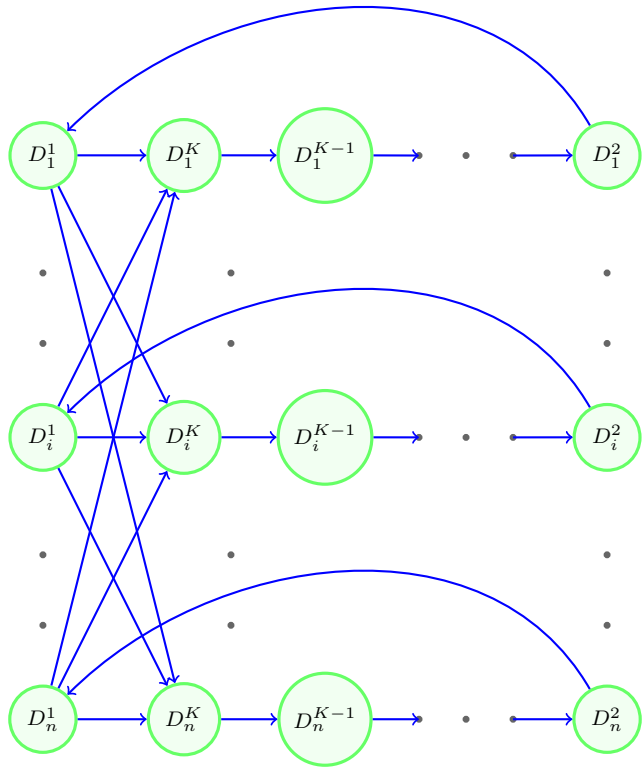
\begin{figure}
	\begin{center}
\begin{tikzpicture}[scale=1.25,
roundnode/.style={circle, draw=green!60, fill=green!5, very thick, minimum size=7mm},
squarednode/.style={rectangle, draw=red!60, fill=red!5, very thick, minimum size=5mm},
]
\node[roundnode]      (D11)      at (0,6)                       {$\scriptstyle{D^1_1}$};
\node[roundnode]      (Dk1)      at (1.5,6)                       {$\scriptstyle{D^{K}_1}$};
\node[roundnode]      (Dk'1)      at (3,6)                       {$\scriptstyle{D^{K-1}_1}$};
\node[roundnode]      (D21)      at (6,6)                       {$\scriptstyle{D^2_1}$};

\node[roundnode]      (D12)      at (0,3)                       {$\scriptstyle{D^1_i}$};
\node[roundnode]      (Dk2)      at (1.5,3)                       {$\scriptstyle{D^{K}_i}$};
\node[roundnode]      (Dk'2)      at (3,3)                       {$\scriptstyle{D^{K-1}_i}$};
\node[roundnode]      (D22)      at (6,3)                       {$\scriptstyle{D^2_i}$};

\node[roundnode]      (D13)      at (0,0)                       {$\scriptstyle{D^1_n}$};
\node[roundnode]      (Dk3)      at (1.5,0)                       {$\scriptstyle{D^{K}_n}$};
\node[roundnode]      (Dk'3)      at (3,0)                       {$\scriptstyle{D^{K-1}_n}$};
\node[roundnode]      (D23)      at (6,0)                       {$\scriptstyle{D^2_n}$};

\filldraw[color=black!60, fill=black!5, very thick](5,6) circle (0.02);
\filldraw[color=black!60, fill=black!5, very thick](4,6) circle (0.02);
\filldraw[color=black!60, fill=black!5, very thick](4.5,6) circle (0.02);

\filldraw[color=black!60, fill=black!5, very thick](5,3) circle (0.02);
\filldraw[color=black!60, fill=black!5, very thick](4,3) circle (0.02);
\filldraw[color=black!60, fill=black!5, very thick](4.5,3) circle (0.02);

\filldraw[color=black!60, fill=black!5, very thick](5,0) circle (0.02);
\filldraw[color=black!60, fill=black!5, very thick](4,0) circle (0.02);
\filldraw[color=black!60, fill=black!5, very thick](4.5,0) circle (0.02);

\filldraw[color=black!60, fill=black!5, very thick](0,4) circle (0.02);
\filldraw[color=black!60, fill=black!5, very thick](0,4.75) circle (0.02);

\filldraw[color=black!60, fill=black!5, very thick](2,4) circle (0.02);
\filldraw[color=black!60, fill=black!5, very thick](2,4.75) circle (0.02);

\filldraw[color=black!60, fill=black!5, very thick](6,4) circle (0.02);
\filldraw[color=black!60, fill=black!5, very thick](6,4.75) circle (0.02);

\filldraw[color=black!60, fill=black!5, very thick](0,4-3) circle (0.02);
\filldraw[color=black!60, fill=black!5, very thick](0,4.75-3) circle (0.02);

\filldraw[color=black!60, fill=black!5, very thick](2,4-3) circle (0.02);
\filldraw[color=black!60, fill=black!5, very thick](2,4.75-3) circle (0.02);

\filldraw[color=black!60, fill=black!5, very thick](6,4-3) circle (0.02);
\filldraw[color=black!60, fill=black!5, very thick](6,4.75-3) circle (0.02);

\draw[blue,->, thick] (D11)--(Dk1);
\draw[blue,->,thick] (Dk1)--(Dk'1);
\draw[blue,->,thick] (Dk'1)--(4,6);
\draw[blue,->, thick] (5,6)--(D21);

\draw[blue,->,thick] (D12)--(Dk2);
\draw[blue,->,thick] (Dk2)--(Dk'2);
\draw[blue,->,thick] (Dk'2)--(4,3);
\draw[blue,->,thick] (5,3)--(D22);

\draw[blue,->,thick] (D13)--(Dk3);
\draw[blue,->,thick] (Dk3)--(Dk'3);
\draw[blue,->,thick] (Dk'3)--(4,0);
\draw[blue,->,thick] (5,0)--(D23);


\draw[blue,->, thick] (D11)--(Dk2);
\draw[blue,->, thick] (D11)--(Dk3);
\draw[blue,->, thick] (D12)--(Dk1);
\draw[blue,->, thick] (D12)--(Dk3);
\draw[blue,->, thick] (D13)--(Dk2);
\draw[blue,->, thick] (D13)--(Dk1);

\path[blue,->,thick, out=120,in=45]    (D21) edge (D11);
\path[blue,->,thick, out=120,in=45]    (D22) edge (D12);
\path[blue,->,thick, out=120,in=45]    (D23) edge (D13);
\end{tikzpicture}
	\end{center}
	\caption{Illustration of the strong connectivity of the exchange graph. Observe that all nodes are reachable from any $D^1_i$ ($i \in [n]$). Also, from any arbitrary $D^{k'}_{i'}$, the node $D^{1}_{i'}$ is reachable and since every node is reachable from $D^1_{i'}$, every node is also reachable from $D^{k'}_{i'}$ as well. Therefore, the exchange graph is strongly connected.}
	\label{exchange-graph} 
\end{figure}

Observe that all nodes are reachable from any $D^1_i$ ($i \in [n]$). Also, from any arbitrary $D^{k'}_{i'}$, the node $D^{1}_{i'}$ is reachable and since every node is reachable from $D^1_{i'}$, every node is also reachable from $D^{k'}_{i'}$ as well. Therefore, the exchange graph is strongly connected. Therefore,

\begin{center}
	$E(I)$ satisfies Condition 2 in Section~\ref{sufficiency}.
\end{center}

Thus, $E(I)$ satisfies Conditions 1 and 2 in Section~\ref{sufficiency} and therefore admits a competitive equilibrium. Let $p(b^k_i)$ denote the price of any chore $b^k_i$ at a competitive equilibrium. We now prove that our instance satisfies the required properties of \emph{pairwise equal endowments}, \emph{price equality}, \emph{fixed earning}, \emph{price regulation} and \emph{reverse ratio amplification}.

\subsection{$E(I)$ Satisfies All the Properties}
\label{property-satisfication}

\paragraph{Pairwise Equal Endowments.} Here, we show that for all $i \in [n]$ and for all $k \in [K]$ the total endowment of $b^k_{2i-1}$ equals the total endowment of $b^k_{2i}$ and the total endowments of each chore in $E(I)$ is $\mathcal{O}(n)$.
\begin{lemma}\label{total-endowments}
	For all $i \in [2n]$, the total endowments of chores $b^k_{2i-1}$ and $b^{k}_{2i}$ is 
	\begin{enumerate}
		\item $n + n \cdot (1- \alpha_K)$, if $k=1$. In particular, $a'_{2i-1}$ and $a'_{2i}$ \emph{together}, own $n \cdot (1- \alpha_K)$ units of chores $b^k_{2i-1}$ and $b^{k}_{2i}$ each.  
		\item $n + \delta_k$, if $2 \leq k \leq K$. 
	\end{enumerate}
\end{lemma}

\begin{proof}
	When $k=1$, the only agents that have positive endowments of $b^1_{2i}$  are $a^1_{2i}$ (has an endowment of $n$ ) , $a'_{2i}$ (has an endowment of $\tfrac{1}{2} \cdot (1 - \alpha_K) \cdot (2n - \sum_{j \in [2n]} \M_{j,2i})$) and $a'_{2i-1}$(has an endowment of  $\tfrac{1}{2} \cdot (1 - \alpha_K) \cdot (2n - \sum_{j \in [2n]} \M_{j,2i-1})$). Therefore, the total endowment of $b^1_{2i}$ from the agents $a'_{2i}$ and $a'_{2i-1}$ is 
	\begin{align*}
	&= \frac{1}{2} \cdot (1 - \alpha_K) \cdot (2n - \sum_{j \in [2n]} \M_{j,2i}) + \frac{1}{2} \cdot (1 - \alpha_K) \cdot (2n - \sum_{j \in [2n]} \M_{j,2i-1})\\
	&= \frac{1}{2} \cdot (1 - \alpha_K) \cdot (4n - \sum_{j \in [2n]} (\M_{j,2i} + \M_{j,2i-1})).
	\end{align*} 
	Recall that $\M_{j,2i} + \M_{j,2i-1} = 1$. Therefore, the total endowment of $b^1_{2i}$ from the agents $a'_{2i}$ and $a'_{2i-1}$ is 
	\begin{align*}
	&= \frac{1}{2} \cdot (1 - \alpha_K) \cdot (4n - 2n)\\
	&= (1 - \alpha_K) \cdot n.
	\end{align*} 
	Therefore, the total endowment of chore $b^1_{2i}$ is $n + n \cdot (1- \alpha_K)$.
	A similar argument will show that the total endowment of chore $b^1_{2i-1}$ is also $n + n \cdot (1 - \alpha_K)$ and that agents $a'_{2i-1}$ and $a'_{2i}$ \emph{together}, own $n \cdot (1- \alpha_K)$ units of it. 
	
	When $2 \leq k \leq K-1$, the only agents that have positive endowments of $b^k_{2i}$ are $a^k_{2i}$ (has an endowment of $n$) and  $\A_i^k$ (has an endowment of $\delta_k$). Therefore, the total endowment is $n + \delta_k$. A similar argument will show that the total endowment of chore $b^k_{2i-1}$ is also $n + \delta_k$. 
	
	When $k =K$, the only agents that have positive endowments of $b^K_{2i}$ are the agents $a_{2i,j}^K$ (has an endowment of $\M_{2i,j}$) for all $j \in [2n]$ and the agent $\A^K_i$ (has an endowment of $\delta_K$). Therefore, the total endowment of chore $b^K_{2i}$ is
	\begin{align*}
	&=\sum_{j \in [2n]} \M_{2i,j} + \delta_K\\
	&=\sum_{j \in [n]} (\M_{2i,2j-1} + \M_{2i,2j}) + \delta_K\\
	&=\sum_{j \in [n]} 1 + \delta_K \\
	&=n + \delta_K
	\end{align*}
	A similar argument will show that the total endowment of chore $b^K_{2i-1}$ is also $n + \delta_K$.
\end{proof}

\paragraph{Price Equality.} Here we will show that the sum of prices of chores $b^1_{2i-1}$ and $b^1_{2i}$ equals that of $b^K_{2i-1}$ and $b^K_{2i}$. Let us define \[\pi^k_i = p(b^k_{2i-1}) + p(b^k_{2i}),\ \ \ \forall i \in [2n], k \in [K]\enspace .\] 
Since the prices corresponding to a competitive equilibrium is scale-invariant, we can assume without loss of generality that $\pi^1_1 = 2$. We now state the main lemma of price equality:

\begin{lemma}
	\label{price-equality}
	For all $i \in [n]$ and for all $k \in [K]$, we have $\pi^k_i = 2$. 
\end{lemma}

\begin{proof}
 We show this in two steps: First we show that $\pi^1_i = \pi^k_i$ for all $k \in [K]$. Then we show that $\pi^1_i = \tfrac{1}{n}\sum_{j \in [n]} \pi^K_j$ for all $i \in [n]$, implying that $\pi^1_i = \pi^1_j$ for all $i,j \in [n]$. Since for all $i \in [n]$ and $k \in [K]$, $\pi^k_i = \pi^1_i$ and $\pi^1_i = \pi^1_1$, we will have that $\pi^k_i = \pi^1_1 = 2$. 
 
 We first show $\pi^1_i = \pi^k_i$ for all $k \in [K]$: Consider any $k \in [2,K]$ and any $i \in [n]$. Observe that the agents $a^{k-1}_{2i-1}$, $a^{k-1}_{2i}$, $\A^k_i$ and chores $b^k_{2i-1}$, $b^{k}_{2i}$ form the connected component $D^k_i$ in the disutility graph. This implies that the agents $a^{k-1}_{2i-1}$, $a^{k-1}_{2i}$ and $\A^k_i$ earn all their money at a competitive equilibrium from chores $b^k_{2i-1}$ and  $b^{k}_{2i}$. Note that $\A^k_i$ owns $\delta_k$ units of both $b^k_{2i-1}$ and $b^k_{2i}$ only and has disutility less than $\tau$ also only for chores $b^k_{2i-1}$ and $b^k_{2i}$. Therefore, at a competitive equilibrium, $\A^k_i$ has to earn $\delta_k \cdot \pi^k_i$ money from chores $b^k_{2i-1}$ and $b^k_{2i}$. Thus, the total money the agents $a^{k-1}_{2i-1}$ and $a^{k-1}_{2i}$ earn from chores $b^k_{2i-1}$ and $b^k_{2i}$ is the total price of these chores remaining after $\A^k_i$ earns her share of $\delta_k \cdot \pi^k_i$, which is $(n + \delta_k) \cdot \pi^k_i - \delta_k \cdot \pi^k_i = n \cdot \pi^k_i$ (as the total endowment of each $b^k_{2i}$ and $b^k_{2i-1}$ is $n+ \delta_k$ for all $k \in [2,K]$ by Lemma~\ref{total-endowments}). At a competitive equilibrium, the total money earned by the agents $a^{k-1}_{2i-1}$ and $a^{k-1}_{2i}$ should be equal to the total prices of chores they own ($n$ units of $b^{k-1}_{2i-1}$ and $b^{k-1}_{2i}$ respectively). Thus we have, $n \cdot \pi^{k-1}_i = n \cdot \pi^k_i$. This, implies that, 
 \begin{align*}
  \pi^k_i = \pi^{k-1}_i = \dots = \pi^1_i.
 \end{align*}
 
 We now show that $\pi^1_i = \tfrac{1}{n}\sum_{j \in [n]} \pi^K_j$: This time we look into the connected component $D^1_i$ of the disutility graph. We can claim that the agents $\cup_{j \in [2n]} a^K_{j,2i-1}$, $\cup_{j \in [2n]} a^K_{j,2i}$ and the agents $a'_{2i-1}$ and $a'_{2i}$ earn all of their money at a competitive equilibrium from chores $b^1_{2i-1}$ and $b^1_{2i}$. Observe that both agents $a'_{2i-1}$ and $a'_{2i}$ own some units  of chores $b^1_{2i-1}$ and $b^1_{2i}$ only. Since the only chores towards which $a'_{2i-1}$ and $a'_{2i}$ have disutility less than $\tau$ are also $b^1_{2i-1}$ and $b^1_{2i}$, we can conclude that at a competitive equilibrium, to pay for their endowments, agents $a'_{2i-1}$ and $a'_{2i}$, together earn $n \cdot (1- \alpha_K) \cdot \pi^1_i$ amount of money from chores $b^1_{2i-1}$ and $b^1_{2i}$ (as from Lemma~\ref{total-endowments}, statement 1,  we have that $a'_{2i-1}$ and $a'_{2i}$ together own $n \cdot (1- \alpha_K)$ units of both chores $b^1_{2i-1}$ and $b^1_{2i}$). Thus, the total money agents $\cup_{j \in [2n]}a^K_{j,2i-1}$ and  $\cup_{j \in [2n]}a^K_{j,2i}$ earn at a competitive equilibrium is the total prices of chores $b^1_{2i-1}$ and $b^1_{2i}$ minus the total money earned by agents $a'_{2i-1}$ and $a'_{2i}$: $(n + n \cdot (1 - \alpha_K)) \cdot \pi^1_i - n \cdot (1 - \alpha_K) \cdot \pi^1_i = n \cdot \pi^1_i$. At a competitive equilibrium, the total money that these agents earn must equal the total prices of chores they own. Recall that each agent $a^{K}_{\ell,\ell'}$ owns $\M_{\ell,\ell'}$ units of $b^K_{\ell}$. Therefore, we have
 \begin{align*}
 n \cdot \pi^1_i &= \sum_{j \in [2n]} \M_{j,2i} \cdot p(b^K_j) +  \sum_{j \in [2n]} \M_{j,2i-1} \cdot p(b^K_j)\\
                 &= \sum_{j \in [2n]}(\M_{j,2i} + \M_{j,2i-1}) \cdot p(b^K_j)\\
                 &= \sum_{j \in [2n]} p(b^K_j) &\text{(using $\M_{j,2i-1} + \M_{j,2i} =1$)}\\
                 &=\sum_{j \in [n]} \pi^K_j 
 \end{align*}            	
 This implies that $\pi^1_i = \tfrac{1}{n}\sum_{j \in [n]} \pi^K_j$.
\end{proof}

\paragraph{Fixed Earning.} Here, we show that in every competitive equilibrium, the earning of each agent $a'_{i}$ for $i \in [2n]$ is fixed.

\begin{lemma}
	\label{fixed-earning}
	For all $i \in [2n]$, we have that the earning of agent $a'_i$ is $(1- \alpha_K) \cdot (2n - \sum_{j \in [2n]} \M_{j,i})$. 
\end{lemma}

\begin{proof}
	Let $i = 2i'$. Then agent $a'_{2i'}$ owns $\tfrac{1}{2} \cdot (1- \alpha_K) \cdot (2n - \sum_{j \in [2n]} \M_{j,2i'})$ units of both chores $b^1_{2i'-1}$ and $b^1_{2i'}$. Since the earning of any agent at a competitive equilibrium equals the sum of prices of chores she owns, we have that the earning of agent $2i'$ is  
	\begin{align*}
	&=\frac{1}{2} \cdot(1- \alpha_K) \cdot (2n - \sum_{j \in [2n]} \M_{j,2i'}) \cdot (p(b^1_{2i'-1}) + p(b^1_{2i'}))\\ 
	&=\frac{1}{2} \cdot(1- \alpha_K) \cdot (2n - \sum_{j \in [2n]} \M_{j,2i'}) \cdot \pi^1_{i'} \\
	&=\frac{1}{2} \cdot (1- \alpha_K) \cdot (2n - \sum_{j \in [2n]} \M_{j,2i'}) \cdot 2 &\text{(by Lemma~\ref{price-equality})}.\\
	&= (1- \alpha_K) \cdot (2n - \sum_{j \in [2n]} \M_{j,2i'}).
	\end{align*}
	Similarly, when $i = 2i'-1$ we can show that the total earning of agent $a'_{2i'-1}$ is $(1- \alpha_K) \cdot (2n - \sum_{j \in [2n]} \M_{j,2i'-1})$. Thus the total earning of any agent $a'_i$ in a competitive equilibrium is $(1- \alpha_K) \cdot (2n - \sum_{j \in [2n]} \M_{j,i})$. 
\end{proof}

\paragraph{Price Regulation.} Here, we show that for all $k \in [K]$ and $i \in [2n]$ the ratio of the prices of chores $b^k_{2i-1}$ and $b^k_{2i}$ is bounded.

\begin{lemma}
	\label{price-regulation}
	For all $k \in [K]$ and for all $i \in [n]$,  we have $\tfrac{1 - \alpha_k}{1 + \alpha_k} \leq \tfrac{p(b^k_{2i-1})}{p(b^k_{2i})} \leq \tfrac{1 + \alpha_k}{1 - \alpha_k}$.
\end{lemma}

\begin{proof}
	We prove that $\tfrac{1 - \alpha_k}{1 + \alpha_k} \leq \tfrac{p(b^k_{2i-1})}{p(b^k_{2i})}$ by contradiction. The proof for the other case is symmetric. So assume that $\tfrac{1 - \alpha_k}{1 + \alpha_k} > \tfrac{p(b^k_{2i-1})}{p(b^k_{2i})}$. In that case, none of the agents in the connected component $D^k_i$ will do any part of chore $b^k_{2i-1}$ (as the disutility to price ratio of $b^k_{2i-1}$ will be strictly more than that of $b^k_{2i}$). Since all the other agents have a disutility of $\tau$ for $b^k_{2i-1}$, it will remain unallocated. Therefore, the current prices for chores are not the prices corresponding to a competitive equilibrium, which is a contradiction. 
\end{proof}

\paragraph{Reverse Ratio Amplification.} Lastly, we show the property that when the price of chore $b^k_i$ is at a limit, then the price of chore $b^{k+1}_i$ is at the opposite limit, i.e., when $p(b^k_i) = 1 + \alpha_k$, then we have $p(b^{k+1}_i) = 1 - \alpha_{k+1}$ and similarly when $p(b^k_i) = 1 - \alpha_k$, then we have $p(b^{k+1}_i) = 1 + \alpha_{k+1}$.  
 
\begin{lemma}
	\label{rev-ratio-amplification1}
	 For all $1 \leq k < K$ and $i \in [n]$, we have that,
	 \begin{enumerate}
	 	\item if $\tfrac{p(b^{k}_{2i-1})}{p(b^k_{2i})} = \tfrac{1 - \alpha_k}{1+\alpha_k}$, then $\tfrac{p(b^{k+1}_{2i-1})}{p(b^{k+1}_{2i})} = \tfrac{1 + \alpha_{k+1}}{1-\alpha_{k+1}}$, and  
	 	\item if $\tfrac{p(b^{k}_{2i-1})}{p(b^k_{2i})} = \tfrac{1 + \alpha_k}{1 - \alpha_k}$, then $\tfrac{p(b^{k+1}_{2i-1})}{p(b^{k+1}_{2i})} = \tfrac{1 - \alpha_{k+1}}{1 + \alpha_{k+1}}$.
	 \end{enumerate}
\end{lemma}

\begin{proof}
	We just show the proof of part 1. The proof for part 2 is symmetric. Let us assume that $\tfrac{p(b^{k}_{2i-1})}{p(b^k_{2i})} = \tfrac{1 - \alpha_k}{1+\alpha_k}$. By Lemma~\ref{price-equality}, we have that $\pi^k_i = p(b^k_{2i-1}) + p(b^k_{2i}) = 2$. Therefore, $p(b^k_{2i-1}) = 1 - \alpha_k$ and $p(b^k_{2i}) = 1+ \alpha_k$. Observe that agent $a^k_{2i}$ owns $n$ units of chore $b^k_{2i}$ and has disutility less than $\tau$ only for chores $b^{k+1}_{2i-1}$ and $b^{k+1}_{2i}$ ($a^k_{2i}$ belongs in the connected component $D^{k+1}_i$). Since at a competitive equilibrium, the total earning of agent $a^k_{2i}$ equals the sum of prices of chores she owns, we have that $a^k_{2i}$ earns $n \cdot p(b^k_{2i}) =n(1 + \alpha_k)$ amount of money from chores $b^{k+1}_{2i-1}$ and $b^{k+1}_{2i}$. Note that it suffices to show that $a^k_{2i}$ earns positive amount of money from both chores $b^{k+1}_{2i-1}$ and $b^{k+1}_{2i}$ as this would imply that  $\tfrac{p(b^{k+1}_{2i-1})}{p(b^{k+1}_{2i})} = \tfrac{1 + \alpha_{k+1}}{1-\alpha_{k+1}}$. Therefore, for the rest of the proof, we show that $a^k_{2i}$ earns positive amount of money from both chores $b^{k+1}_{2i-1}$ and $b^{k+1}_{2i}$. By Lemma~\ref{price-regulation}, we have $\tfrac{p(b^{k+1}_{2i-1})}{p(b^{k+1}_{2i})} \leq \tfrac{1 + \alpha_{k+1}}{1-\alpha_{k+1}}$, implying that agent $a^k_{2i}$ always earns positive amount of money from chore $b^{k+1}_{2i}$. Thus, it only suffices to show that $a^k_{2i}$ earns positive amount of money from chore $b^k_{2i-1}$ as well. We prove this by contradiction. So let us assume that $a^k_{2i}$ earns money only from chore $b^k_{2i}$. We will now show that the current prices of chores are not the prices corresponding to a competitive equilibrium by distinguishing between two cases,
	\begin{itemize}
		\item \textbf{$p(b^{k+1}_{2i}) = 1 + x$ for some $x > 0$:} In this case, we have $p(b^{k+1}_{2i-1}) = 1 -x$ (as $\pi^{k+1}_i=2$) and therefore $p(b^{k+1}_{2i}) > p(b^{k+1}_{2i-1})$. Observe that in this case, agent $\A^{k+1}_i$ will only earn her entire money of $\delta_{k+1} \cdot (p(b^{k+1}_{2i-1}) + p(b^{k+1}_{2i-1})) = 2\delta_{k+1}$ from $b^{k+1}_{2i}$ (as the disutility to price ratio of $b^{k+1}_{2i}$ is strictly smaller than that of $b^{k+1}_{2i-1}$). Therefore, we have that the total money agents $a^k_{2i}$ and $\A^{k+1}_i$ earn from $b^{k+1}_{2i}$ is,
		\begin{align*}
		  &=2\delta_{k+1} + n \cdot (1 + \alpha_{k})\\
		  &=2\delta_{k+1} + n \cdot (1 + \frac{2}{3} \cdot \alpha_{k+1}) &(\text{as $\alpha_{k+1} = \frac{3}{2} \cdot \alpha_k$})\\
		  &= n \cdot (1 + \frac{2}{3} \cdot\alpha_{k+1} + \frac{2\delta_{k+1}}{n})\\
		  &= n \cdot (1 + \frac{2}{3} \cdot \alpha_{k+1} + \alpha_{k+1}) &(\text{as $\delta_{k+1} = \frac{n}{2} \cdot \alpha_{k+1}$})\\
		  &>n \cdot (1 + \frac{3}{2} \cdot \alpha_{k+1} + \frac{\alpha_{k+1}^2}{2}) &(\text{as $\alpha_{k+1} \ll \frac{1}{3}$ by Claim~\ref{alpha-technical}})\\
		  &=n \cdot (1 + \frac{\alpha_{k+1}}{2}) \cdot (1 + \alpha_{k+1})\\
		  &=(n + \delta_{k+1}) \cdot (1 + \alpha_{k+1}) &(\text{as $\delta_{k+1} = \frac{n}{2} \cdot \alpha_{k+1}$}),
		\end{align*} 
		which is a contradiction, as the total price of $b^{k+1}_{2i}$ is at most $(n + \delta_{k+1}) \cdot (1 + \alpha_{k+1})$ (there is a total endowment of $n + \delta_{k+1}$ for chore $b^{k+1}_{2i}$ by Lemma~\ref{total-endowments}, and $p(b^{k+1}_{2i}) \leq 1+ \alpha_{k+1}$).
		\item \textbf{$p(b^{k+1}_{2i}) = 1 - x$ for $0 \leq x < \alpha_{k+1}$:} Since $x < \alpha_{k+1}$, agent $a^k_{2i}$ will earn her entire money only from chore $b^{k+1}_{2i}$ as the disutility to price ratio of chore $b^{k+1}_{2i}$ is still less than that of chore $b^{k+1}_{2i-1}$. Since the total endowment of $b^{k+1}_{2i}$ is $n+ \delta_{k+1}$ by Lemma~\ref{total-endowments} and $p(b^{k+1}_{2i}) = 1-x$, the total price of chore $b^{k+1}_{2i}$ is, 
		\begin{align*}
		 &=(n + \delta_{k+1}) \cdot (1-x)\\
		 &\leq (n+ \delta_{k+1})\\
		 &<(n + \frac{4}{3}\delta_{k+1}) \\
		 &=n \cdot (1 + \frac{4 \delta_{k+1}}{3n})\\
		 &=n \cdot (1 + \frac{2 \alpha_{k+1}}{3}) &(\text{as $\delta_{k+1} = \frac{n}{2} \cdot \alpha_{k+1}$})\\
		 &=n \cdot (1 + \alpha_k) &(\text{as $\alpha_{k+1} = \frac{3}{2} \cdot \alpha_k$}),
		\end{align*}   
		which is the total money that agent $a^{k}_{2i}$ earns from $b^{k+1}_{2i}$, which is a contradiction. \qedhere
	\end{itemize} 
\end{proof}

Since $K$ is even, a repeated application of Lemma~\ref{rev-ratio-amplification1} will yield the following lemma, 

\begin{lemma}
	\label{rev-ratio-amplification2}
	We have,
	\begin{enumerate}
		\item if $\tfrac{p(b^{1}_{2i-1})}{p(b^1_{2i})} = \tfrac{1 - \alpha_1}{1+\alpha_1}$, then $\tfrac{p(b^{K}_{2i-1})}{p(b^{K}_{2i})} = \tfrac{1 + \alpha_{K}}{1-\alpha_{K}}$, and  
		\item if $\tfrac{p(b^{1}_{2i-1})}{p(b^1_{2i})} = \tfrac{1 + \alpha_1}{1 - \alpha_1}$, then $\tfrac{p(b^{K}_{2i-1})}{p(b^{K}_{2i})} = \tfrac{1 - \alpha_{K}}{1 + \alpha_{K}}$.
	\end{enumerate}	
\end{lemma}

Now that we have shown that our instance satisfies the desired properties of {price equality}, {fixed earning}, {price regulation} and {reverse ratio amplification}, we are ready to outline how to determine the equilibrium strategy vector $x$ for the instance $I$ of the polymatrix game, given the competitive equilibrium prices of the instance $E(I)$ of chore division:
\begin{align*}
x_i &= \frac{p(b^K_{i}) - (1-\alpha_K)}{2 \cdot \alpha_K}
\end{align*}

It is clear that given the prices of chores at a competitive equilibrium,  the equilibrium strategy  vector can be obtained in linear time. We will now show that $x$ is the desired equilibrium strategy vector for instance $I$ of the polymatrix game. 

\begin{lemma}
	\label{polymatrix-reduction}
	$x = \langle x_1,x_2, \dots, x_{2n} \rangle $ is an equilibrium strategy vector for the polymatrix game instance $I$.
\end{lemma}

\begin{proof}
First, observe that since our instance satisfies the price equality (Lemma~\ref{price-equality}) and price regulation (Lemma~\ref{price-regulation}) we have that for all $i \in [2n]$, $1 -\alpha_K \leq p(b^K_i) \leq 1 + \alpha_K$. Therefore, for all $i \in [2n]$   $x_i \geq 0$ . Furthermore, for all $i \in [n]$ we have $x_{2i-1} +x_{2i} =  \tfrac{p(b^K_{2i-1})  + p(b^K_{2i}) - 2(1-\alpha_K)}{2 \cdot \alpha_K} = \tfrac{2\alpha_K}{2\alpha_K} = 1$ (as our instance satisfies price equality: by Lemma~\ref{price-equality} we have $p(b^K_{2i-1}) + p(b^K_{2i})=2$). Now we will show that if $x^T \cdot \M_{*,{2i}} > x^T \cdot \M_{*,2i-1} + \tfrac{1}{n}$, then $ x_{2i-1} = 0$. The proof for the other symmetric condition will be similar. So let us assume that  $x^T \cdot \M_{*,{2i}} > x^T \cdot \M_{*,2i-1} + \tfrac{1}{n}$. Observe that the agents that have a disutility of $1- \alpha_1$ towards chore $b^1_{2i}$ are $\left\{ \cup_{j \in [2n]} a^K_{j,2i} \right\} \cup a'_{2i}$. Observe that at a competitive equilibrium, the total earning of the agents $\left\{ \cup_{j \in [2n]} a^K_{j,2i} \right\} \cup a'_{2i}$ equals the sum of prices of chores they own, which is,

\begin{align*}
&=\sum_{j \in [2n]} \M_{j,2i} \cdot p(b^K_j) + (1- \alpha_K) \cdot (2n - \sum_{j \in [2n]} \M_{j,2i}) &(\text{by Lemma~\ref*{fixed-earning}})\\
&=\sum_{j \in [2n]} \M_{j,2i} \cdot (2\alpha_K \cdot x_{j} + (1-\alpha_K)) +  (1- \alpha_K) \cdot (2n - \sum_{j \in [2n]} \M_{j,2i}) &\text{(substituting $p(b^K_j)$)}\\ 
&=\sum_{j \in [2n]} 2\alpha_K \cdot x_{j} \cdot \M_{j,2i} + (1-\alpha_K) \cdot\sum_{j \in [2n]}  \M_{j,2i} +  (1- \alpha_K) \cdot (2n - \sum_{j \in [2n]} \M_{j,2i})\\          
&= 2\alpha_K x^T \cdot \M_{*,2i} +  2n \cdot (1-\alpha_K).               
\end{align*}

Similarly, the total earning of the agents that have a disutility of $1 - \alpha_1$ towards $b^1_{2i-1}$ is  $2 \alpha_K x^T \cdot \M_{*,2i-1} +  2n \cdot (1-\alpha_K)$. Observe that the agents with disutility $1-\alpha_1$ towards $b^1_{2i}$ can earn all their money only from either $b^1_{2i}$ or $b^1_{2i-1}$ (as these are the only chores towards which they have disutility less than $\tau$). Also note that both chores $b^1_{2i-1}$ and $b^1_{2i}$ have the same total endowment which is $n + n \cdot (1 - \alpha_K)$ by Lemma~\ref{total-endowments}(part 1). Now if, the agents with disutility $1-\alpha_1$ towards $b^1_{2i}$ earn their money entirely from $b^1_{2i}$, then we will have $p(b^1_{2i}) \geq \tfrac{2\alpha_K x^T \cdot \M_{*,2i} +  2n \cdot (1-\alpha_K)}{n + n\cdot (1- \alpha_K)}$ and $p(b^1_{2i-1}) \leq \tfrac{2\alpha_K x^T \cdot \M_{*,2i-1} +  2n \cdot (1-\alpha_K)}{n + n\cdot (1- \alpha_K)}$. Since,  $x^T \cdot \M_{*,{2i}} > x^T \cdot \M_{*,2i-1} + \tfrac{1}{n}$. we have $p(b^1_{2i}) > p(b^1_{2i-1}) +  \tfrac{1}{n} \cdot \tfrac{2\alpha_K}{n + n \cdot (1- \alpha_K)} > p(b^1_{2i-1}) + \tfrac{\alpha_K}{n^2}$. Again, since $\tfrac{\alpha_K}{n^2} \gg \alpha_1$ (by Claim~\ref{alpha-technical}), we have that $\tfrac{p(b^1_{2i})}{p(b^1_{2i-1})} > \tfrac{1 + \alpha_1} {1 - \alpha_1}$, which is a contradiction as our instance satisfies price-regulation property (by Lemma~\ref{price-regulation}). Therefore, the agents that have a disutility of $1-\alpha_1$ towards $b^1_{2i}$ should also earn their money from $b^1_{2i-1}$. But this is only possible if $\tfrac{p(b^1_{2i})}{p(b^1_{2i-1})} = \tfrac{1 - \alpha_1}{1 + \alpha_1}$. Since our instance also satisfies the reverse ratio amplification, by Lemma~\ref{rev-ratio-amplification2} we have that $\tfrac{p(b^K_{2i})}{p(b^K_{2i-1})} = \tfrac{1 + \alpha_K}{1 - \alpha_K}$. Since $p(b^K_{2i}) + p(b^K_{2i-1}) = 2$ by price equality property (Lemma~\ref{price-equality}), we have that $p(b^K_{2i-1}) = 1 - \alpha_K$. Therefore, we have
\begin{align*}
x_{2i-1} &=\frac{(1- \alpha_K) - (1-\alpha_K)}{2 \cdot \alpha_K}\\
&=0.
\end{align*}

A very similar argument will show that when $x^T \cdot \M_{*,{2i-1}} > x^T \cdot \M_{*,2i} + \tfrac{1}{n}$, then $ x_{2i} = 0$.  
\end{proof}

Thus, this immediately implies the main result of this section.
\begin{theorem}
	\label{PPADhard}
	Let $\mathcal{I}$ be the set of all instances that satisfy Conditions 1 and 2 in Section~\ref{sufficiency}. Chore division is PPAD-hard even when restricted to the set of instances $\mathcal{I}$.
\end{theorem}

\begin{proof}
	We bring all the points together. Normalized polymatrix game is PPAD-hard~\cite{chen2017complexity}. Given an instance $I$ of the normalized polymatrix game, in polynomial time we can determine the instance $E(I)$. $E(I)$ satisfies the sufficiency conditions mentioned in Section~\ref{sufficiency} and therefore admits a competitive equilibrium. Given the equilibrium prices for $E(I)$, in polynomial time we can determine the equilibrium strategy vector for the polymatrix game. Therefore, chore division is PPAD-hard even on instances that satisfy the sufficiency conditions in Section~\ref{sufficiency}.
\end{proof}


\newcommand{\etalchar}[1]{$^{#1}$}

\end{document}